\documentclass{article}

\usepackage[a4paper, total={6in, 8in}]{geometry}

\usepackage{amsmath,amssymb,amsthm,amsfonts}
\usepackage{dsfont}
\usepackage[affil-it]{authblk}

\usepackage{enumitem}
\usepackage{braket}
\usepackage{mathtools}

\usepackage[numbers]{natbib}
\usepackage{xcolor}

\usepackage{fullpage}

\usepackage{algorithm}
\usepackage{float}
\floatstyle{ruled}
\restylefloat{algorithm}
\usepackage{xurl}
\usepackage[colorlinks=true, linkcolor=blue,
citecolor=blue, urlcolor=blue]{hyperref}
\hypersetup{pdftitle={Emulation Capacity between Idempotent Channels}, pdfauthor={Idris Delsol, Omar Fawzi, Li Gao, Mizanur Rahaman}, pdflang={en-GB}}

\DeclareMathOperator{\supp}{supp}
\newcommand{\eps}{\varepsilon}
\DeclareMathOperator{\tr}{tr}

\newcommand{\la}{\lambda}

\newcommand{\R}{{\mathbb R}}
\newcommand{\N}{\mathbb{N}}
\newcommand{\C}{{\mathbb C}}
\newcommand{\Q}{\mathbb{Q}}

\newcommand{\rmR}{\mathrm{R}}
\newcommand{\EA}{\mathrm{EA}}
\newcommand{\SR}{\mathrm{SR}}
\newcommand{\UA}{\mathrm{UA}}

\newcommand{\A}{\mathcal{A}}
\newcommand{\B}{\mathcal{B}}
\newcommand{\D}{\mathcal{D}}
\newcommand{\E}{\mathcal{E}}
\newcommand{\F}{\mathcal{F}}
\newcommand{\cG}{\mathcal{G}}
\newcommand{\cH}{\mathcal{H}}
\newcommand{\cL}{\mathcal{L}}
\newcommand{\cN}{\mathcal{N}}
\newcommand{\cM}{\mathcal{M}}

\newcommand{\T}{\mathcal{T}}

\newcommand{\CEA}[2]{C_{\mathrm{EA}}(#1 \mapsto #2)}
\newcommand{\CSR}[2]{C_{\mathrm{SR}}(#1 \mapsto #2)}
\newcommand{\CUA}[2]{C_{\mathrm{UA}}(#1 \mapsto #2)}
\newcommand{\CR}[2]{C_{\mathrm{R}}(#1 \mapsto #2)}

\newcommand{\deltaEA}[4]{\delta_{#1,#2}^{\mathrm{EA}}(#3 \mapsto #4)}
\newcommand{\deltaSR}[4]{\delta_{#1,#2}^{\mathrm{SR}}(#3 \mapsto #4)}
\newcommand{\deltaUA}[4]{\delta_{#1,#2}^{\mathrm{UA}}(#3 \mapsto #4)}
\newcommand{\deltaR}[4]{\delta_{#1,#2}^{\rmR}(#3 \mapsto #4)}

\newcommand{\proj}[1]{|#1\rangle\!\langle#1|}
\newcommand{\ketbra}[2]{|#1\rangle\!\langle#2|}

\newcommand{\inflognormratio}[2]{%
  \inf_{p \in [1,+\infty]} \frac{\log(\|\lambda(#1)\|_p)}{\log(\|\lambda(#2)\|_p)}%
}

\newcommand{\infnormratio}[4]{%
    \inf_{p \in [1,+\infty]} \frac{\|\lambda(#1)\|_p^{#3}}{\|\lambda(#2)\|_p^{#4}}%
}

\newcommand{\Rg}{\operatorname{Rg}}

\newcommand{\Id}{\mathrm{Id}}

\theoremstyle{plain}
\newtheorem{theorem}{Theorem}
\newtheorem{prop}{Proposition}  %Somebody cited our Prop. 1 in their paper, so aligning its numbering convention with that of  theorems would break their reference.
\newtheorem{lemma}[theorem]{Lemma}
\newtheorem{corollary}[theorem]{Corollary}
\newtheorem{mainthm}{Theorem}

\theoremstyle{definition}
\newtheorem{definition}[theorem]{Definition}
\newtheorem{example}[theorem]{Example}

\theoremstyle{remark}
\newtheorem{remark}[theorem]{Remark}

\begin{document}

\title{Emulation Capacity between Idempotent Channels\thanks{A short version of this work was presented at ISIT 2026~\cite{self.short}. The present version additionally treats the shared randomness and entanglement-assisted regimes, and establishes a strong converse that matches the unassisted capacity $\CUA{\cG}{\F}$ for the full range $p \in [1,+\infty]$, whereas the converse of the short version was restricted to $p \in \{1,+\infty\}$.}}

\author[1]{Idris Delsol}
\author[1]{Omar Fawzi}
\author[2,3]{Li Gao}
\author[1,4,5]{Mizanur Rahaman}

\date{}

\affil[1]{Inria, ENS de Lyon, UCBL, LIP, Lyon, France}
\affil[2]{School of Mathematics and Statistics, Wuhan University, Wuhan, China}
\affil[3]{Wuhan Institute of Quantum Technology, Wuhan, China}
\affil[4] {Department of Mathematical Sciences, Chalmers University of Technology, Gothenburg, Sweden}
\affil[5]{Wallenberg Centre for Quantum Technology, Chalmers University of Technology, Gothenburg, Sweden}

\maketitle
\begin{abstract}
We study the optimal rates of emulation (also called interconversion) between quantum channels, in three regimes: when the encoder and the decoder are unassisted, when they share classical randomness, and when they share an unlimited amount of entanglement.
When the source and the target channels are idempotent, we completely characterise these three regimes, in terms of the shape vectors of the two channels, which describe the structure of their ranges.
In the unassisted and shared randomness regimes, we give a single-letter expression for the emulation capacity, show that it coincides with the zero-error emulation capacity, and prove a finite-blocklength converse which implies a strong converse.
In these two regimes, shared randomness does not increase the emulation capacity, and the emulation of idempotent channels is not reversible.
Moreover, we provide an efficient algorithm to approximate this capacity to arbitrary precision.
When an unlimited amount of shared entanglement is allowed, we determine the minimal emulation error exactly at every finite blocklength.
It is governed by a single quantity, the $\ell_2$-norm of the shape vector, so that the emulation becomes reversible, and exact emulation is achievable with finitely many maximally entangled states.
Altogether, these results draw a comprehensive picture of the task of emulating idempotent channels with each other.
\end{abstract}

\tableofcontents

\section{Introduction}

Given target and source quantum channels $\F$ and $\cG$ (respectively), a central problem in quantum information theory is to understand the optimal interconversion rates $k/n$ such that one can use $n$ copies of $\cG$, combined with an encoder and a decoder channel, to emulate the action of $k$ copies of $\F$.
When $\F$ is a noiseless channel, this is the problem of channel coding for $\cG$, which is at the heart of Shannon theory~\cite{Shannon48} and of its quantum generalisations~\cite{Holevo98,Schumacher97}.
When $\cG$ is a noiseless channel, this is the problem of channel simulation for $\F$, which is also well-studied in both classical and quantum Shannon theory, where it gives rise to reverse Shannon theorems~\cite{Bennett02,Bennett14,Berta18,Cao24}.
More generally, the interconversion of channels is a natural question in resource theories of quantum channels~\cite{Liu20,BenDana17,Diaz18}.

Determining such optimal rates is notoriously difficult.
Already for channel coding, the capacities of a quantum channel are in general given by regularised formulas, which involve an optimisation over an unbounded number of uses of the channel, and superadditivity phenomena show that this regularisation cannot be dispensed with in general, see e.g.~\cite{leditzky2023platypus} and the references therein.
Even for classical channels, no single-letter formula is known for the zero-error capacity, and its computability is a delicate question~\cite{BocheDeppe21}.
The interconversion of two general channels is a fortiori harder: it contains both channel coding and channel simulation as special cases, and neither the source nor the target channel is noiseless.
It is therefore natural to look for families of channels for which the interconversion problem can be solved completely.

In this article, we solve this problem when both $\F$ and $\cG$ are \emph{idempotent}, i.e.\ when they satisfy $\F \circ \F = \F$ and $\cG \circ \cG = \cG$.
Examples of idempotent channels include the identity channel, the completely depolarising channel, the completely dephasing channel and channels on bipartite systems $\cL(\cH_A \otimes \cH_B)$ of the form $x \mapsto \tr_B(x) \otimes \rho_B$ with $\rho_B$ a state in $\cL(\cH_B)$.
Idempotent channels naturally appear as the limits of iterated noise: for an arbitrary channel $\Phi$, a subsequence of the sequential compositions $\Phi^{n} = \Phi \circ \cdots \circ \Phi$ converges to an idempotent channel~\cite{kuperberg,Wolf12}, which describes the information that survives an arbitrarily long Markovian evolution.
The capacities of quantum Markovian noise for large times were studied from this perspective in~\cite{Fawzi24,Singh.2024,Singh25}.
As we will see, all the properties of an idempotent channel $\F$ relevant for interconversion are captured by an integer vector $\la(\F)$, called its \emph{shape vector}.
The range of $\F$ decomposes as a direct sum of blocks of the form $\cL(\cH_{k,1}) \otimes \rho_k$, with $\rho_k$ a fixed state, and $\la(\F)$ is the vector of the dimensions of the spaces $\cH_{k,1}$, sorted in non-increasing order (see Proposition~\ref{prop:structure_idempotent_channel} for the precise statement).
For instance, $\la(\Id_d) = (d)$ and $\la(\Delta_d) = (1,\dots,1) \in \N^d$, with $\Id_d$ and $\Delta_d$ the identity and the completely dephasing channels on $\cL(\C^d)$.
Shape vectors were introduced by Kuperberg~\cite{kuperberg} to compare hybrid quantum memories, and, given an idempotent channel $\F : \cL(\cH) \rightarrow \cL(\cH)$, $\la(\F)$ can be computed in time polynomial in the dimension of $\cH$~\cite{Fawzi24}.
We write $\|u\|_p$ for the $\ell_p$-norm of a vector $u$ (see eq.~\eqref{eq:Lp_norm}).

We consider three regimes, according to the resources that the encoder and the decoder share: in the \emph{unassisted} regime ($\UA$), they share nothing; in the \emph{shared randomness} regime ($\SR$), they can use a common random variable to choose the pair of encoding and decoding channels that they apply; in the \emph{entanglement-assisted} regime ($\EA$), they share an arbitrary entangled state, of arbitrary finite dimension.
In each regime $\rmR \in \{\UA,\SR,\EA\}$ and for $k,n \in \N_{>0}$, we denote by $\deltaR{k}{n}{\cG}{\F}$ the minimal error with which $k$ copies of $\F$ can be approximated by a channel built from $n$ uses of $\cG$, preceded by an encoding and followed by a decoding channel, possibly with the help of the shared resource allowed in the regime, the error being measured in normalised diamond distance.
In each regime $\rmR$, the \emph{emulation capacity} $\CR{\cG}{\F}$ is the largest rate $r$ such that, along some sequence of values of $n$ going to infinity, $\lceil rn \rceil$ uses of $\F$ can be emulated by $n$ uses of $\cG$ with an error $\deltaR{\lceil rn \rceil}{n}{\cG}{\F}$ tending to zero.
We say that an emulation is \emph{exact}, or has \emph{zero error}, when the channel it implements is exactly $\F^{\otimes k}$.
The \emph{zero-error emulation capacity} is the supremum of the ratios $k/n$ over all pairs $(k,n)$ of positive integers for which $\F^{\otimes k}$ can be emulated exactly by $\cG^{\otimes n}$.
The zero-error emulation capacity is at most the emulation capacity, and the two quantities can differ for general channels.
The precise definitions are given in Section~\ref{subsec:emulation}.

Our first main result gives a complete description of the unassisted and shared randomness regimes for idempotent channels.

\begin{mainthm}[Unassisted and shared randomness regimes]\label{thm:main_UA_SR}
Let $\F$ and $\cG$ be two idempotent channels with $\la(\F) \neq (1)$.
\begin{enumerate}[label=(\roman*)]
    \item \emph{Capacity.} The unassisted and the shared randomness emulation capacities of $\F$ by $\cG$ coincide and are given by the single-letter formula
    \begin{equation}\label{eq:emul_capa}
        \CUA{\cG}{\F} = \CSR{\cG}{\F} = \inf_{p \in [1,+\infty]} \frac{\log(\|\la(\cG)\|_p)}{\log(\|\la(\F)\|_p)},
    \end{equation}
    where the term $p = +\infty$ is omitted when $\|\la(\F)\|_\infty = 1$. Moreover, in both regimes, this capacity coincides with the zero-error emulation capacity.
    \item \emph{Finite blocklength.} For all $k,n \in \N_{>0}$,
    \begin{equation}\label{eq:bound_error_approx_case}
        \deltaUA{k}{n}{\cG}{\F} \geq \deltaSR{k}{n}{\cG}{\F} \geq 1 - \inf_{p \in [1,+\infty]}\frac{\|\la(\cG)\|_p^n}{\|\la(\F)\|_p^k}.
    \end{equation}
    \item \emph{Strong converse.} At rates above the capacity, the minimal error does not merely fail to vanish, but instead  tends to its maximal value $1$, exponentially fast in $n$. More precisely, for every $\eps > 0$ there is a constant $c > 0$, depending only on $\eps$, $\F$ and $\cG$, such that
    \begin{equation*}
    \deltaUA{k}{n}{\cG}{\F} \geq \deltaSR{k}{n}{\cG}{\F} \geq 1 - 2^{-cn} \quad \text{for all } k,n \in \N_{>0} \text{ with } k/n \geq \CSR{\cG}{\F} + \eps.
    \end{equation*}
    In particular, for any sequences $(k_\nu)$, $(n_\nu)$ with $n_\nu \to \infty$ and $\liminf_\nu k_\nu/n_\nu > \CSR{\cG}{\F}$, both $\deltaUA{k_\nu}{n_\nu}{\cG}{\F}$ and $\deltaSR{k_\nu}{n_\nu}{\cG}{\F}$ tend to $1$ exponentially fast in $n_\nu$.    \item \emph{Computability.} The capacity $\CUA{\cG}{\F} = \CSR{\cG}{\F}$ can be approximated up to an additive error $2^{-b}$ in time polynomial in $b$ and in the dimensions of the input spaces of $\F$ and $\cG$ (see Corollary~\ref{coro:algorithm_channels} for the computational model).
\end{enumerate}
\end{mainthm}

Item~(i) of Theorem~\ref{thm:main_UA_SR} is proved in Section~\ref{sec:emul_capa} (achievability) and Section~\ref{sec:bound_finite_blocklength_UA_SR} (converse), items~(ii) and~(iii) in Section~\ref{sec:bound_finite_blocklength_UA_SR}, and item~(iv) in Section~\ref{sec:algorithm}.

Our second main result gives a complete description of the entanglement-assisted regime for idempotent channels, at every finite blocklength.

\begin{mainthm}[Entanglement-assisted regime]\label{thm:main_EA}
Let $\F$ and $\cG$ be two idempotent channels.
\begin{enumerate}[label=(\roman*)]
    \item \emph{Capacity.} If $\la(\F) \neq (1)$, the entanglement-assisted emulation capacity of $\F$ by $\cG$ is
    \begin{equation}\label{eq:capacity_zero_error_EA}
        \CEA{\cG}{\F} = \frac{\log(\|\la(\cG)\|_2)}{\log(\|\la(\F)\|_2)},
    \end{equation}
    and it coincides with the zero-error entanglement-assisted emulation capacity.
    \item \emph{Exact finite blocklength error.} For all $k,n \in \N_{>0}$,
    \begin{equation}\label{eq:error_approx_case_EA}
        \deltaEA{k}{n}{\cG}{\F} = \max\left\{0, 1 - \frac{\|\la(\cG)\|_2^{2n}}{\|\la(\F)\|_2^{2k}}\right\}.
    \end{equation}
    Furthermore, when this quantity vanishes, $\F^{\otimes k}$ can be emulated exactly by $\cG^{\otimes n}$ using finitely many maximally entangled states.
    \item \emph{Strong converse.} If $\la(\F) \neq (1)$, for every $\eps > 0$, there is a constant $c > 0$, depending only on $\eps$ and $\F$, such that $\deltaEA{k}{n}{\cG}{\F} \geq 1 - 2^{-cn}$ for all $k,n \in \N_{>0}$ with $k/n \geq \CEA{\cG}{\F} + \eps$.
\end{enumerate}
\end{mainthm}

All items of Theorem~\ref{thm:main_EA} are proved in Section~\ref{sec:EA}.

\textbf{Ideas of the proofs.}
Both theorems rest on a reduction to \emph{block pinchings}.
An idempotent channel $\F$ is, up to encoding and decoding, a very simple object: it transmits a classical index $k \in [K]$ together with a perfect quantum system of dimension $d_k$, and the list $\la(\F) = (d_1,\dots,d_K)$ of these dimensions is its \emph{shape vector}. More precisely, $\F$ can be converted exactly into the \emph{block pinching} $\Pi_{\la(\F)}$, which acts on $\bigoplus_k \C^{d_k}$ by erasing the coherences between different blocks and leaving each block untouched (Definition~\ref{def:block_pinching}), and conversely $\Pi_{\la(\F)}$ can be converted exactly into $\F$ (Proposition~\ref{prop:idempotent_vs_pinching}). Consequently, every emulation question only depends on the shape vectors of the two channels (Corollary~\ref{coro:only_shape_vectors}).
Since shape vectors are multiplicative under tensor products (Proposition~\ref{prop:multiplicativity}), and so are their $\ell_p$-norms, all the expressions above are single-letter.

The achievability part of Theorem~\ref{thm:main_UA_SR} (Section~\ref{sec:emul_capa}) is proved in the unassisted regime, which suffices for the shared randomness regime as well, since an unassisted strategy is a shared randomness strategy.
It replaces both channels by the tracial conditional expectations onto the $*$-algebras $\bigoplus_k \cL(\C^{d_k})$, and relies on a theorem of Kuperberg~\cite{kuperberg}: when $\|\la(\F)\|_p < \|\la(\cG)\|_p$ for all $p \in [1,+\infty]$, the $n$-th tensor power of the first algebra embeds into the $n$-th tensor power of the second one for $n$ large enough.
We lift this embedding to an encoding and a decoding channel.
The converse (Section~\ref{sec:bound_finite_blocklength_UA_SR}) is a testing argument: we check whether the emulated channel preserves a suitable mixture of maximally entangled states on the blocks of $\Pi_{\la(\F)}$, and the Cauchy--Schwarz and Hölder inequalities bound the probability of passing this test by $\|\la(\cG)\|_p/\|\la(\F)\|_p$, uniformly over the shared randomness. 
Taking tensor powers yields the finite blocklength bound of item~(ii) in Theorem~\ref{thm:main_UA_SR}, from which all the other converse statements follow.
This argument is close to the proof of the Hölder-type inequality for tracial conditional expectations in~\cite[Theorem~4.1]{kuperberg}.
It nevertheless differs from it, as it bounds the diamond-norm error of arbitrary idempotent channels at every finite blocklength, instead of the fidelity of encoding some states, which is what yields a non-asymptotic strong converse.
The algorithm of item~(iv) (Section~\ref{sec:algorithm}) first restricts $p$ to a bounded range thanks to an explicit tail bound, then minimises certified polynomial approximations of the ratio $\log(\|\lambda(\cG)\|_p)/\log(\|\lambda(\F)\|_p)$ on logarithmically many dyadic intervals, by real root isolation.

In the entanglement-assisted regime (Section~\ref{sec:EA}), teleportation and superdense coding, applied block by block, show that $\Pi_{\la(\F)}$, and hence $\F$, can be converted exactly with the help of shared entanglement, with a single use in each direction, into the noiseless classical channel on an alphabet of size $\|\la(\F)\|_2^2$, and back.
The problem therefore reduces to the entanglement-assisted emulation of noiseless classical channels with one another.
A strategy based on a random permutation of the alphabet, which only uses shared randomness, is optimal by a no-signalling argument.

\textbf{Relation to previous work in the unassisted and shared randomness regimes.}
Several special cases of Theorem~\ref{thm:main_UA_SR} were known.
When both $\F$ and $\cG$ are identity channels, the capacity and a finite blocklength converse were obtained by Kretschmann and Werner~\cite[Section~3.2]{Kretschman04}. 
For these channels, our finite blocklength bound strengthens theirs by a factor of $2$.
When $\F$ is a noiseless quantum or classical channel, the infimum in Theorem~\ref{thm:main_UA_SR} is attained at an endpoint and the formula simplifies. 
Writing $\la(\cG) = (d_1,\dots,d_K)$, the identity channel $\Id_d$ has shape vector $(d)$ and the completely dephasing channel $\Delta_d$ has shape vector $(1,\dots,1)$, so that
\begin{equation*}
\CR{\cG}{\Id_d} = \frac{\log \max_k d_k}{\log d}, \qquad \CR{\cG}{\Delta_d} = \frac{\log \sum_k d_k}{\log d},
\end{equation*}
for $\rmR \in \{\UA,\SR\}$, with the infimum attained at $p = +\infty$ and $p = 1$ respectively. 
Taking $d = 2$, these are the quantum and classical capacities of the idempotent channel $\cG$, equal to $\log_2 \max_k d_k$ qubits and $\log_2 \sum_k d_k$ bits per use.
The quantum capacity and its strong converse were first obtained for self-adjoint idempotent channels (also known as tracial conditional expectations) in~\cite{gao2018capacity}, and then for the idempotent channels arising as large-time limits of Markovian noise in~\cite[Theorem~2.3 and Proposition~2.5]{Fawzi24}, together with the classical capacity, and independently, in the zero-error case, in~\cite{Singh.2024}; see also~\cite{Singh25}.
Theorem~\ref{thm:main_UA_SR} complements these results with an exponential strong converse in diamond norm, which also holds when the encoder and decoder share randomness.

\textbf{Relation to previous work in the entanglement-assisted regime.}
In the asymptotic regime with vanishing error, the value of $\CEA{\cG}{\F}$ in Theorem~\ref{thm:main_EA} already follows from known results.
Indeed, the entanglement-assisted classical capacity of an idempotent channel is $C_E(\F) = \log\|\la(\F)\|_2^2$~\cite[Proposition~3.4]{gao2018capacity}.
Moreover, by the entanglement-assisted coding theorem~\cite{Bennett02} combined with the quantum reverse Shannon theorem~\cite{Bennett02,Bennett14,Berta18}, any channel $\cG$ can emulate any channel $\F$ with shared entanglement and vanishing error at rate $C_E(\cG)/C_E(\F)$, and at no higher rate, which gives the expression of item~(ii) in Theorem~\ref{thm:main_EA} for idempotent channels.
However, for general input states, the quantum reverse Shannon theorem requires entangled resources beyond maximally entangled states, such as embezzling states~\cite{Bennett14,Berta18}.
Its statements are asymptotic, and its one-shot refinements~\cite{Fang20} express the simulation cost through smooth entropic quantities or semidefinite programs rather than in closed form.
Theorem~\ref{thm:main_EA} goes beyond this asymptotic statement in three respects.
First, it is a single-shot result: for all $k,n \in \N_{>0}$, the minimal error is given by a closed formula, which involves no vanishing error, no smoothing and no optimisation. In particular, the zero-error capacity coincides with $\CEA{\cG}{\F}$, and the error tends to $1$ exponentially fast above the capacity, with an explicit exponent.
Second, the optimal protocols are explicit and use only finitely many maximally entangled states: teleportation and superdense coding, applied block by block, reduce both channels to noiseless classical channels, which are then emulated with a random permutation of the alphabet.
Third, the converse holds for arbitrary shared entangled states, and even for arbitrary no-signalling assistance, for which only linear and semidefinite programming characterisations are available in general~\cite{Matthews12,DuanWinter16,Fang20}. Hence, for idempotent channels, maximally entangled states are as powerful as any no-signalling correlation, at every finite blocklength.
Zero-error capacities and exact simulation costs of classical channels assisted by non-local correlations were determined along these lines in~\cite{CLMW11}.

\textbf{Consequences.}
Theorems~\ref{thm:main_UA_SR} and~\ref{thm:main_EA} have several consequences, which we discuss in Sections~\ref{sec:emul_capa} and~\ref{sec:EA}.
First, in the unassisted and shared randomness regimes, the emulation of idempotent channels is in general not reversible: as shown in Example~\ref{ex:UA_teleportation}, there are idempotent channels for which $\CUA{\cG}{\F}\CUA{\F}{\cG} \neq 1$, because the infimum in item~(i) of Theorem~\ref{thm:main_UA_SR} can be attained at different values of $p$ for the two directions.
Second, in these regimes, the emulation capacity can be strictly superadditive under tensor products of source channels, i.e.\ there are channels $\cG_1,\cG_2$ and $\F$ for which $\CUA{\cG_1 \otimes \cG_2}{\F} > \CUA{\cG_1}{\F} + \CUA{\cG_2}{\F}$ (Example~\ref{ex:UA_complementary}).
In contrast, with shared entanglement, only the value $p = 2$ matters: if $\la(\F) \neq (1)$ and $\la(\cG) \neq (1)$, the emulation is reversible, $\CEA{\cG}{\F}\CEA{\F}{\cG} = 1$, and the capacity is additive under tensor products of source channels (Corollary~\ref{coro:EA_reversibility_additivity}).
At finite blocklength, the additive quantity $\log\|\la(\F)\|_2$ completely characterises exact entanglement-assisted conversions of idempotent channels (Corollary~\ref{coro:EA_exact_emulation}).
Finally, taking $\F = \Id_2$, we obtain the quantum capacity of arbitrary idempotent channels, together with a strong converse, in the unassisted, shared randomness and entanglement-assisted regimes.

\textbf{Organisation of the article.}
This article is structured as follows.
Section~\ref{sec:preliminaries} introduces our notation, basic preliminaries on quantum information and operator algebras, and the precise definitions of the emulation tasks that we consider.
Section~\ref{sec:idempotent_channels} studies the structure of idempotent channels and their reduction to block pinchings.
Section~\ref{sec:emul_capa} proves the capacity formula of Theorem~\ref{thm:main_UA_SR} in the unassisted and shared randomness regimes (item~(i)), starting with the achievability part.
Section~\ref{sec:bound_finite_blocklength_UA_SR} proves the finite blocklength bound (item~(ii)), from which we deduce the converse part of item~(i) and the strong converse (item~(iii)).
Section~\ref{sec:algorithm} gives an efficient algorithm to compute $\CSR{\cG}{\F} = \CUA{\cG}{\F}$ (item~(iv)).
Section~\ref{sec:EA} proves Theorem~\ref{thm:main_EA}, together with its consequences for entanglement-assisted emulation.
Finally, Appendix~\ref{sec:algebraic_converse} gives a purely algebraic proof of the converse of item~(i) for exact unassisted emulation, and Appendix~\ref{sec:approximate_algebras} applies the theory of approximate $C^*$-algebras recently developed by Kitaev~\cite{kitaev25} to strengthen this converse for reduced idempotent channels.

\section{Preliminaries and notation}\label{sec:preliminaries}
All Hilbert spaces in this work are assumed to be finite-dimensional.
Given a Hilbert space $\cH$, we write $\cL(\cH)$ for the set of linear operators on $\cH$.
This set is endowed with the Hilbert--Schmidt inner product, denoted $\langle x, y\rangle= \tr(x^* y)$ for $x,y \in \cL(\cH)$, with $x^*$ the conjugate transpose of $x$ and $\tr(\cdot)$ the trace on $\cL(\cH)$.
Note that $x^*$ is the adjoint of $x$ for the Hermitian product of $\cH$.
We denote by $\mathds{1} \in \cL(\cH)$ the identity operator on $\cH$.
Given $x \in \cL(\cH)$, we call the \emph{(right) support} of $x$, denoted $\supp(x) \subseteq \cH$, the orthogonal complement of its (right) kernel.
We write $x^0$ for the orthogonal projector onto the support of $x$, parallel to its kernel.
Furthermore, for $x \in \cL(\cH)$, we write $\|x\|_1$ for its Schatten $1$-norm and $\|x\|$ for its Schatten $\infty$-norm, which coincides with its operator norm as a linear operator on $\cH$. 
Recall that the set of Hermitian operators in $\cL(\cH)$ is a partially ordered set: for $x=x^*,y =y^*\in \cL(\cH)$, we write $x \leq y$ whenever $y-x$ is positive semidefinite.

A state is represented by a positive semidefinite operator $\rho \in \cL(\cH)$ with trace one, which is called a \emph{density operator}.
We will take $\D(\cH) \subseteq \cL(\cH)$ to be the set of density operators belonging to $\cL(\cH)$.
A measurement on a physical system $\cL(\cH)$ is represented by a positive operator-valued measure, POVM for short, which is a set of positive operators $\mu_1, \dots, \mu_K \in \cL(\cH)$, called \emph{effects}, which sum to the identity, i.e.\ $\sum_{k=1}^K \mu_k = \mathds{1}$.

A quantum channel $\F : \cL(\cH) \rightarrow \cL(\cH')$ is a completely positive trace preserving linear map. 
Each linear map $\F : \cL(\cH) \rightarrow \cL(\cH')$ admits an adjoint $\F^* : \cL(\cH') \rightarrow \cL(\cH)$ for the Hilbert--Schmidt inner product  such that for all $x \in \cL(\cH)$, $y \in \cL(\cH')$, $\langle \F(x), y \rangle = \langle x, \F^*(y)\rangle$.
When $\F$ is a channel, its adjoint $\F^*$ is a unital completely positive map. 
We write $\Id_d : \cL(\C^d) \rightarrow \cL(\C^d)$ for the identity channel on $\cL(\C^d)$ and $\Delta_d : \cL(\C^d) \rightarrow \cL(\C^d)$ for the completely dephasing one. 
The latter maps every $x \in \cL(\C^d)$ onto its diagonal $\operatorname{Diag}(x_{11}, \dots, x_{dd})$ (with respect to the standard orthonormal basis).
It is a `classical identity', as it perfectly preserves classical information while mapping its input onto a diagonal, hence classical, state.

We use the diamond norm as a figure of merit for the quality of an emulation.
We recall that for $\F : \cL(\cH) \rightarrow \cL(\cH')$, the diamond norm is defined as 
\begin{equation*}
\|\F\|_{\diamond} = \sup_{n \in \N_{>0}} \|\F \otimes \Id_n\|_{1 \to 1},
\end{equation*}
with $\|\cdot\|_{1 \to 1}$ the `one-to-one' norm, i.e.\ the map norm when the input and output spaces are endowed with the Schatten 1-norm, denoted $\|\cdot\|_1$.
That is, for all linear maps $\F : \cL(\cH) \rightarrow \cL(\cH')$,
\begin{equation*}
    \|\F\|_{1 \to 1} = \sup_{\|x\|_1 \leq 1} \|\F(x)\|_1.
\end{equation*}

Then, we recall basic properties of the $\ell_p$-norms. For $u \in \C^n$ and $p \in [1,+\infty)$, the $\ell_p$-norm of $u$ is defined as
\begin{equation}\label{eq:Lp_norm}
\|u\|_p = \Bigl(\sum_{k=1}^n |u_k|^p\Bigr)^{\frac{1}{p}},
\end{equation}
and $\|u\|_{\infty} = \sup_k |u_k|$. 
In the discussions of the next sections, we will frequently make use of the fact that the $\ell_p$-norms are \emph{multiplicative} under the tensor product of vectors.
Given $u \in \C^n$, $v \in \C^m$, the tensor product $u \otimes v \in \C^{nm}$ is the vector whose coordinates are the products of the coordinates of $u$ and $v$, i.e.\ $(u_kv_l)_{k,l}$ for $k \in [n]$, $l \in [m]$.
Then, for all $p \in [1,+\infty]$, $\|u\otimes v\|_p = \|u\|_p \|v\|_p$.

Finally, we end this section by introducing the notions of finite-dimensional $*$-algebras and their morphisms. 

In this article, we actually consider only $*$-subalgebras of $\cL(\cH)$.
Such a $*$-subalgebra $\A \subseteq \cL(\cH)$ is a linear subspace of $\cL(\cH)$ which is stable under the multiplication and the $*$ operation of $\cL(\cH)$, i.e.\ taking the conjugate transpose. 
That is, for all $x,y \in \A$, $xy \in \A$ and for all $x \in \A$, $x^* \in \A$. 
Furthermore, all $*$-algebras in this work are assumed to be \emph{unital}, that is there is an identity $\mathds{1} \in \A$ satisfying $\mathds{1}^* = \mathds{1}$ and, for all $x \in \A$, $\mathds{1}x = x \mathds{1} = x$. 
Such an identity operator is necessarily unique in each $*$-algebra. 

Given two $*$-algebras $\A$ and $\B$, a $*$-homomorphism $\iota : \A \rightarrow \B$ is a linear map that preserves the product and the $*$ operation in $\A$. 
That is, for all $x,y \in \A$, $\iota(xy) = \iota(x)\iota(y)$ and $\iota(x^*) = \iota(x)^*$. 
In many texts, $*$-homomorphisms between unital $*$-algebras are assumed to be unital as well, that is, they are assumed to map the identity of $\A$ onto the identity of $\B$. 
Nevertheless, in this article, we will work with more general \textbf{subunital} $*$-homomorphisms. 
Such a map $\iota : \A \rightarrow \B$ maps the identity of $\A$ onto $\iota(\mathds{1}_\A)$, which satisfies $\iota(\mathds{1}_\A) \leq \mathds{1}_\B$. 
Note that $*$-homomorphisms are completely positive~\cite{Davidson93}.

\subsection{Emulation of quantum channels}\label{subsec:emulation}

We now define precisely the emulation tasks that we consider.
Throughout this subsection, $\F$ and $\cG$ are two quantum channels, called the \emph{target} and the \emph{source} channel respectively, and $k,n \in \N_{>0}$.

\paragraph{Unassisted regime.}
An unassisted emulation of $\F^{\otimes k}$ by $\cG^{\otimes n}$ is given by a pair of channels $(\E,\D)$, where the encoding channel $\E$ maps the input space of $\F^{\otimes k}$ to the input space of $\cG^{\otimes n}$ and the decoding channel $\D$ maps the output space of $\cG^{\otimes n}$ to the output space of $\F^{\otimes k}$.
Its error is $\frac{1}{2}\|\F^{\otimes k} - \D\circ\cG^{\otimes n}\circ\E\|_\diamond$, and we write
\begin{equation}\label{eq:def_delta_UA}
    \deltaUA{k}{n}{\cG}{\F} = \inf_{\D,\E} \frac{1}{2}\|\F^{\otimes k} - \D\circ \cG^{\otimes n} \circ \E\|_{\diamond}
\end{equation}
for the minimal error, where the infimum is taken over all encoding and decoding channels $\E,\D$.

\paragraph{Shared randomness regime.}
We can naturally extend the power of the encoder and decoder by giving them access to \emph{shared randomness}.
In this regime, the pair of encoding and decoding channels is chosen at random from a set $\{(\E_\lambda,\D_\lambda) \mid \lambda \in [\Lambda]\}$, with $\lambda$ drawn according to a probability distribution $(p_\lambda)_{\lambda \in [\Lambda]}$, and the error of the emulation is $\frac{1}{2}\|\F^{\otimes k} - \sum_{\lambda \in [\Lambda]} p_\lambda \D_\lambda \circ \cG^{\otimes n} \circ \E_\lambda\|_\diamond$.
We write
\begin{equation}\label{eq:def_delta_SR}
    \deltaSR{k}{n}{\cG}{\F} = \inf_{(p_{\lambda}, \D_\lambda,\E_\lambda)_{\lambda \in [\Lambda]}} \frac{1}{2} \|\F^{\otimes k} - \sum_{\lambda \in [\Lambda]} p_\lambda \D_\lambda \circ \cG^{\otimes n} \circ \E_\lambda\|_\diamond,
\end{equation}
where the infimum is taken over all $\Lambda \in \N_{>0}$, all probability distributions $(p_\lambda)_{\lambda \in [\Lambda]}$ and all sets of pairs of encoding and decoding channels $\{(\E_\lambda,\D_\lambda) \mid \lambda \in [\Lambda]\}$.

\paragraph{Entanglement-assisted regime.}
In the \emph{entanglement-assisted} regime, the encoder and the decoder share a bipartite state $\omega_{A_0B_0}$ on systems $A_0,B_0$ of arbitrary finite dimension.
The encoding channel $\E$ acts on the input of $\F^{\otimes k}$ and on $A_0$, and the decoding channel $\D$ acts on the output of $\cG^{\otimes n}$ and on $B_0$.
With $\T_{\omega_{A_0B_0}} : x \mapsto x \otimes \omega_{A_0B_0}$ for $x \in \cL(\cH_{\F}^{\otimes k})$ the input of $\F^{\otimes k}$, the error of the emulation is
\begin{equation}\label{eq:def_emulation_EA}
    \frac{1}{2}\|\F^{\otimes k} - \D\circ
   (\cG^{\otimes n}\otimes\Id_{B_0})\circ(\E\otimes\Id_{B_0})\circ \T_{\omega_{A_0B_0}}\|_\diamond,
\end{equation}
and we write
\begin{equation}\label{eq:def_delta_EA}
    \deltaEA{k}{n}{\cG}{\F} = \inf_{\omega_{A_0B_0}, \E,\D}  \frac{1}{2}   \|\F^{\otimes k} - \D\circ
   (\cG^{\otimes n}\otimes\Id_{B_0})\circ(\E\otimes\Id_{B_0})\circ \T_{\omega_{A_0B_0}}\|_\diamond,
\end{equation}
where the infimum is taken over all bipartite states $\omega_{A_0B_0}$, with systems $A_0, B_0$ of arbitrary finite dimension, and over all encoding and decoding channels $\E,\D$.

\paragraph{Capacities.}
Let $\rmR \in \{\UA,\SR,\EA\}$ be one of the three regimes.
A rate $r \geq 0$ is \emph{achievable} in the regime $\rmR$ if there exist two sequences $(k_\nu)_{\nu \in \N}$ and $(n_\nu)_{\nu \in \N}$ of positive integers such that
\begin{equation}\label{eq:def_achievable_rate}
    \lim_{\nu \to \infty} n_\nu = +\infty, \qquad \liminf_{\nu \to \infty}\frac{k_\nu}{n_\nu} \geq r, \qquad \lim_{\nu \to \infty} \deltaR{k_\nu}{n_\nu}{\cG}{\F} = 0.
\end{equation}
The \emph{emulation capacity} of $\F$ by $\cG$ in the regime $\rmR$, denoted $\CR{\cG}{\F}$, is the supremum of the achievable rates, and it is understood to be zero when no rate is achievable.
We may also use the notation $C(\cG \mapsto \F)$ and $\delta_{k,n}(\cG \mapsto \F)$ when the regime is unambiguous.

We say that $\F^{\otimes k}$ can be emulated \emph{exactly}, or with \emph{zero error}, by $\cG^{\otimes n}$ in the regime $\rmR$ if there is an emulation in this regime whose error is equal to zero.
The \emph{zero-error emulation capacity} in the regime $\rmR$ is the supremum of the ratios $k/n$ over the pairs $(k,n)$ for which $\F^{\otimes k}$ can be emulated exactly by $\cG^{\otimes n}$, and is zero when there is no such pair.
If $\F^{\otimes k}$ can be emulated exactly by $\cG^{\otimes n}$, then, for every $m \in \N_{>0}$, $\F^{\otimes mk}$ can be emulated exactly by $\cG^{\otimes mn}$, by taking the $m$-th tensor powers of the encoding and decoding channels, and of the shared probability distribution or shared state.
Hence, every ratio $k/n$ of this form is an achievable rate, and the zero-error emulation capacity is at most the emulation capacity $\CR{\cG}{\F}$.
For idempotent channels, these two capacities coincide in all three regimes, by item~(i) of Theorems~\ref{thm:main_UA_SR} and~\ref{thm:main_EA}.

These three regimes form a hierarchy: an unassisted strategy is a shared randomness strategy with a single pair of encoding and decoding channels, and a shared randomness strategy is an entanglement-assisted protocol in which the shared state is a classical correlated state, a standard fact about shared randomness as a particular case of shared quantum states, see e.g.~\cite{watrous18}.
Hence, we have the inequalities
\begin{align}
    \CUA{\cG}{\F} &\leq \CSR{\cG}{\F} \leq \CEA{\cG}{\F},\\
    \deltaUA{k}{n}{\cG}{\F} &\geq \deltaSR{k}{n}{\cG}{\F} \geq \deltaEA{k}{n}{\cG}{\F}.
\end{align}
Shared randomness allows one to average several strategies, which a single deterministic pair cannot do, and shared entanglement provides correlations that no classical randomness can reproduce.
For idempotent channels, we will see that the first inclusion does not change the emulation capacity (item~(i) of Theorem~\ref{thm:main_UA_SR}), whereas the second one can increase it (Example~\ref{ex:EA_teleportation}).

We summarise the notation we use throughout this article in Table~\ref{table:notations}.

\begin{table}[ht]
\centering 
\begin{tabular}{|c|c|c|}
\hline
  $\cH$ & A finite-dimensional Hilbert space &\\
  $\cL(\cH)$ & Set of linear operators on $\cH$ & \\
  $\D(\cH)$ & Set of density operators in $\cL(\cH)$ & \\
  $\A, \B$ & $*$-subalgebras of $\cL(\cH)$ & \\
  $\rho$, $\sigma$ & Density operators &  \\
  $\supp(x)$ & Support of $x$ & \\
  $x^0$ & Orthogonal projector on the support of $x$  & \\
  $x^*$ & Hermitian adjoint of $x$ & \\
  $\|\cdot\|_p$ & $\ell_p$-norm &  Eq.~\eqref{eq:Lp_norm} \\
 $\F^*$ & Adjoint of $\F$ for the Hilbert--Schmidt product & \\
 $\widehat{\F}$ & Reduced channel associated to $\F$ & Prop.~\ref{prop:structure_idempotent_channel}, eq.~\eqref{eq:def_reduced_channel} \\
 $\Rg(\F)$ & Range of $\F$ & \\
 $\lambda(\F)$ & Shape vector of $\F$  &  Prop.~\ref{prop:structure_idempotent_channel} \\
 $\T_{\omega_{A_0B_0}}$ & Map $x \mapsto x \otimes \omega_{A_0B_0}$ & Eq.~\eqref{eq:def_emulation_EA} \\ 
 $[n]$ & The set of integers $\{1,\dots, n\}$ &\\
 $\R_{\geq 0}^K$, $\R_{>0}^K$ & Vectors in $\mathbb{R}^K$ with non-negative (resp. positive) entries & Lem.~\ref{lem:bilinear} \\ [1pt]
 $\N_{> 0}^K$ & Elements of $\mathbb{N}^K$ with non-zero entries & Lem.~\ref{lem:key} \\[2pt]
 $H(w)$ & Shannon entropy of a probability vector $w$ & Lem.~\ref{lem:properties_A_x} \\[2pt]
 $\Q[X]$ & Set of univariate polynomials with rational coefficients & Lem.~\ref{lem:root_isolation},~\ref{lem:rational_minimum} \\[2pt]
 $\lfloor y\rfloor_k$ &  A dyadic lower approximation of $y$ with error at most $2^{-k}$ & Eq.~\eqref{eq:truncation} \\[2pt]
 \hline
\end{tabular}
\caption{Notation}\label{table:notations}
\end{table}

\section{Decomposition properties of idempotent channels}\label{sec:idempotent_channels}

In this section, we discuss the decomposition properties of the range of idempotent channels, which are summarised in the following propositions.
The range of a map $f : \Omega \rightarrow \Omega'$, denoted $\Rg(f)$, is the set of all the images of the elements of its domain $\Omega$, i.e.\ $\Rg(f) = \{f(x) \mid x\in \Omega\}$. 
The range of an idempotent channel coincides with its set of fixed points and the structure of this space is well understood, see for example the lecture notes~\cite{Wolf12}.
The following proposition summarises the properties that we use.

\begin{prop}[Structure of the range of idempotent channels, see \cite{ Wolf12}]\label{prop:structure_idempotent_channel}
Let $\F : \cL(\cH) \rightarrow \cL(\cH)$ be an idempotent channel, i.e.\ satisfying $\F = \F \circ \F$.
We can write $\cH$ as the direct sum
\begin{equation}\label{eq:first_decompo_H}
    \cH = \cH_0 \oplus \cH_1,
\end{equation}
with $\cH_0 = \supp(\F(\mathds{1}))$ and $\cH_1 = \ker(\F(\mathds{1}))$, and $\cH_0$ can further be decomposed as the orthogonal direct sum
\begin{equation}\label{eq:adapted_decomposition}
    \cH_0 = \bigoplus_{k=1}^K \cH_{k,1} \otimes \cH_{k,2},
\end{equation}
so that there is a set of density operators $\{\rho_k \in \D(\cH_{k,2}) \mid  k \in [K]\}$, such that
\begin{equation}\label{eq:decompo_range_F}
    \Rg(\F) = \Bigl(\bigoplus_{k=1}^K \cL(\cH_{k,1}) \otimes \rho_k \Bigr) \oplus 0.
\end{equation}

Let $e = \F(\mathds{1})^0 \in \cL(\cH)$ be the orthogonal projector onto $\cH_0$, and $V : \cH_0 \rightarrow \cH$ be the isometry embedding $\cH_0$ into $\cH$ such that $VV^* = e$ and $V^*V = \mathds{1}_0$ the identity operator in $\cL(\cH_0)$.
The \emph{reduced quantum channel} associated to $\F$, denoted $\widehat{\F} : \cL(\cH_0) \rightarrow \cL(\cH_0)$ is then defined as
\begin{equation}\label{eq:def_reduced_channel}
    \widehat{\F} : x \mapsto V^*\F(VxV^*)V, 
\end{equation}
 which is idempotent and whose range can be written as
\begin{equation}\label{eq:decompo_range_reduced_channel}
    \Rg(\widehat{\F}) = \bigoplus_{k=1}^K \cL(\cH_{k,1})\otimes \rho_k.
\end{equation}

Furthermore, we can express $\widehat{\F}$ and its adjoint $\widehat{\F}^*$ in the following way:
\begin{equation}\label{eq:form_idempotent_channel}
\widehat{\F}(x) = \sum_{k=1}^K \tr_{k,2}(P_k x P_k) \otimes \rho_k,
\end{equation}
and
\begin{equation}\label{eq:decompo_F_star}
\widehat{\F}^*(x) = \sum_{k=1}^K\tr_{k,2}(P_kxP_k(\mathds{1}_{d_k} \otimes \rho_k))\otimes \mathds{1}_{m_k},
\end{equation}
where, for all $k \in [K]$, $P_k$ is the orthogonal projector onto $\cH_{k,1} \otimes \cH_{k,2}$, $d_k = \dim(\cH_{k,1}) \geq 1$, $m_k =  \dim(\cH_{k,2})$ and $\tr_{k,2}: \cL(\cH_{k,1} \otimes \cH_{k,2}) \rightarrow \cL(\cH_{k,1})$ is the partial trace over $\cH_{k,2}$.

Therefore, the range of $\widehat{\F}^*$ is a unital $*$-subalgebra of $\cL(\cH_0)$ and the following decomposition holds:

\begin{equation}\label{eq:range_F_star}
\Rg(\widehat{\F}^*) = \bigoplus_{k=1}^K \cL(\cH_{k,1}) \otimes \mathds{1}_{m_k}.
\end{equation}

We call the \emph{shape} of $\F$ the $K$-dimensional vector whose coordinates are the dimensions $(d_k)_{k \in [K]}$ sorted in non-increasing order, which we denote by $\lambda(\F)$.
\end{prop}

Note that $\F$ and its reduced channel $\widehat{\F}$ both have the same shape vector, i.e.\ $\la(\F) = \la(\widehat{\F})$.

As an example of a shape vector, consider the qubit channel $\F : \cL(\C^2) \rightarrow \cL(\C^2)$ defined as
\begin{equation*}
\F : x \mapsto \tr(x)\proj{0} = \begin{pmatrix}
\tr(x) & 0 \\
0 & 0
\end{pmatrix},
\end{equation*}
where the matrix expression is written in the basis $(\ket{0}, \ket{1})$.
We have 
\begin{equation*}
\F(\mathds{1}) = \begin{pmatrix} 2 & 0\\ 0 & 0
\end{pmatrix}.
\end{equation*}
Then, in the notation of eq.~\eqref{eq:adapted_decomposition}, we get $\cH_0 = \supp(\F(\mathds{1})) = \C\ket{0}$, $\cH_1 = \ker(\F(\mathds{1})) = \C\ket{1}$, we can write $\Rg(\F) = \C \otimes \proj{0} \oplus 0$ and thus $\la(\F) = (1)$.
To construct the reduced version of $\F$ we can take the isometry $V = \begin{pmatrix} 1 \\ 0\end{pmatrix}: \cH_0  \rightarrow \cH$, where the expression of $V$ is written in the basis $(\ket{0}, \ket{1})$.
Note that we clearly have 
\begin{equation*}
VV^* = \begin{pmatrix}
    1 & 0 \\
    0 & 0
\end{pmatrix} = \F(\mathds{1})^0.
\end{equation*}
Then, $\widehat{\F} : \cL(\C) \rightarrow \cL(\C)$ is defined by the equation
\begin{equation*}
\widehat{\F}(x) = V^*\F(VxV^*)V = x,
\end{equation*}
so that $\widehat{\F}$ is the identity on $\cL(\C)$.

Furthermore, we remark that $\la(\F) = (1)$ if and only if the reduced channel of $\F$ is a replacer channel, i.e.\ $\widehat{\F} : x \mapsto \tr(x)\rho$ with $\rho$ a fixed state. This fact follows easily from eq.~\eqref{eq:form_idempotent_channel} taking $K=1$.

Before we present the proof of Proposition~\ref{prop:structure_idempotent_channel}, we need to state the following two lemmas, which will be used in the proof.

\begin{lemma}\label{lemma:stability_support_one}
Let $\A \subseteq \cL(\cH)$ be a $*$-subalgebra containing the identity operator $\mathds{1}_{\cH}$, and let $\Phi : \A \rightarrow \cL(\cH')$ be a positive linear map.
Let $\cH_0 = \supp(\Phi(\mathds{1}_{\cH})) \subseteq \cH'$ and let $V : \cH_0 \rightarrow \cH'$ be the inclusion isometry, so that $VV^*$ is the orthogonal projector onto $\cH_0$.
Then, for every $x \in \A$,
\begin{equation*}
    VV^*\Phi(x)VV^* = \Phi(x).
\end{equation*}
\end{lemma}

\begin{proof}
Let $x \in \A$ be positive semidefinite, then $0 \leq x \leq \|x\|\mathds{1}$, so that by positivity of $\Phi$, we have 
\begin{equation*}
0 \leq \Phi(x) \leq \|x\|\Phi(\mathds{1}).
\end{equation*}
Hence, $\Phi(x)$ is supported on the support of $\Phi(\mathds{1})$ which is, by definition, $\cH_0$.
Furthermore, as $x$ and $\mathds{1}$ are positive semidefinite, $\Phi(x)$ and $\Phi(\mathds{1})$ are positive semidefinite as well so that their respective ranges (as linear maps on $\cH'$) coincide with their supports.
Hence, we have 
    \begin{equation*}
    VV^*\Phi(x)VV^* = \Phi(x).
    \end{equation*}
Now, if $x \in \A$ is arbitrary, we can use the polarisation identity (see~\cite[eq.~(1.45)]{Wolf12}) to write 
\begin{equation*}
x = \frac{1}{4}\sum_{k=0}^3 i^k(x + i^k \mathds{1})^*(x + i^k \mathds{1}),
\end{equation*}
so that $x$ is a linear combination of four positive semidefinite elements of $\A$.
Now, by linearity, we get 
    \begin{equation*}
    VV^*\Phi(x)VV^* = \Phi(x).
    \end{equation*}
\end{proof}

Furthermore, we remark that the set of fixed points and the range of an idempotent function coincide.

\begin{lemma}\label{lemma:fixed_points}
    Let $f : \Omega \rightarrow \Omega$ be an arbitrary idempotent function, with $\Omega$ an arbitrary set, then
    \begin{equation}\label{eq:fixed_points_vs_image}
    \Rg(f) = \{x \in \Omega\mid f(x) = x\}.
    \end{equation}
\end{lemma}

\begin{proof}
    If $x$ is a fixed point, $x = f(x) \in \Rg(f)$. 
    For the other inclusion, if $y = f(x)$, we also have $f(y) = f(f(x)) = f(x) = y$, so $y$ is a fixed point.  
\end{proof}

Proposition~\ref{prop:structure_idempotent_channel} then basically follows from statements in~\cite[Section~6.4]{Wolf12} and we detail it here only for the convenience of the reader.
A similar derivation can be found in~\cite{Amato.2025} in the Heisenberg picture.

\begin{proof}[Proof of Proposition~\ref{prop:structure_idempotent_channel}]
We will first prove that the reduced channel $\widehat{\F}$, defined in eq.~\eqref{eq:def_reduced_channel}, is actually a channel.
It is clearly completely positive as it is a composition of completely positive maps and it is trace preserving because for all $x \in \cL(\cH_0)$,
\begin{equation*}
\tr(\widehat{\F}(x)) = \tr(V^*\F(VxV^*)V) = \tr(VV^*\F(VxV^*)VV^*) = \tr(\F(VxV^*)) = \tr(VxV^*) = \tr(x),
\end{equation*}
where the third equality follows from Lemma~\ref{lemma:stability_support_one}.
Furthermore, $\widehat{\F}$ is idempotent as, for $x \in \cL(\cH_0)$,
\begin{align*}
    \widehat{\F}(\widehat{\F}(x)) &=V^*\F(VV^*\F(VxV^*)VV^*)V \\
    &= V^*\F(\F(VxV^*))V\\
    &= V^*\F(VxV^*)V\\
    &= \widehat{\F}(x).
\end{align*}
 The second equation follows from Lemma~\ref{lemma:stability_support_one}, and the third from the idempotence of $\F$.

 Furthermore, $\widehat{\F}$ admits a full-rank fixed point.
 Indeed, $\widehat{\F}(V^*\F(\mathds{1})V)$ is a fixed point of $\widehat{\F}$, since $\widehat{\F}$ is idempotent, and we have
\begin{equation*}
\widehat{\F}(V^*\F(\mathds{1})V) = V^*\F(VV^*\F(\mathds{1})VV^*)V = V^*\F\F(\mathds{1})V = V^*\F(\mathds{1})V,
\end{equation*}
which is full-rank in $\cH_0 = \supp(\F(\mathds{1}))$, again using Lemma~\ref{lemma:stability_support_one} for the second equality.

Therefore, we can apply~\cite[Theorem~6.12]{Wolf12} to the adjoint $\widehat{\F}^*$ of $\widehat{\F}$. 
Hence, the set of fixed points of $\widehat{\F}^*$ is a $*$-subalgebra of $\cL(\cH_0)$. 
By Lemma~\ref{lemma:fixed_points}, since $\widehat{\F}^*$ is idempotent, its set of fixed points coincides with its range, so that its range is a $*$-subalgebra of $\cL(\cH_0)$.
It is well known (see e.g.~\cite{takesaki03}) that if $\A \subseteq \cL(\cH_0)$ is a $*$-subalgebra of $\cL(\cH_0)$, we can find a basis of $\cH_0$ in which it can be expressed as an orthogonal direct sum as in eq.~\eqref{eq:range_F_star}, which proves this equation. 
Furthermore, we get by~\cite[Proposition~1.5]{Wolf12} that $\widehat{\F}^*$ takes the form of eq.~\eqref{eq:decompo_F_star}. 
Then, taking the adjoint of this equation, we obtain that $\widehat{\F} = \widehat{\F}^{**}$ can be written as in eq.~\eqref{eq:form_idempotent_channel}. 
Then, for $x \in \cL(\cH_0)$, we have that
\begin{equation*}
\widehat{\F}(x) = \sum_{k=1}^K \tr_{k,2}(P_k x P_k) \otimes \rho_k \in \bigoplus_{k=1}^K \cL(\cH_{k,1}) \otimes \rho_k.
\end{equation*}
Conversely, if 
\begin{equation*}
    x = \sum_{k=1}^K x_k \otimes \rho_k \in \bigoplus_{k=1}^K \cL(\cH_{k,1}) \otimes \rho_k,
\end{equation*}
where, for each $k \in [K]$, $x_k \in \cL(\cH_{k,1})$, we have that $P_lx_k\otimes\rho_kP_l = \delta_{k,l}x_k\otimes \rho_k$, since the spaces $\cH_{k,1}\otimes \cH_{k,2}$ are pairwise orthogonal, so that $x = \widehat{\F}(x)$, and finally $x \in \Rg(\widehat{\F})$. 
This proves eq.~\eqref{eq:decompo_range_reduced_channel}.

Finally, to prove eq.~\eqref{eq:decompo_range_F}, let $x \in \Rg(\F)$.
We then have
\begin{equation}\label{eq:support_element_image_F}
x = \F(x) = VV^*\F(x)VV^* = VV^*xVV^*,
\end{equation}
so that $x \in \cL(\cH_0) \oplus 0$ where the $0$ operator acts on $\cH_1$.
Then, $V^*xV \in \cL(\cH_0)$ and
\begin{equation*}
    \widehat{\F}(V^*xV) = V^*\F(VV^*xVV^*)V = V^*\F(x)V = V^*xV,
\end{equation*}
by Lemma~\ref{lemma:stability_support_one} and by eq.~\eqref{eq:support_element_image_F}.
Hence,
\begin{equation*}
  V^*xV \in \bigoplus_{k=1}^K \cL(\cH_{k,1}) \otimes \rho_k,  
\end{equation*}
by eq.~\eqref{eq:decompo_range_reduced_channel} and 
\begin{equation*}
  x = VV^*xVV^* \in \bigoplus_{k=1}^K \cL(\cH_{k,1}) \otimes \rho_k \oplus 0.  
\end{equation*}
Conversely, if 
\begin{equation*}
    x = \begin{pmatrix}
    x_0 & 0 \\ 0 & 0
    \end{pmatrix},
\end{equation*}
with 
\begin{equation*}
    x_0 \in \bigoplus_{k=1}^K \cL(\cH_{k,1}) \otimes \rho_k,
\end{equation*}
then 
\begin{equation*}
    x = Vx_0V^* = V\widehat{\F}(x_0)V^* = VV^*\F(Vx_0V^*)VV^* = \F(Vx_0V^*) = \F(x),
\end{equation*}
so that $x \in \Rg(\F)$.
\end{proof}

We also show in this section that the shape vector introduced in Proposition~\ref{prop:structure_idempotent_channel} is the object that governs the interconversion of idempotent channels, that is, for idempotent $\F$ and $\cG$, it suffices to know their shape vectors to specify the rates that will achieve any given error for the emulation of $\F$ by $\cG$.

We first show that if $\F$ is idempotent, it can be interconverted exactly with its reduced channel $\widehat{\F}$, which can then be interconverted with the block pinching $\Pi_\mathbf{f}$, introduced in Definition~\ref{def:block_pinching}, with $\mathbf{f} = \lambda(\F)$.

\begin{prop}\label{prop:conversion_to_reduced_channel}
    Let $\F : \cL(\cH) \rightarrow \cL(\cH)$ be an idempotent channel, let $\cH_0 = \supp(\F(\mathds{1}))$ and let $\widehat{\F} : \cL(\cH_0) \rightarrow \cL(\cH_0)$ be its reduced quantum channel.
    Then, there are channels $\widehat{\E} : \cL(\cH) \rightarrow \cL(\cH_0)$, $\widehat{\D} : \cL(\cH_0) \rightarrow \cL(\cH)$, $\E : \cL(\cH_0) \rightarrow \cL(\cH)$ and $\D : \cL(\cH) \rightarrow \cL(\cH_0)$ such that 
    \begin{equation}
        \F = \widehat{\D}\widehat{\F}\widehat{\E}, \text{ and } \widehat{\F} = \D\F\E.
    \end{equation}
\end{prop}

\begin{proof}
Let $V : \cH_0 \rightarrow \cH$ be the isometry embedding $\cH_0$ into $\cH$  which defines $\widehat{\F}$ as in Proposition~\ref{prop:structure_idempotent_channel}. 
For the first equality, let $\widehat{\E} : \cL(\cH) \rightarrow \cL(\cH_0)$ be defined by
\begin{equation}
\widehat{\E} : x \mapsto V^*\F(x)V,
\end{equation}
and $\widehat{\D} : \cL(\cH_0) \rightarrow \cL(\cH)$ be defined by
\begin{equation}
\widehat{\D} : x \mapsto VxV^*.
\end{equation}
 Both $\widehat{\E}$ and $\widehat{\D}$ are completely positive trace preserving maps. 
 The complete positivity is clear. 
 $\widehat{\D}$ is trace preserving because $\tr(VxV^*) = \tr(V^*Vx) = \tr(x)$ and $\widehat{\E}$ is trace preserving as well because for all $x \in \cL(\cH)$, 
 \begin{equation}
     \tr(V^*\F(x)V)  = \tr(VV^*\F(x)VV^*) =  \tr(\F(x)) =  \tr(x),
 \end{equation}
 where the second equality follows from Lemma~\ref{lemma:stability_support_one}.
 Then we have the factorisation
\begin{equation}
\F = \widehat{\D}\widehat{\F}\widehat{\E}.
\end{equation}
 Indeed, let $x \in \cL(\cH)$, we have 
 \begin{equation*}
    \widehat{\D}\widehat{\F}\widehat{\E}(x) = VV^*\F(VV^*\F(x)VV^*)VV^* 
    = \F(\F(x)) 
    = \F(x),
 \end{equation*}
 where the second equality follows from Lemma~\ref{lemma:stability_support_one}.
 Now, for the factorisation of $\widehat{\F}$, let $\E = \widehat{\D}$ and $\D = \widehat{\E}$, for $x \in \cL(\cH_0)$, we have:
 \begin{equation*}
     \D\F\E(x) = \widehat{\E}\F\widehat{\D}(x) 
     = V^*\F(\F(VxV^*))V
     = V^*\F(VxV^*)V 
     = \widehat{\F}(x).
 \end{equation*}
\end{proof}

We now introduce the block pinching $\Pi_{\mathbf{f}}$, which serves as the canonical idempotent channel with shape vector $\lambda(\Pi_{\mathbf{f}}) = \mathbf{f}$. 

\begin{definition}[Block pinching]\label{def:block_pinching}
Let $\mathbf{d} = (d_1,\dots,d_K) \in \N_{>0}^K$ and let
\begin{equation}\label{eq:def_H_d}
    \cH_{\mathbf{d}} = \bigoplus_{k=1}^K \C^{d_k}.
\end{equation}
For $k \in [K]$, let $J_k : \C^{d_k} \rightarrow \cH_{\mathbf{d}}$ be the isometry embedding $\C^{d_k}$ as the $k$-th summand. That is, for $\nu \in [d_k]$, $J_k\ket{\nu} = \ket{\nu}_k$, with $\cup_{k \in [K]} \{\ket{\nu}_k \mid \nu \in [d_k]\}$ an orthonormal basis of $\cH_{\mathbf{d}}$. Let $\widetilde{P}_k = J_kJ_k^* \in \cL(\cH_{\mathbf{d}})$ be the orthogonal projector onto this summand.
The \emph{block pinching} associated to $\mathbf{d}$ is the channel
\begin{equation}\label{eq:def_pinching_d}
    \Pi_{\mathbf{d}} : \cL(\cH_{\mathbf{d}}) \rightarrow \cL(\cH_{\mathbf{d}}), \qquad \Pi_{\mathbf{d}}(z) = \sum_{k=1}^K \widetilde{P}_kz\widetilde{P}_k.
\end{equation}
Given an idempotent channel $\F$, we write $\Pi_{\mathbf{f}} = \Pi_{\la(\F)}$ for the block pinching associated to its shape vector.
\end{definition}

Operationally, $\Pi_{\mathbf{d}}$ perfectly transmits a classical label $k \in [K]$ together with a $d_k$-dimensional quantum system.
It is a channel since $\sum_{k \in  [K]} \widetilde{P}_k = \mathds{1}_{\mathbf{d}}$ the identity of $\cL(\cH_{\mathbf{d}})$.
Furthermore, it is idempotent, and its shape vector is $\mathbf{d}$, up to the order of the entries, since $\Rg(\Pi_{\mathbf{d}}) = \bigoplus_{k=1}^K\cL(\C^{d_k})$.
In particular $\la(\Pi_{\mathbf{f}}) = \la(\F)$.
Given $\F$ an idempotent channel, we can interconvert $\F$ with $\Pi_{\mathbf{f}}$ and vice versa, as shown in the following proposition.

\begin{prop}\label{prop:idempotent_vs_pinching}
Let $\F : \cL(\cH_\F) \rightarrow \cL(\cH_\F)$ be an idempotent channel and let $\Pi_{\mathbf{f}}$ be the block pinching associated to its shape vector.
Then there exist encoding and decoding channels $\E_{1},\D_{1},\E_{2},\D_{2}$ such that
\begin{equation}\label{eq:idempotent_vs_pinching}
    \F = \D_{1}\circ\Pi_{\mathbf{f}}\circ\E_{1}, \qquad \Pi_{\mathbf{f}} = \D_{2}\circ\F\circ\E_{2}.
\end{equation}
\end{prop}

\begin{proof}
By Proposition~\ref{prop:conversion_to_reduced_channel}, $\F$ and its reduced channel $\widehat{\F}$ can be converted one into the other by encoding and decoding channels, and $\la(\widehat{\F}) = \la(\F)$ by Proposition~\ref{prop:structure_idempotent_channel}.
Since compositions of channels are channels, it therefore suffices to prove eq.~\eqref{eq:idempotent_vs_pinching} for $\widehat{\F}$, and we assume from now on that $\F = \widehat{\F}$.

We use the notation of Proposition~\ref{prop:structure_idempotent_channel}: we have $\cH_\F = \bigoplus_{k=1}^K\cH_{k,1}\otimes\cH_{k,2}$ with $d_k = \dim(\cH_{k,1})$ and $\mathbf{f} = \la(\F) = (d_1,\dots,d_K)$ up to the order of the entries, $P_k$ is the orthogonal projector onto $\cH_{k,1}\otimes\cH_{k,2}$, and
\begin{equation*}
    \F(x) = \sum_{k=1}^K\tr_{k,2}(P_kxP_k)\otimes\rho_k.
\end{equation*}
We identify $\cH_{\mathbf{f}} = \bigoplus_{k=1}^K\C^{d_k}$ with $\bigoplus_{k=1}^K\cH_{k,1}$, so that $\widetilde{P}_k$ is the orthogonal projector onto $\cH_{k,1}$ and $\Pi_{\mathbf{f}}(z) = \sum_{k=1}^K\widetilde{P}_kz\widetilde{P}_k$.
Define
\begin{equation}\label{eq:pinching_conversions}
    \F_1(x) = \sum_{k=1}^K\tr_{k,2}(P_kxP_k), \qquad \F_2(z) = \sum_{k=1}^K\widetilde{P}_kz\widetilde{P}_k\otimes\rho_k,
\end{equation}
so that $\F_1 : \cL(\cH_\F)\to\cL(\cH_{\mathbf{f}})$ and $\F_2 : \cL(\cH_{\mathbf{f}})\to\cL(\cH_\F)$.
Both are channels: they are completely positive as compositions of conjugations by the $P_k$ or $\widetilde P_k$, of partial traces and of tensorisations by the fixed states $\rho_k$, and they are trace preserving since $\sum_kP_k = \mathds{1}_{\cH_\F}$ and $\sum_k\widetilde P_k = \mathds{1}_{\cH_{\mathbf{f}}}$.
Furthermore, we have 
\begin{equation}\label{eq:factorisation_f}
    \F = \F_2 \circ \F_1, \quad \Pi_{\mathbf{f}} = \F_1 \circ \F_2.
\end{equation}

\emph{ $\Pi_{\mathbf{f}} \mapsto \F$:}
We have, by eq.~\eqref{eq:factorisation_f} and by idempotence of $\F$
\begin{equation*}
    \F = \F \circ \F =  \F_2 \circ \F_1 \circ  \F_2 \circ \F_1 = \F_2 \circ \Pi_{\mathbf{f}} \circ \F_1.
\end{equation*}

\emph{$\F \mapsto \Pi_{\mathbf{f}}$:}
Again by eq.~\eqref{eq:factorisation_f} and by idempotence of $\Pi_{\mathbf{f}}$
\begin{equation*}
    \Pi_{\mathbf{f}} = \Pi_{\mathbf{f}} \circ \Pi_\mathbf{f} = \F_1 \circ \F_2 \circ \F_1 \circ \F_2 = \F_1 \circ \F \circ \F_2. 
\end{equation*}
\end{proof}

Therefore, we can interconvert exactly any two idempotent channels having the same shape vector.

\begin{prop}[Exact emulation]\label{prop:emulation_same_shape}
Let $\F_1,\F_2$ be idempotent channels such that $\lambda(\F_1) = \lambda(\F_2)$, then there are encoding and decoding channels such that
\begin{equation}
    \F_1 = \D_1\F_2\E_1,\quad\F_2 = \D_2\F_1\E_2.
\end{equation}
\end{prop}

\begin{proof}
    By Proposition~\ref{prop:idempotent_vs_pinching}, since $\lambda(\F_1) = \lambda(\F_2) = \mathbf{f}$ we can perform the emulation of $\F_2$ by $\F_1$ through $\Pi_{\mathbf{f}}$ and vice versa.
\end{proof}

We now show that if we can interconvert exactly general channels $\F_1$ with $\F_2$ as well as $\cG_1$ with $\cG_2$, then it is equivalent to study the emulation of $\F_1$ by $\cG_1$ or that of $\F_2$ by $\cG_2$.

\begin{prop}\label{prop:equivalence_emulation}
Let $\F_1, \F_2$ and $\cG_1, \cG_2$ be general quantum channels such that 
\begin{align}
    &\F_1 = \D_{1,f} \F_2 \E_{1,f},  &\F_2 = \D_{2,f}\F_1\E_{2,f} \\
    &\cG_1 = \D_{1,g} \cG_2 \E_{1,g},  &\cG_2 = \D_{2,g}\cG_1 \E_{2,g},
\end{align}
with $ \E_{i,a},\D_{i,a}$ for $(i,a)\in \{1,2\}\times\{f,g\}$ pairs of encoding and decoding channels. 
Then, in every assistance regime $\rmR$, i.e.\ $\rmR \in \{ \UA, \SR, \EA\}$, and for all integers $k,n \in \N_{>0}$
\begin{equation}
    \deltaR{k}{n}{\cG_1}{\F_1} = \deltaR{k}{n}{\cG_2}{\F_2}.
\end{equation}
In particular,
\begin{equation}
    \CR{\cG_1}{\F_1} = \CR{\cG_2}{\F_2}.
\end{equation}
\end{prop}

\begin{proof}
    We prove the statement for $k=n=1$.
    The general case follows by applying the same argument to the tensor powers $\F_i^{\otimes k}$, $\cG_i^{\otimes n}$, $i \in [2]$, using the tensor powers of the conversion maps $\E_{i,a},\D_{i,a}$ for $(i,a)\in \{1,2\}\times\{f,g\}$.

    Let us first show $\deltaSR{1}{1}{\cG_1}{\F_1} = \deltaSR{1}{1}{\cG_2}{\F_2}$, since the same construction proves the unassisted case by taking $\Lambda=1$.
    Let $\Lambda \in \N_{>0}$, $(p_\lambda)_{\lambda \in [\Lambda]}$ be a probability distribution and $\{(\E_\lambda,\D_\lambda) \mid \lambda \in [\Lambda]\}$ be a set of pairs of encoding and decoding channels. 
    Then, 
    \begin{align*}
        \|\F_1 - \sum_{\lambda \in [\Lambda]}p_\lambda\D_\lambda\cG_1\E_\lambda\|_\diamond &= \|\F_1 - \sum_{\lambda \in [\Lambda]}p_\lambda\D_\lambda \D_{1,g}\cG_2\E_{1,g}\E_\lambda\|_\diamond \\
        &= \|\D_{2,f}\|_\diamond\|\F_1 - \sum_{\lambda \in [\Lambda]}p_\lambda\D_\lambda \D_{1,g}\cG_2\E_{1,g}\E_\lambda\|_\diamond\|\E_{2,f}\|_\diamond \\
        &\geq \|\D_{2,f}\F_1\E_{2,f} - \D_{2,f} \sum_{\lambda \in [\Lambda]}p_\lambda\D_\lambda \D_{1,g}\cG_2\E_{1,g}\E_\lambda \E_{2,f}\|_\diamond \\
        &= \|\F_2 - \sum_{\lambda \in [\Lambda]}p_\lambda\D_{2,f}\D_\lambda\D_{1,g}\cG_2\E_{1,g}\E_\lambda\E_{2,f}\|_\diamond,
    \end{align*}
    where we used in the second equality that, as $\E_{2,f}$, $\D_{2,f}$ are channels, $\|\E_{2,f}\|_\diamond = \|\D_{2,f}\|_\diamond = 1$.

    Therefore,
    \begin{equation*}
    \deltaSR{1}{1}{\cG_2}{\F_2} \leq \deltaSR{1}{1}{\cG_1}{\F_1}.
    \end{equation*}
    We can apply the same line of argument but converting $\F_2$ to $\F_1$ and $\cG_1$ to $\cG_2$ to prove the converse inequality.

    The proof in the entanglement-assisted case is very similar. Recall that given a bipartite state $\omega_{A_0B_0}$ on systems $A_0,B_0$, we write $\T_{\omega_{A_0B_0}}: x \mapsto x \otimes \omega_{A_0B_0}$ for $x \in \cL(\cH_{\F_1})$. 
    Hence given $(\E,\D)$ a pair of encoding and decoding channels,
    \begin{equation}\label{eq:reduced_version_EA}
        \|\F_1 - \D
   (\cG_1\E\otimes\Id_{B_0}) \T_{\omega_{A_0B_0}}\|_\diamond \geq \|\F_2 - \D_{2,f} \D  (\D_{1,g}\cG_2\E_{1,g}\E \otimes \Id_{B_0})\T_{\omega_{A_0B_0}}\E_{2,f}\|_\diamond.
    \end{equation}
    Now, in the right-hand side of eq.~\eqref{eq:reduced_version_EA}, $\T_{\omega_{A_0B_0}}$ and $\E_{2,f}$ commute, in the sense that for every $x \in \cL(\cH_{\F_2})$, 
    \begin{equation*}
    \T_{\omega_{A_0B_0}} \E_{2,f}(x) = \E_{2,f}(x) \otimes \omega_{A_0B_0} = (\E_{2,f} \otimes \Id_{A_0B_0})\T_{\omega_{A_0B_0}}(x).
    \end{equation*}    
    Hence, 
    \begin{equation}
        \|\F_1 - \D
   (\cG_1\E\otimes\Id_{B_0}) \T_{\omega_{A_0B_0}}\|_\diamond \geq \|\F_2 - \D^{\sharp}(\cG_2 \E^{\sharp} \otimes \Id_{B_0})\T_{\omega_{A_0B_0}}\|_\diamond,
    \end{equation}
    where $\D^{\sharp} = \D_{2,f}\D(\D_{1,g}\otimes \Id_{B_0})$ and $\E^{\sharp} = \E_{1,g}\E (\E_{2,f} \otimes \Id_{A_0})$.
    The converse inequality is proven by swapping the roles of $\F_1$ and $\cG_1$ with those of $\F_2$ and $\cG_2$.
\end{proof}

Therefore, the shape vector is the only relevant object governing the interconversion of idempotent channels.

\begin{corollary}[Interconversion is governed by shape vectors]\label{coro:only_shape_vectors}
    Let $\F_1,\F_2$ and $\cG_1, \cG_2$ be idempotent channels such that $\lambda(\F_1) = \lambda(\F_2)$ and $\lambda(\cG_1) = \lambda(\cG_2)$. Then, in all regimes $\rmR \in \{\UA,\SR,\EA\}$ and for all $k,n \in \N_{>0}$
    \begin{equation}
        \deltaR{k}{n}{\cG_1}{\F_1} = \deltaR{k}{n}{\cG_2}{\F_2}.
    \end{equation}
    In particular
    \begin{equation}
        \CR{\cG_1}{\F_1} = \CR{\cG_2}{\F_2}.
    \end{equation}
\end{corollary}

\begin{proof}
Follows from Propositions~\ref{prop:emulation_same_shape} and~\ref{prop:equivalence_emulation}.
\end{proof}

Throughout this work, we will therefore use the following corollary to swap the idempotent channels for their reduced versions or their associated block pinching maps to simplify our proofs.
\begin{corollary}\label{coro:equivalence_reduced_version}
    Let $\F : \cL(\cH_\F) \rightarrow \cL(\cH_\F)$ and $\cG : \cL(\cH_\cG) \rightarrow \cL(\cH_\cG)$ be two idempotent channels and $\widehat{\F}$, $\widehat{\cG}$ be their associated reduced channels, as defined in Proposition~\ref{prop:structure_idempotent_channel}, then
    \begin{equation}
    \deltaR{k}{n}{\cG}{\F} = \deltaR{k}{n}{\widehat{\cG}}{\widehat{\F}} = \deltaR{k}{n}{\Pi_{\mathbf{g}}}{\Pi_{\mathbf{f}}},
    \end{equation}
    so that
    \begin{equation}
        \CR{\cG}{\F} = \CR{\widehat{\cG}}{\widehat{\F}} = \CR{\Pi_{\mathbf{g}}}{\Pi_{\mathbf{f}}},
    \end{equation}
    where $\rmR$ denotes the assistance regime, i.e.\ $\rmR \in \{\UA,\SR,\EA\}$ and $\mathbf{f} = \lambda(\F)$, $\mathbf{g} = \lambda(\cG)$.
\end{corollary}

\begin{proof}
This follows directly from Propositions~\ref{prop:conversion_to_reduced_channel},~\ref{prop:idempotent_vs_pinching} and~\ref{prop:equivalence_emulation}.
\end{proof}

As mentioned in the introduction, the shape vector $\lambda(\F)$ of an idempotent channel $\F : \cL(\cH) \rightarrow \cL(\cH)$ can be computed efficiently, in time polynomial in the dimension of $\cH$.

\begin{prop}[{\cite[Theorem~2.6 and Algorithm~1]{Fawzi24}}]\label{prop:computation_shape}
Given $\F : \cL(\cH) \rightarrow \cL(\cH)$ an idempotent channel, there is an algorithm that computes $\lambda(\F)$ in time $\mathcal{O}(d^6\log(d))$ where $d$ is the dimension of $\cH$, using the algorithm and computational cost convention of~\cite{Fawzi24}.
\end{prop}

We now show that the shape vector of idempotent channels preserves the tensor product, a property that will be used in the proofs of our main results, Theorems~\ref{thm:main_UA_SR} and~\ref{thm:main_EA}. Indeed, it will allow us to reduce the study of the problem of emulating $k$ copies of $\F$ using $n$ copies of $\cG$ for general pairs of integers $k,n \in \N_{>0}$ to the case where $k=n$.

\begin{prop}\label{prop:multiplicativity}
    Let $\F_1 : \cL(\cH_1) \rightarrow \cL(\cH_1)$ and $\F_2 : \cL(\cH_2) \rightarrow \cL(\cH_2)$ be two idempotent channels. Then $\F_1 \otimes \F_2$ is idempotent and
    \begin{align}
        \Rg(\F_1^* \otimes \F_2^*) &= \Rg(\F_1^*) \otimes \Rg(\F_2^*),\label{eq:mult_range_UCP} \\
        \la(\F_1\otimes \F_2) &= \la(\F_1) \otimes \la(\F_2)\label{eq:mult_shape_vector},
    \end{align}
    where the last equality is understood up to a permutation of the entries of the shape vectors.
\end{prop}

\begin{proof}
That $\F_1\otimes\F_2$ is idempotent follows from $(\F_1\otimes\F_2)\circ(\F_1\otimes\F_2) = (\F_1\circ\F_1)\otimes(\F_2\circ\F_2) = \F_1\otimes\F_2$.

We first show that, for all linear maps $\Phi : \cL(\cH_1) \rightarrow \cL(\cH_1')$ and $\Psi : \cL(\cH_2) \rightarrow \cL(\cH_2')$,
\begin{equation}\label{eq:mult_range_channels}
    \Rg(\Phi \otimes \Psi) = \Rg(\Phi) \otimes \Rg(\Psi),
\end{equation}
which gives eq.~\eqref{eq:mult_range_UCP} by taking $\Phi = \F_1^*$ and $\Psi = \F_2^*$.
For the inclusion from right to left, $\Rg(\Phi) \otimes \Rg(\Psi)$ is spanned by the product elements $y_1 \otimes y_2$ with $y_1 = \Phi(x_1)$ and $y_2 = \Psi(x_2)$. 
Then $y_1 \otimes y_2 = (\Phi \otimes \Psi)(x_1 \otimes x_2) \in \Rg(\Phi\otimes\Psi)$.
Conversely, every $x \in \cL(\cH_1\otimes\cH_2) = \cL(\cH_1)\otimes\cL(\cH_2)$ can be written $x = \sum_{j \in [J]} x_1^{(j)} \otimes x_2^{(j)}$ with $x_i^{(j)} \in \cL(\cH_i)$, $i \in [2]$, so that
\begin{equation*}
    (\Phi \otimes \Psi)(x) = \sum_{j \in [J]} \Phi(x_1^{(j)}) \otimes \Psi(x_2^{(j)}) \in \Rg(\Phi) \otimes \Rg(\Psi).
\end{equation*}

It remains to prove eq.~\eqref{eq:mult_shape_vector}.
By Proposition~\ref{prop:structure_idempotent_channel}, for $i \in [2]$ there are integers $K_i$, dimensions $(d^{(i)}_k)_{k \in [K_i]}$ and states $(\rho^{(i)}_k)_{k \in [K_i]}$ such that $\la(\F_i) = (d^{(i)}_k)_{k \in [K_i]}$ up to the order of the entries and
\begin{equation*}
    \Rg(\F_i) = \Bigl(\bigoplus_{k \in [K_i]} \cL(\cH^{(i)}_{k,1}) \otimes \rho^{(i)}_k\Bigr)\oplus 0, \qquad d^{(i)}_k = \dim(\cH^{(i)}_{k,1}).
\end{equation*}
Applying eq.~\eqref{eq:mult_range_channels} to $\Phi = \F_1$ and $\Psi = \F_2$ and distributing the tensor product over the direct sums, we get
\begin{equation*}
    \Rg(\F_1\otimes\F_2) = \Bigl(\bigoplus_{(k,l) \in [K_1]\times[K_2]} \cL(\cH^{(1)}_{k,1}\otimes\cH^{(2)}_{l,1}) \otimes (\rho^{(1)}_k\otimes\rho^{(2)}_l)\Bigr)\oplus 0,
\end{equation*}
where we used $\cL(\cH)\otimes\cL(\cH') = \cL(\cH\otimes\cH')$ and the fact that $\rho^{(1)}_k\otimes\rho^{(2)}_l$ is a state.
This is a decomposition of $\Rg(\F_1\otimes\F_2)$ of the form of eq.~\eqref{eq:decompo_range_F}, indexed by the pairs $(k,l)$, with blocks of dimensions $\dim(\cH^{(1)}_{k,1}\otimes\cH^{(2)}_{l,1}) = d^{(1)}_kd^{(2)}_l$.
By the definition of the shape vector, $\la(\F_1\otimes\F_2)$ is the family $(d^{(1)}_kd^{(2)}_l)_{(k,l)}$ sorted in non-increasing order, that is $\la(\F_1)\otimes\la(\F_2)$.
This identification does not depend on the chosen decomposition: by eq.~\eqref{eq:range_F_star}, the shape vector of an idempotent channel is that of the $*$-algebra $\Rg(\widehat{\F}^*)$, which is unique up to the order of its entries, and, since $\widehat{\F_1\otimes\F_2} = \widehat{\F}_1\otimes\widehat{\F}_2$, eq.~\eqref{eq:mult_range_UCP} gives $\Rg((\widehat{\F_1\otimes\F_2})^*) = \Rg(\widehat{\F}_1^*)\otimes\Rg(\widehat{\F}_2^*)$, whose shape vector is $\la(\F_1)\otimes\la(\F_2)$.
\end{proof}

Hence, for $k \in \N_{> 0}$, $\lambda(\F^{\otimes k}) = \lambda(\F)^{\otimes k}$.
Moreover, for block pinchings, we directly have the stability under tensor products: identifying $\cH_{\mathbf{d}}\otimes\cH_{\mathbf{d}'}$ with $\cH_{\mathbf{d}\otimes\mathbf{d}'}$, we have $\Pi_{\mathbf{d}}\otimes\Pi_{\mathbf{d}'} = \Pi_{\mathbf{d}\otimes\mathbf{d}'}$, so that $\Pi_{\mathbf{f}}^{\otimes k} = \Pi_{\la(\F^{\otimes k})}$.

We end this section with the following remark on the origin of the concept of shape vectors.

\begin{remark}
Note that the shape of an idempotent channel $\F$ coincides with the shape of the $*$-algebra $\Rg(\widehat{\F}^*)$, with $\widehat{\F}$ the reduced channel associated to $\F$. The notion of shape vector for $*$-algebras was introduced by Kuperberg in~\cite{kuperberg}. If $\A$ is a $*$-algebra, we can decompose
\begin{equation*}
\A \cong \bigoplus_{k=1}^K \cL(\C^{d_k}),
\end{equation*}
for a set of integers $\{d_k\mid k \in [K]\}$ and we write $\la(\A)$ for the shape vector of $\A$, whose coordinates are the integers $(d_k)_{k \in [K]}$ sorted in non-increasing order\footnote{Throughout this article and in Kuperberg's work~\cite{kuperberg}, the order with respect to which the coordinates of shape vectors are sorted is actually irrelevant.}.
\end{remark}

\section{Unassisted emulation of idempotent channels with zero error}\label{sec:emul_capa}

In this section and the next one, we prove item~(i) of Theorem~\ref{thm:main_UA_SR}, that is, eq.~\eqref{eq:emul_capa}, together with the fact that, in both regimes, the emulation capacity coincides with the zero-error emulation capacity, i.e.\ with the supremum of the ratios $k/n$ for which $\F^{\otimes k}$ can be emulated exactly by $\cG^{\otimes n}$.
Throughout, we assume that $\la(\F) \neq (1)$: otherwise, $\F$ is essentially a replacer channel (see Proposition~\ref{prop:structure_idempotent_channel}), and it is simple to see that it can be emulated exactly by any other idempotent channel $\cG$ with an arbitrarily large rate.
When $\la(\F) = (1,\dots,1)$, i.e.\ when $\|\la(\F)\|_\infty = 1$, the denominator in eq.~\eqref{eq:emul_capa} vanishes at $p = +\infty$, and the term $p=+\infty$ is omitted from the infimum, as in the last row of Table~\ref{tab:lower_bounds}.
When $\F$ or $\cG$ is either the identity channel or the completely dephasing channel, it suffices to restrict the infimum to $p \in \{1,+\infty\}$ (see Table~\ref{tab:lower_bounds} below):
\begin{equation}
\CUA{\cG}{\F} = \CSR{\cG}{\F} = \min_{p \in \{1,+\infty\}} \frac{\log(\|\la(\cG)\|_p)}{\log(\|\la(\F)\|_p)}.
\end{equation}

On the right-hand side of eq.~\eqref{eq:emul_capa}, the two terms $\lambda(\F)$ and $\lambda(\cG)$ are the shape vectors of Proposition~\ref{prop:structure_idempotent_channel}.
A remarkable aspect of this result is that it is a single-letter formula, i.e.\ the expression does not involve a limit.
This comes from the fact that the norm of the shape vector is multiplicative under tensor product.

Now that we have presented the decomposition properties of idempotent channels, we move on to proving the achievability part of item~(i) of Theorem~\ref{thm:main_UA_SR} in this section.
Formally, we prove that every rate below 
\begin{equation}
\inf_{p \in [1,+\infty]} \frac{\log(\|\lambda(\cG)\|_p)}{\log(\|\lambda(\F)\|_p)}
\label{eq:capacity_formula}
\end{equation}
is achievable in the unassisted regime, and even with zero error.
This yields the bounds
\begin{equation}
\inf_{p \in [1,+\infty]} \frac{\log(\|\lambda(\cG)\|_p)}{\log(\|\lambda(\F)\|_p)}
 \leq \CUA{\cG}{\F} \leq \CSR{\cG}{\F}.
\end{equation}
 We show in Section~\ref{sec:bound_finite_blocklength_UA_SR} that any rate strictly greater than eq.~\eqref{eq:capacity_formula} is not achievable, as a consequence of item~(ii) of Theorem~\ref{thm:main_UA_SR}.

Before presenting these proofs, we provide some examples of the expression in eq.~\eqref{eq:capacity_formula}.
In particular, we gather in Table~\ref{tab:lower_bounds} the expression of the emulation capacities $\CUA{\cG}{\F} = \CSR{\cG}{\F}$ when $\F$ or $\cG$ is an identity quantum channel or a completely dephasing channel. 
In fact, when $\F = \Id_d$, then $\lambda(\F) = (d)$ and thus $\| \lambda(\F)\|_p = d$ for all $p \geq 1$. 
As $\| \lambda(\cG)\|_{p}$ is non-increasing in $p$, this shows that $\CUA{\cG}{\F} = \frac{\log(\|\lambda(\cG)\|_\infty)}{\log(d)}$. 
The other cases are similar. Note that, when $\cG = \Delta_d$, $\log(\|\lambda(\Delta_d)\|_\infty) = \log(1) = 0$. 
Therefore, in Table~\ref{tab:lower_bounds}, we split the case $\cG = \Delta_d$ into two subcases: the one where $\F$ is any idempotent channel such that at least one coordinate of $\lambda(\F)$ is greater than or equal to $2$, and the one where $\lambda(\F) = (1,\dots,1)$, in which case, for all $p \in [1,+\infty)$, 
\begin{equation*}
\frac{\log(\|\lambda(\Delta_d)\|_p)}{\log(\|\lambda(\F)\|_p)} = \frac{p\log(\|\lambda(\Delta_d)\|_p)}{p\log(\|\lambda(\F)\|_p)} = \frac{\log(d)}{\log(\|\lambda(\F)\|_1)}.
\end{equation*}
Note that in the first case, i.e.\ when $\la(\F)$ has a coefficient greater than or equal to two, the emulation capacity of $\F$ by the completely dephasing channel is zero, that is it is impossible to emulate a channel preserving some quantum information with a classical one, without error.

\begin{table}[h!]
    \centering
    \begin{tabular}{|c|c|c|}
        \hline
        $\F$ & $\cG$ & $\CUA{\cG}{\F} = \CSR{\cG}{\F}$  \rule{0pt}{13pt}\\
        \hline
        $\Id_d$ & $\cG$ & $\frac{\log(\|\lambda(\cG)\|_\infty)}{\log(d)}$ \rule{0pt}{11pt}\\[4pt]
        $\F$ & $\Id_d$ & $\frac{\log(d)}{\log(\|\lambda(\F)\|_1)}$ \\[4pt]
        $\Delta_d$ & $\cG$ & $\frac{\log(\|\lambda(\cG)\|_1)}{\log(d)}$ \\[4pt]
        $\F \text{ s.t. } \lambda(\F) \neq (1,\dots,1)$ & $\Delta_d$ & $0$ \\[4pt]
        $\F \text{ s.t. } \lambda(\F) = (1,\dots,1)$ & $\Delta_d$ & $\frac{\log(d)}{\log(\|\lambda(\F)\|_1)}$ \\
        \hline
    \end{tabular}
    \caption{Emulation capacities when either $\F$ or $\cG$ is the identity or the completely dephasing channel (with $d \geq 2$ and $\lambda(\F) \neq (1)$).}\label{tab:lower_bounds}
\end{table}

Then, we also note that, as a corollary of item~(i) of Theorem~\ref{thm:main_UA_SR}, the emulation capacity of idempotent channels is, in general, superadditive under tensor products of source channels.
However, since the infimum in eq.~\eqref{eq:emul_capa} collapses to a single value of $p$ when the target channel is either a quantum or a classical identity, the quantum and classical capacities of idempotent channels are additive.

\begin{corollary}\label{coro:additivity}
    Let $\F$, $\cG_1$ and $\cG_2$ be three idempotent channels, and $\rmR \in \{\UA,\SR\}$, then 
    \begin{equation}
        \CR{\cG_1 \otimes \cG_2}{\F} \geq \CR{\cG_1}{\F} + \CR{\cG_2}{\F}.
    \end{equation}
    However, when $\F$ is either the quantum or the classical identity, i.e.\ $\F \in \{\Id_d, \Delta_d\}$, then the emulation capacity becomes additive:
    \begin{equation}
        \CR{\cG_1 \otimes \cG_2}{\F} = \CR{\cG_1}{\F} + \CR{\cG_2}{\F}.
    \end{equation}
\end{corollary}

\begin{proof}
If $\lambda(\F)=(1)$, all the capacities in the statement are infinite, so the claims are immediate.
We therefore assume that $\lambda(\F)\neq(1)$.
By item~(i) of Theorem~\ref{thm:main_UA_SR}, we have, for $\rmR \in \{\UA,\SR\}$, 
\begin{align*}
    \CR{\cG_1 \otimes \cG_2}{\F} &= \inf_{p \in [1,+\infty]} \frac{\log(\|\la(\cG_1 \otimes \cG_2)\|_p)}{\log(\|\la(\F)\|_p)} \\
    &\geq \inf_{p \in [1,+\infty]} \frac{\log(\|\la(\cG_1)\|_p)}{\log(\|\la(\F)\|_p)} + \inf_{p \in [1,+\infty]} \frac{\log(\|\la(\cG_2)\|_p)}{\log(\|\la(\F)\|_p)} \\
    &= \CR{\cG_1}{\F} + \CR{\cG_2}{\F},
\end{align*}
where we used Proposition~\ref{prop:multiplicativity} and the multiplicativity of the $\ell_p$-norms over the tensor product in the inequality.

However, when $\F = \Id_d$, as per Table~\ref{tab:lower_bounds},
\begin{align*}
    \CR{\cG_1 \otimes \cG_2}{\Id_d} &= \frac{\log(\|\lambda(\cG_1 \otimes \cG_2)\|_\infty)}{\log(d)} \\
    &= \frac{\log(\|\lambda(\cG_1)\|_\infty)}{\log(d)} + \frac{\log(\|\lambda(\cG_2)\|_\infty)}{\log(d)} \\
    &= \CR{\cG_1}{\Id_d} + \CR{\cG_2}{\Id_d}.
\end{align*}
When $\F = \Delta_d$, the same line of argument holds replacing the $\ell_\infty$ norm by the $\ell_1$ norm.
\end{proof}

The following example shows that the superadditivity of Corollary~\ref{coro:additivity} can be strict in general.

\begin{example}[Complementary source channels]\label{ex:UA_complementary}
Let $\cG_1 = \Id_2$, $\cG_2 = \Delta_2$ and $\F = \cG_1 \otimes \cG_2$.
By Proposition~\ref{prop:multiplicativity},
\begin{equation*}
    \la(\cG_1) = (2), \qquad \la(\cG_2) = (1,1), \qquad \la(\F) = \la(\cG_1)\otimes\la(\cG_2) = (2,2).
\end{equation*}
By the second row of Table~\ref{tab:lower_bounds}, $\CUA{\cG_1}{\F} = \log(2)/\log(\|(2,2)\|_1) = 1/2$, and by its fourth row, $\CUA{\cG_2}{\F} = 0$, since $\la(\F)$ has an entry equal to $2$.
However, $\la(\cG_1\otimes\cG_2) = \la(\F)$, so that $\CUA{\cG_1\otimes\cG_2}{\F} = 1$.
Hence the inequality of Corollary~\ref{coro:additivity} is strict here:
\begin{equation*}
    \CUA{\cG_1\otimes\cG_2}{\F} = 1 > \frac{1}{2} = \CUA{\cG_1}{\F} + \CUA{\cG_2}{\F}.
\end{equation*}
The two source channels are complementary: $\cG_1$ carries the quantum part and $\cG_2$ the classical part of the hybrid target channel $\F$, and neither of them can emulate the target efficiently without entanglement.
Indeed, without entanglement, $\cG_1$ can only reach the rate $1/2$, while $\cG_2$ is even useless, as it cannot emulate $\F$ at any positive rate; it is only in combination with $\cG_1$ that it contributes.
We will see in Example~\ref{ex:EA_complementary} that shared entanglement restores additivity for these channels.
\end{example}

The following example shows that the emulation capacity is not \emph{reversible} in the unassisted and shared-randomness regimes, i.e.\ that $\CR{\cG}{\F} \CR{\F}{\cG} \neq 1$ in general, for $\rmR \in \{\UA,\SR\}$, and that the failure of reversibility can be as extreme as possible.

\begin{example}[Emulation of a quantum channel by a classical one]\label{ex:UA_teleportation}
    Let $\F = \Id_2$ be the identity channel on $\cL(\C^2)$ and $\cG = \Delta_2$ be the completely dephasing channel on $\cL(\C^2)$, so that $\la(\F) = (2)$ and $\la(\cG) = (1,1)$.
    Since $\la(\F)$ has a coordinate equal to $2$, the fourth row of Table~\ref{tab:lower_bounds} gives
    \begin{equation*}
    \CUA{\Delta_2}{\Id_2} = \CSR{\Delta_2}{\Id_2} = 0,
    \end{equation*}
    that is, a noiseless qubit channel cannot be emulated with zero error by a classical channel, even with shared randomness and at an arbitrarily small rate.
    In the other direction, by the second row of Table~\ref{tab:lower_bounds},
    \begin{equation*}
    \CUA{\Id_2}{\Delta_2} = \frac{\log(2)}{\log(\|\la(\Delta_2)\|_1)} = \frac{\log(2)}{\log(2)} = 1,
    \end{equation*}
    since one use of a noiseless qubit channel transmits one bit.
    Therefore
    \begin{equation*}
    \CUA{\Delta_2}{\Id_2}\,\CUA{\Id_2}{\Delta_2} = 0,
    \end{equation*}
    which is the most extreme failure of reversibility possible: not only is $\CR{\F}{\cG}$ different from $\CR{\cG}{\F}^{-1}$, but $\CR{\cG}{\F}^{-1}$ is not even defined, as $\CR{\cG}{\F} = 0$.
    We will see in Example~\ref{ex:EA_teleportation} that shared entanglement restores reversibility here, the two capacities becoming $1/2$ and $2$.
\end{example}

Although the infimum in eq.~\eqref{eq:emul_capa} is reached for $p \in \{1,+\infty\}$ when either $\F$ or $\cG$ is an identity channel or a completely dephasing one, we show in the following example that this is not the case in general.

In particular, the fact that the infimum in eq.~\eqref{eq:emul_capa} can be reached for $p \not \in \{1,+\infty\}$ implies that, for general idempotent channels $\F$ and $\cG$, the chain of emulation $\cG \mapsto \Id_d \mapsto \F$ (where each arrow represents an emulation) is less efficient than directly emulating $\F$ by $\cG$, even when shared randomness is allowed.

\begin{example}
Let $\F : \cL(\C^{8}) \rightarrow \cL(\C^8)$ and $\cG : \cL(\C^{15}) \rightarrow \cL(\C^{15})$ be two idempotent channels such that $\la(\F) = (5,3)$ and $\la(\cG) = (10,3,1,1)$. Numerically, we find that the infimum 
\begin{equation*}
\CUA{\cG}{\F} = \inf_{p \in [1,+\infty]} \frac{\log(\|\la(\cG)\|_p)}{\log(\|\la(\F)\|_p)} \approx 1.29916,
\end{equation*}
is reached for $p \approx 1.15401$. In contrast, the infimum
\begin{equation*}
\CUA{\F}{\cG} = \inf_{p\in [1,+\infty]} \frac{\log(\|\la(\F)\|_p)}{\log(\|\la(\cG)\|_p)} = \frac{\log(5)}{\log(10)},
\end{equation*}
is reached for $p = +\infty$. 
\end{example}

\subsection{Achievability in item~(i) of Theorem~\ref{thm:main_UA_SR}}

To prove the achievability direction of item~(i) of Theorem~\ref{thm:main_UA_SR}, we first show the following proposition in the case where only one copy of $\F$ is emulated for each copy of $\cG$ used; we will then use Proposition~\ref{prop:multiplicativity} on the compatibility of the shape vector with the tensor product to conclude.

\begin{prop}\label{prop:achievability_arxiv}
Let $\F$ and $\cG$ be two idempotent channels such that, for all $p \in [1,+\infty]$,
\begin{equation*}\|\la(\F)\|_p < \|\la(\cG)\|_p.\end{equation*}
Then, there exists $n \in \N_{>0}$, $\E$ and $\D$ two channels such that 
\begin{equation*}\F^{\otimes n} = \D \cG^{\otimes n} \E.\end{equation*}
\end{prop}

To prove this proposition, we use the following theorem on embeddability of $*$-algebras.

\begin{theorem}[Proven in \cite{kuperberg}]\label{thm:kuperberg_arxiv}
If $\A$ and $\B$ are two finite-dimensional $*$-algebras such that, for all $p \in [1,+\infty]$,
\begin{equation}
\|\la(\A)\|_p < \|\la(\B)\|_p,
\end{equation}
then, there is $n \in \N_{>0}$ such that there exists a subunital injective $*$-homomorphism $\iota : \A^{\otimes n} \rightarrow \B^{\otimes n}$.
\end{theorem}

\begin{proof}
This theorem, despite not being stated explicitly in~\cite{kuperberg}, is proven in the proof of~\cite[Theorem~1.1]{kuperberg}. 
\end{proof}

To connect the problem of embedding $*$-algebras with that of interconverting idempotent channels, we use tracial conditional expectations.
Given $\A \subseteq \cL(\cH_\A)$ a $*$-subalgebra of $\cL(\cH_\A)$ sharing its identity operator with $\cL(\cH_\A)$, we can write the orthogonal direct sum decompositions
\begin{equation*}
\cH_\A = \bigoplus_{k \in [K]} \cH_{k,1}\otimes \cH_{k,2}, \quad \A = \bigoplus_{k \in [K]} \cL(\cH_{k,1}) \otimes \mathds{1}_{m_k}, 
\end{equation*}
with $m_k = \dim(\cH_{k,2})$ for $k \in [K]$, and $P_k$ the orthogonal projector onto $\cH_{k,1}\otimes\cH_{k,2}$.

\begin{definition}[Tracial conditional expectation, see e.g.~{\cite[Proposition~1.5]{Wolf12}}]\label{def:tracial_cond_exp}
With the notation above, the \emph{tracial conditional expectation onto $\A$} is the map
\begin{equation}\label{eq:self_adjoint_conditional_expectation}
    E_\A : \cL(\cH_\A) \rightarrow \A, \qquad E_\A(x) = \sum_{k \in [K]} \tr_{k,2}(P_kxP_k) \otimes \frac{1}{m_k}\mathds{1}_{m_k}.
\end{equation}
It is unital, completely positive, and satisfies $E_\A(x) = x$ for every $x \in \A$.

Given an idempotent channel $\F : \cL(\cH_\F) \rightarrow \cL(\cH_\F)$, we write $E_\F$ for the tracial conditional expectation onto the $*$-algebra $\Rg(\widehat{\F}^*)$, where $\widehat{\F} : \cL(\cH_0) \rightarrow \cL(\cH_0)$ is the reduced channel associated to $\F$.
This is well-defined: by Proposition~\ref{prop:structure_idempotent_channel}, $\Rg(\widehat{\F}^*)$ is a unital $*$-subalgebra of $\cL(\cH_0)$ which shares its identity operator with $\cL(\cH_0)$.
\end{definition}

The tracial conditional expectation onto $\A$ is moreover an idempotent channel whose shape vector is that of $\A$.

\begin{prop}\label{prop:cond_exp_idempotent}
Let $\A \subseteq \cL(\cH_\A)$ be a unital $*$-subalgebra whose identity operator coincides with that of $\cL(\cH_\A)$.
Then $E_\A$ is an idempotent channel and $\la(E_\A) = \la(\A)$.
In particular, for every idempotent channel $\F$, the map $E_\F$ is an idempotent channel with $\la(E_\F) = \la(\F)$.
\end{prop}

\begin{proof}
$E_\A$ is completely positive and trace preserving by eq.~\eqref{eq:self_adjoint_conditional_expectation}, and $E_\A(x) = x$ for every $x \in \A$, so that $E_\A\circ E_\A = E_\A$ since $\Rg(E_\A) \subseteq \A$.
Moreover, with the notation above,
\begin{equation*}
    \Rg(E_\A) = \bigoplus_{k \in [K]}\cL(\cH_{k,1})\otimes\frac{\mathds{1}_{m_k}}{m_k},
\end{equation*}
which is a decomposition of the form of eq.~\eqref{eq:decompo_range_F} with the states $\rho_k = \mathds{1}_{m_k}/m_k$ and blocks of dimensions $\dim(\cH_{k,1})$.
Hence $\la(E_\A) = (\dim(\cH_{k,1}))_{k \in [K]} = \la(\A)$.
The last statement follows by applying this to $\A = \Rg(\widehat{\F}^*)$, whose shape vector is $\la(\F)$ by eq.~\eqref{eq:range_F_star} and by the definition of the shape vector of an idempotent channel.
\end{proof}

Since $\la(E_\F) = \la(\F)$ and $\la(E_\cG) = \la(\cG)$ for all idempotent channels $\F$ and $\cG$, by Proposition~\ref{prop:emulation_same_shape} we can interconvert exactly $\F$ with $E_{\F}$, $\cG$ with $E_{\cG}$ and vice versa, so that $\F = \D\cG\E$ for channels $\D,\E$ if and only if $E_\F = \D'E_\cG\E'$.
We may therefore replace $\F$ and $\cG$ by the tracial conditional expectations $E_\F$ and $E_\cG$, whose ranges are the $*$-algebras on which Theorem~\ref{thm:kuperberg_arxiv} operates directly.

We now lift the $*$-homomorphism $\iota$ of Theorem~\ref{thm:kuperberg_arxiv} into a unital completely positive map. 
\begin{lemma}\label{lemma:lift_homomorphism}
    Let  $\A \subseteq \cL(\cH_\A)$, $\B \subseteq \cL(\cH_\B)$ be two finite-dimensional unital $*$-subalgebras of $\cL(\cH_\A)$ and $\cL(\cH_\B)$ respectively, where we assume that the identity operator in $\A$ coincides with the one in $\cL(\cH_\A)$, denoted $\mathds{1}_\A$ and the one in $\B$ coincides with the one in $\cL(\cH_\B)$, denoted $\mathds{1}_\B$. Let $\iota : \A \rightarrow \B$ be an injective subunital $*$-homomorphism.
    
    Then, writing $\cH_\B^{(0)} = \supp(\iota(\mathds{1}_\A)) \subseteq \cH_\B$, we have that $P_\iota = \iota(\mathds{1}_\A) \in \cL(\cH_{\B})$ is the orthogonal projector onto $\cH_\B^{(0)}$. We denote by $V : \cH_\B^{(0)} \rightarrow \cH_\B$ the isometry embedding $\cH_\B^{(0)}$ into $\cH_\B$ such that $VV^* = P_\iota$ and we define the \emph{reduced $*$-homomorphism} $\widehat{\iota} : \A \rightarrow \cL(\cH_\B^{(0)})$ by 
    \begin{equation}
        \widehat{\iota} : x \mapsto V^*\iota(x)V.
    \end{equation}
    Then $\widehat{\iota} : \A \rightarrow \cL(\cH_\B^{(0)})$ is a \emph{unital} injective $*$-homomorphism and $\widehat{\iota}(\A)$ is a unital $*$-subalgebra of $\cL(\cH_\B^{(0)})$ where both $\widehat{\iota}(\A)$ and $\cL(\cH_{\B}^{(0)})$ share the same identity operator. We will write $\widehat{\iota}^{-1} : \widehat{\iota}(\A) \rightarrow \A$ for the inverse of $\widehat{\iota}$, which is itself a unital injective $*$-homomorphism.

    We write $E_\A : \cL(\cH_\A) \rightarrow \A$ and $E_{\widehat{\iota}(\A)} : \cL(\cH_\B^{(0)}) \rightarrow \widehat{\iota}(\A)$ for the tracial conditional expectations on $\A$ and $\widehat{\iota}(\A)$ respectively. Then, the maps
    \begin{equation}\label{eq:lift_homomorphism}
        \begin{aligned}
            \widetilde{\iota} : \cL(\cH_\A) &\rightarrow \B \subseteq \cL(\cH_\B) \\
            x &\mapsto \iota(E_\A(x)) + \frac{\tr(x)}{\tr(\mathds{1}_\A)}(\mathds{1}_\B - P_\iota),
        \end{aligned}
    \end{equation}
    and 
    \begin{equation}
        \begin{aligned}
            \widetilde{\iota}^{-1} : \cL(\cH_\B) &\rightarrow \A \subseteq \cL(\cH_\A) \\
            x &\mapsto \widehat{\iota}^{-1}(E_{\widehat{\iota}(\A)}(V^*xV))
        \end{aligned}
    \end{equation}
     are unital completely positive and satisfy for all $x \in \A$, 
\begin{equation}\label{eq:inverse_channel}
         \widetilde{\iota}^{-1} \circ \widetilde{\iota}(x) = x.
     \end{equation}
    \end{lemma}

    \begin{proof}
        First, we prove that $P_\iota = \iota(\mathds{1}_\A)$ is an orthogonal projector on $\supp(\iota(\mathds{1}_\A))$. As $\iota$ is a $*$-homomorphism, it preserves both the operation of taking the adjoint and the multiplication, so that $P_\iota^* = \iota(\mathds{1}_\A)^* = \iota(\mathds{1}_\A^*) = P_\iota$ and $P_\iota^2 = \iota(\mathds{1}_\A)^2 = \iota(\mathds{1}_\A^2) = P_\iota$. 

        We move on to prove that $\widehat{\iota} : \A \rightarrow \cL(\cH_\B^{(0)})$, with $\cH_\B^{(0)} = \supp(\iota(\mathds{1}_\A))$ is a $*$-homomorphism, which is also unital and injective. First $\widehat{\iota}$ is clearly linear. Let $x \in \A$, \begin{equation*}\widehat{\iota}(x)^* = (V^*\iota(x)V)^* = V^*\iota(x)^*V = V^*\iota(x^*)V = \widehat{\iota}(x^*),\end{equation*}
        and for $y \in \A$, 
        \begin{equation*}\widehat{\iota}(xy) = V^*\iota(xy)V = V^*\iota(x)\iota(y)V  \overset{(a)}{=} V^*VV^*\iota(x)VV^*\iota(y)V  =  \widehat{\iota}(x)\widehat{\iota}(y).\end{equation*}
        The equality $(a)$ holds by applying Lemma~\ref{lemma:stability_support_one} on the map $\Phi = \iota$ which is positive since it is a $*$-homomorphism\footnote{Actually, $*$-homomorphisms are even completely positive~\cite{Dixmier69}.}. To show the unitality, we have
        \begin{equation*}
        \widehat{\iota}(\mathds{1}_\A) = V^*\iota(\mathds{1}_\A)V = V^*P_\iota V = V^*VV^*V = \mathds{1}_{\cH_\B^{(0)}}, 
        \end{equation*}
        with $\mathds{1}_{\cH_\B^{(0)}}$ the identity operator in $\cL(\cH_\B^{(0)})$. We finish with the injectivity property. Let $x \in \ker(\widehat{\iota})$,
        \begin{equation*}
        0 = \widehat{\iota}(x) = V^*\iota(x)V.
        \end{equation*}
        Hence, by Lemma~\ref{lemma:stability_support_one}
        \begin{equation*}
        \iota(x) = VV^*\iota(x)VV^* = 0,
        \end{equation*}
        so that $x \in \ker(\iota)$, thus $x = 0$ by injectivity of $\iota$.

        Since $\widehat{\iota}$ is a unital $*$-homomorphism, $\widehat{\iota}(\A)$ is a $*$-subalgebra of $\cL(\cH_\B^{(0)})$ and has the same identity operator as $\cL(\cH_\B^{(0)})$, so that the tracial conditional expectation $E_{\widehat{\iota}(\A)}$ is well-defined by eq.~\eqref{eq:self_adjoint_conditional_expectation}. Furthermore, as $\widehat{\iota}$ is injective, $\widehat{\iota}^{-1}$ is well-defined and is a unital $*$-homomorphism as well. 

        Now, we show that $\widetilde{\iota}$ and $\widetilde{\iota}^{-1}$ are unital completely positive maps. By unitality of $E_\A$, $E_{\widehat{\iota}(\A)}$ and $\widehat{\iota}^{-1}$, it is clear that $\widetilde{\iota}$ and $\widetilde{\iota}^{-1}$ are unital. Furthermore, $\widetilde{\iota}^{-1}$ is the composition of completely positive maps and is thus completely positive. Moreover, as $\iota$ was assumed to be subunital, $\iota(\mathds{1}_\A) = P_\iota \leq \mathds{1}_\B$ so that $\widetilde{\iota}$ is a sum of two completely positive maps and is thus completely positive itself.

        We also have that $\widetilde{\iota}(\cL(\cH_\A)) \subseteq \B$. Indeed, as $\B$ is unital of unit $\mathds{1}_\B$ and as the range of $\iota$ is in $\B$, $\mathds{1}_\B - P_\iota \in \B$ and $\iota(E_\A(x)) \in \B$ for all $x \in \cL(\cH_\A)$. 

        Finally, note that $\mathds{1}_\B - P_\iota$ is the orthogonal projector on $\ker(P_\iota)$ which is orthogonal to $\supp(P_\iota) = \Rg(V)$, so that $(\mathds{1}_\B - P_\iota)V = 0$. Hence, for $x \in \A$,
        \begin{equation*}\widetilde{\iota}^{-1}(\widetilde{\iota}(x)) = \widehat{\iota}^{-1}(E_{\widehat{\iota}(\A)}(V^*\iota(E_\A(x))V)) = \widehat{\iota}^{-1}(E_{\widehat{\iota}(\A)}(\widehat{\iota}(x))) = \widehat{\iota}^{-1}(\widehat{\iota}(x)) = x.\end{equation*}
    \end{proof}

We now prove Proposition~\ref{prop:achievability_arxiv}.

\begin{proof}[Proof of Proposition~\ref{prop:achievability_arxiv}]
By Proposition~\ref{prop:emulation_same_shape}, since $\la(\F) = \la(E_\F)$ and $\la(\cG) = \la(E_\cG)$, it suffices to show that there are channels $\E',\D'$ and an integer $n \in \N_{>0}$ such that $E_\F^{\otimes n} = \D'E_\cG^{\otimes n} \E'$.

Let $\A = \Rg(E_\F)$ and $\B = \Rg(E_\cG)$ respectively seen as unital $*$-subalgebras of $\cL(\cH_{\widehat{\F}})$ and $\cL(\cH_{\widehat \cG})$, where $\cH_{\widehat{\F}} = \supp(\F(\mathds{1}))$ and $\cH_{\widehat{\cG}} = \supp(\cG(\mathds{1}))$ are the spaces on which the reduced versions of $\F$ and $\cG$, respectively $\widehat \F$ and $\widehat \cG$, act.
By assumption, $\|\la(\A)\|_p < \|\la(\B)\|_p$ for all $p \in [1,+\infty]$, so that Theorem~\ref{thm:kuperberg_arxiv} provides an $n \in \N_{>0}$ and an injective subunital $*$-homomorphism $\iota : \A^{\otimes n} \rightarrow \B^{\otimes n}$.
Let $\widetilde\iota$ and $\widetilde\iota^{-1}$ be the unital completely positive maps associated to $\iota$ by Lemma~\ref{lemma:lift_homomorphism}, and set
\begin{equation*}
    \E = (\widetilde{\iota}^{-1})^{*} \circ E_\A^{\otimes n}, \qquad \D = E_\A^{\otimes n}\circ \widetilde{\iota}^{\,*},
\end{equation*}
which are channels, as compositions of channels: the adjoint of a unital completely positive map is completely positive and trace preserving, and $E_\A^{\otimes n}$ is a channel.

It remains to check that $\D\circ E_\B^{\otimes n}\circ\E = E_\A^{\otimes n}$.
Taking adjoints, this is equivalent to
\begin{equation}\label{eq:factorisation_cond_exp}
    E_\A^{\otimes n} = \E^* \circ E_\B^{\otimes n} \circ \D^* = E_\A^{\otimes n}\circ\widetilde{\iota}^{-1}\circ E_\B^{\otimes n}\circ\widetilde{\iota}\circ E_\A^{\otimes n},
\end{equation}
where we used that $E_\A$ and $E_\B$ are self-adjoint, being tracial conditional expectations.
Let $x \in \cL(\cH_\A^{\otimes n})$.
Then $E_\A^{\otimes n}(x) \in \A^{\otimes n}$, so that $\widetilde{\iota}(E_\A^{\otimes n}(x)) \in \B^{\otimes n}$ by Lemma~\ref{lemma:lift_homomorphism}.
Since $E_\B^{\otimes n}$ acts as the identity on $\B^{\otimes n}$, we get $E_\B^{\otimes n}(\widetilde{\iota}(E_\A^{\otimes n}(x))) = \widetilde{\iota}(E_\A^{\otimes n}(x))$.
Applying eq.~\eqref{eq:inverse_channel} to $E_\A^{\otimes n}(x) \in \A^{\otimes n}$ gives $\widetilde{\iota}^{-1}(\widetilde{\iota}(E_\A^{\otimes n}(x))) = E_\A^{\otimes n}(x)$, and a final application of $E_\A^{\otimes n}$, which fixes $\A^{\otimes n}$, yields eq.~\eqref{eq:factorisation_cond_exp}.
\end{proof}

Then we can prove the achievability in item~(i) of Theorem~\ref{thm:main_UA_SR}.

\begin{proof}[Proof of achievability in item~(i) of Theorem~\ref{thm:main_UA_SR}]

Let $\F$, $\cG$ be idempotent channels.
If the infimum in eq.~\eqref{eq:emul_capa} is zero, the desired lower bound follows from the non-negativity of the capacity. 
We therefore assume that this infimum is positive.
Let a rate $r \geq 0$ be such that
\begin{equation*}
r < \inf_{p\in [1,+\infty]} \frac{\log(\|\la(\cG)\|_p)}{\log(\|\la(\F)\|_p)}.
\end{equation*}
By density of the rational numbers in the reals, there exist $k,n \in \N_{>0}$ such that 
\begin{equation}\label{eq:upper_bound_rate_achievable}
r \leq \frac{k}{n} < \inf_{p \in [1,+\infty]} \frac{\log(\|\la(\cG)\|_p)}{\log(\|\la(\F)\|_p)}.
\end{equation}
If both shape vectors have only unit entries, that is both $\F$ and $\cG$ are classical, then $\|\lambda(\F)\|_1^k < \|\lambda(\cG)\|_1^n$, so embedding the corresponding classical alphabets gives $\F^{\otimes k} = \D\cG^{\otimes n}\E$ for suitable channels $\E$ and $\D$.
Suppose that $\F$ is not a classical channel, so that $\log(\|\la(\F)\|_p) > 0$ for all $p\in [1,+\infty]$.
Then, for every $p \in [1,+\infty]$, we have
\begin{equation*}
    \|\la(\F^{\otimes k})\|_p < \|\la(\cG^{\otimes n})\|_p,
\end{equation*}
where we used Proposition~\ref{prop:multiplicativity} on the shape vectors of the tensor product of channels. 
If $\F$ is classical and $\cG$ is not, we still have for all $p \in [1,+\infty)$, by eq.~\eqref{eq:upper_bound_rate_achievable}, $\|\la(\F^{\otimes k})\|_p < \|\la(\cG^{\otimes n})\|_p$.
Then, for $p=+\infty$, $\|\la(\F^{\otimes k})\|_{\infty} = 1 < 2 \leq \|\lambda(\cG^{\otimes n})\|_\infty$.
Hence, we still have in this case $\|\la(\F^{\otimes k})\|_p < \|\la(\cG^{\otimes n})\|_p$ for all $p \in [1,+\infty]$.

Then, by Proposition~\ref{prop:achievability_arxiv}, there is $m \in \N_{>0}$, $\E$ and $\D$ channels such that 
\begin{equation*}
\F^{\otimes km} = \D \cG^{\otimes nm}\E,
\end{equation*}
so that $\F^{\otimes km}$ can be emulated exactly by $\cG^{\otimes nm}$ without assistance.
By Section~\ref{subsec:emulation}, $km/(nm) = k/n \geq r$ is thus an achievable rate, and the zero-error unassisted emulation capacity is at least $k/n \geq r$.
\end{proof}

\section{Minimal error for unassisted and shared-randomness-assisted emulation for finite blocklength}\label{sec:bound_finite_blocklength_UA_SR}
In this section, we prove item~(ii) of Theorem~\ref{thm:main_UA_SR}, that is, the lower bound of eq.~\eqref{eq:bound_error_approx_case} on the minimal error for the emulation of idempotent channels at finite blocklength, both in the unassisted and in the shared randomness regimes, which holds for any two idempotent channels $\F$ and $\cG$ and all $k,n \in \N_{>0}$.
We then deduce from it the converse part of item~(i) and the strong converse of item~(iii).

Since, in general, $\deltaUA{k}{n}{\cG}{\F} \geq \deltaSR{k}{n}{\cG}{\F}$, it actually suffices to show the right-most inequality in eq.~\eqref{eq:bound_error_approx_case}, that is, for all $\F$ and $\cG$ idempotent channels and every pair $k,n \in \N_{>0}$,
\begin{equation}\label{eq:target_lower_bound}
    \deltaSR{k}{n}{\cG}{\F} \geq 1 - \infnormratio{\cG}{\F}{n}{k}.
\end{equation}

We first reduce the problem to the emulation of block pinchings.
Let $\mathbf{f} = \la(\F)$ and $\mathbf{g} = \la(\cG)$ be the shape vectors of $\F$ and $\cG$.
We write $\Pi_{\mathbf{f}}$ and $\Pi_{\mathbf{g}}$ for the pinchings of shape vectors $\la(\Pi_{\mathbf{f}}) = \la(\F)$ and $\la(\Pi_{\mathbf{g}}) = \la(\cG)$, as per Definition~\ref{def:block_pinching}, so that Corollary~\ref{coro:only_shape_vectors} gives, in every regime $\rmR \in \{\UA,\SR,\EA\}$ and for all $k,n \in \N_{>0}$,
\begin{equation}\label{eq:reduction_to_pinchings}
    \deltaR{k}{n}{\cG}{\F} = \deltaR{k}{n}{\Pi_{\mathbf{g}}}{\Pi_{\mathbf{f}}}.
\end{equation}
It is therefore enough to prove eq.~\eqref{eq:target_lower_bound} for the block pinchings $\Pi_{\mathbf{f}}$ and $\Pi_{\mathbf{g}}$, which have the same shape vectors as $\F$ and $\cG$.
The gain is that $\Pi_{\mathbf{f}}$ and $\Pi_{\mathbf{g}}$ act on the spaces $\cH_{\mathbf{f}}$ and $\cH_{\mathbf{g}}$ of eq.~\eqref{eq:def_H_d}, which carry no multiplicity spaces: every quantity below is expressed directly in terms of the entries of $\lambda(\F) = \mathbf{f}$ and $\lambda(\cG) = \mathbf{g}$.

To prove this bound, we use the Holevo--Helstrom bound, applied to the problem of channel discrimination.

\begin{lemma}[see e.g.~{\cite[Theorem~3.52]{watrous18}}]\label{lemma:channel_Holevo}
    Let $\Phi_1,~\Phi_2 : \cL(\cH) \rightarrow \cL(\cH')$ be two quantum channels. For any choice of auxiliary system $\cL(\cH'')$ of dimension $d = \dim(\cH'')$, for any choice of positive semidefinite operator $ \mu \in \cL(\cH'\otimes\cH'')$ satisfying $\mu \leq \mathds{1}$ and of density operator $\sigma \in \D(\cH \otimes \cH'')$, 
    \begin{equation}\label{eq:uniform_Holevo}
    \frac{1}{2} \|\Phi_1 - \Phi_2\|_\diamond \geq \tr(\mu \Phi_1 \otimes \Id_d(\sigma)) - \tr(\mu \Phi_2 \otimes \Id_d(\sigma)).
    \end{equation}
\end{lemma}

Let $\Lambda \in \N_{>0}$, let $(p_\lambda)_{\lambda \in [\Lambda]}$ be a probability distribution and let $\{(\D_\lambda,\E_\lambda) \mid \lambda \in [\Lambda]\}$ be a set of pairs of encoding and decoding channels for the emulation of $\Pi_{\mathbf{f}}$ by $\Pi_{\mathbf{g}}$, with $\E_\lambda : \cL(\cH_{\mathbf{f}}) \to \cL(\cH_{\mathbf{g}})$ and $\D_\lambda : \cL(\cH_{\mathbf{g}}) \to \cL(\cH_{\mathbf{f}})$.
For $\lambda \in [\Lambda]$, we write
\begin{equation}\label{eq:N-def}
    \cN_\lambda = \D_\lambda\circ\Pi_{\mathbf{g}}\circ\E_\lambda \; : \; \cL(\cH_{\mathbf{f}}) \to \cL(\cH_{\mathbf{f}}),
\end{equation}
which is a channel, as a composition of channels, so that $\sum_{\lambda \in [\Lambda]}p_\lambda\cN_\lambda$ is a channel as well.
Applying Lemma~\ref{lemma:channel_Holevo} to $\Pi_{\mathbf{f}}$ and to $\sum_{\lambda \in [\Lambda]}p_\lambda\cN_\lambda$, the proof of item~(ii) of Theorem~\ref{thm:main_UA_SR} then proceeds by constructing ansätze $\mu,\sigma^{(p)}$, for $p \in [1,+\infty)$, such that
\begin{equation}
    \frac{1}{2}\Bigl\|\Pi_{\mathbf{f}} - \sum_{\lambda \in [\Lambda]}p_\lambda\D_\lambda\Pi_{\mathbf{g}}\E_\lambda\Bigr\|_\diamond \geq \tr(\mu\,\Pi_{\mathbf{f}} \otimes \Id_d(\sigma^{(p)})) - \sum_{\lambda \in [\Lambda]}p_\lambda\tr(\mu\,\cN_\lambda \otimes \Id_d(\sigma^{(p)})),
\end{equation}
with
\begin{equation}\label{eq:goal_ansatz}
    \tr(\mu\,\Pi_{\mathbf{f}} \otimes \Id_d(\sigma^{(p)})) = 1, \qquad \tr(\mu\,\cN_\lambda \otimes \Id_d(\sigma^{(p)})) \leq \frac{\|\la(\cG)\|_p}{\|\la(\F)\|_p},
\end{equation}
for every $\lambda \in [\Lambda]$ and $p \in [1,+\infty)$.

The only structural property of the maps $\cN_\lambda$ that we shall use is the following Kraus representation, which exhibits the block projectors $\widetilde{P}'_b$ of $\cH_{\mathbf{g}}$.

\begin{lemma}\label{lem:concentration}
    For $\lambda \in [\Lambda]$, the channel $\cN_\lambda$ of eq.~\eqref{eq:N-def} admits the Kraus representation
    \begin{equation}\label{eq:Kraus}
        \cN_\lambda(z) = \sum_{a,b,c}\bigl(D^{(\lambda)}_a\widetilde P'_bE^{(\lambda)}_c\bigr)\,z\,\bigl(D^{(\lambda)}_a\widetilde P'_bE^{(\lambda)}_c\bigr)^*,
    \end{equation}
    where $\{\widetilde P'_b\}_{b \in [K']}$ are the block projectors of $\cH_{\mathbf{g}}$ of Definition~\ref{def:block_pinching}, $\{D^{(\lambda)}_a\}_a$ is any Kraus representation of $\D_\lambda$ and $\{E^{(\lambda)}_c\}_c$ is any Kraus representation of $\E_\lambda$.
\end{lemma}

\begin{proof}
By eq.~\eqref{eq:def_pinching_d}, the operators $\{\widetilde P'_b\}_{b \in [K']}$ form a Kraus representation of $\Pi_{\mathbf{g}}$.
Eq.~\eqref{eq:Kraus} then follows from the factorisation $\cN_\lambda = \D_\lambda\circ\Pi_{\mathbf{g}}\circ\E_\lambda$ of eq.~\eqref{eq:N-def}.
\end{proof}

Now, we introduce the ansätze $\mu,\sigma$ which, as we will show, satisfy eq.~\eqref{eq:goal_ansatz}.
In the following, let $\cH_R$ denote a copy of $\cH_{\mathbf{f}}$, with blocks $\cH_{k}^R \cong \C^{d_k}$ and orthonormal bases $\{\ket{\nu}_k^R\}_{\nu \in [d_k]}$, for $k \in [K]$, where $\mathbf{f} = (d_1,\dots,d_K)$.

\begin{definition}\label{def:ansatz}
For $k\in[K]$ let
\begin{equation*}
  \ket{\phi_k}=\frac{1}{\sqrt{d_k}}\sum_{\nu=1}^{d_k}\ket{\nu}_k\otimes\ket{\nu}^R_k
  \;\in\;\C^{d_k}\otimes\cH^R_{k}\subseteq\cH_{\mathbf{f}}\otimes\cH_R.
\end{equation*}
Given a probability vector $q=(q_1,\dots,q_K)$, set
\begin{equation*}
  \sigma_q=\sum_{k=1}^{K}q_k\,\proj{\phi_k}\;\in\;\D(\cH_{\mathbf{f}}\otimes\cH_R),
  \qquad
  \mu=\sum_{k=1}^{K}\proj{\phi_k}\;\in\;\cL(\cH_{\mathbf{f}}\otimes\cH_R).
\end{equation*}
For $p\in[1,+\infty)$ let $\sigma^{(p)}=\sigma_{q^{(p)}}$ with
\begin{equation}\label{eq:tilt}
  q^{(p)}_k=\frac{d_k^{\,p}}{\sum_{j=1}^{K}d_j^{\,p}}=\frac{d_k^{\,p}}
  {\lVert\la(\F)\rVert_p^{\,p}}.
\end{equation}
\end{definition}

We then directly have the following two lemmas.

\begin{lemma}[$\mu$ is a valid test]\label{lem:mu-valid}
$0\le\mu\le\mathds{1}_{\cH_{\mathbf{f}}\otimes\cH_R}$, since $\mu$ is an orthogonal projector.
\end{lemma}

\begin{proof}
The subspaces $\C^{d_k}\otimes\cH^R_{k}$, $k\in[K]$, are mutually orthogonal inside
$\cH_{\mathbf{f}}\otimes\cH_R$.
Hence the rank-one projectors $\proj{\phi_k}$ have mutually orthogonal ranges and their sum is an orthogonal projector.
\end{proof}

\begin{lemma}[The ideal statistic]\label{lem:ideal}
For every probability vector $q$, $(\Pi_{\mathbf{f}}\otimes\Id)(\sigma_q)=\sigma_q$ and
\begin{equation}
  \tr\bigl(\mu\,(\Pi_{\mathbf{f}}\otimes\Id)(\sigma_q)\bigr)=1 .
\end{equation}
\end{lemma}

\begin{proof}
We can write $\proj{\phi_k}=\frac1{d_k}\sum_{\nu,\eta}\ketbra{\nu}{\eta}_k\otimes\ketbra{\nu}{\eta}^R_k$ and $\Pi_{\mathbf{f}}(\ketbra{\nu}{\eta}_k)=\ketbra{\nu}{\eta}_k$ since both $\ket{\nu}_k$ and $\ket{\eta}_k$ lie in the same block $\C^{d_k}$ of $\cH_{\mathbf{f}}$.
Hence $(\Pi_{\mathbf{f}}\otimes\Id)(\proj{\phi_k})=\proj{\phi_k}$ and $(\Pi_{\mathbf{f}}\otimes\Id)(\sigma_q)=\sigma_q$. 
Since the $\proj{\phi_k}$ are mutually orthogonal rank-one projectors, $\tr(\mu\sigma_q)=\sum_{k,l}q_l\tr(\proj{\phi_k}\proj{\phi_l})=\sum_kq_k=1$.
\end{proof}

To finish the proof of item~(ii) of Theorem~\ref{thm:main_UA_SR}, it suffices now to show that, for every $\lambda \in [\Lambda]$ and every $p \in [1,+\infty)$
\begin{equation}\label{eq:target_simulated_statistic}
    \tr\bigl(\mu\,(\cN_\lambda\otimes\Id_d)(\sigma^{(p)})\bigr) \leq \frac{\|\la(\cG)\|_p}{\|\la(\F)\|_p}.
\end{equation}
Bounding the trace $\tr\bigl(\mu\,(\cN_\lambda\otimes\Id_d)(\sigma^{(p)})\bigr)$ requires more involved mathematical machinery, which we develop in the following lemmas.

\begin{lemma}[The simulated statistic as an entanglement fidelity]\label{lem:sim-statistic}
Let $\cM : \cL(\cH_{\mathbf{f}}) \to \cL(\cH_{\mathbf{f}})$ be a channel with Kraus representation $\{K_i\}_i$.
Then, for every probability vector $q = (q_1,\dots,q_K)$,
\begin{equation}\label{eq:ent-fid}
  \tr\bigl(\mu\,(\cM\otimes\Id)(\sigma_q)\bigr)
  =\sum_{k=1}^{K}\frac{q_k}{d_k^{2}}\sum_{i}\bigl|\tr(\widetilde P_kK_i)\bigr|^{2},
\end{equation}
with $\mathbf{f} = (d_1,\dots,d_K)$.
\end{lemma}

\begin{proof}
Write $\proj{\phi_k}=\frac1{d_k}\sum_{\nu,\eta\in[d_k]}\ketbra{\nu}{\eta}_k\otimes\ketbra{\nu}{\eta}^R_k$, so
\begin{equation*}
  (\cM\otimes\Id)(\sigma_q)=\sum_{k}\frac{q_k}{d_k}\sum_{\nu,\eta\in[d_k]}\sum_i
  K_i\ketbra{\nu}{\eta}_kK_i^*\otimes\ketbra{\nu}{\eta}^R_k,\qquad
  \mu=\sum_{l}\frac1{d_l}\sum_{\nu',\eta'\in[d_l]}\ketbra{\nu'}{\eta'}_l\otimes
  \ketbra{\nu'}{\eta'}^R_l .
\end{equation*}
Using $\tr[(X\otimes Y)(X'\otimes Y')]=\tr(XX')\tr(YY')$, the reference-system factor is
\begin{equation*}
  \tr\bigl(\ketbra{\nu'}{\eta'}^R_l\,\ketbra{\nu}{\eta}^R_k\bigr)
  =\delta_{kl}\,\delta_{\eta'\nu}\,\delta_{\eta\nu'}.
\end{equation*}
Hence only the terms $l=k$, $\nu'=\eta$, $\eta'=\nu$ survive and
\begin{align*}
  \tr\bigl(\mu(\cM\otimes\Id)(\sigma_q)\bigr)
  &=\sum_{k}\frac{q_k}{d_k^2}\sum_i\sum_{\nu,\eta\in[d_k]}
    \tr\bigl(\ketbra{\eta}{\nu}_k\,K_i\ketbra{\nu}{\eta}_kK_i^*\bigr)\\
  &=\sum_{k}\frac{q_k}{d_k^2}\sum_i\sum_{\nu,\eta\in[d_k]}
    \bra{\nu}_kK_i\ket{\nu}_k\;\bra{\eta}_kK_i^*\ket{\eta}_k\\
  &=\sum_{k}\frac{q_k}{d_k^2}\sum_i\Bigl(\sum_{\nu}\bra{\nu}_kK_i\ket{\nu}_k\Bigr)
    \overline{\Bigl(\sum_{\eta}\bra{\eta}_kK_i\ket{\eta}_k\Bigr)}
  =\sum_{k}\frac{q_k}{d_k^2}\sum_i\bigl|\tr(\widetilde P_kK_i)\bigr|^2.\qedhere
\end{align*}
\end{proof}

We can now use the Cauchy--Schwarz inequality to reduce the whole problem to a finite optimisation over two matrices.
This method draws inspiration from the proof techniques of~\cite[Theorem~2.3]{Fawzi24}.

\begin{lemma}\label{lem:cauchy_schwarz}
Let $\lambda \in [\Lambda]$ and let $\cN_\lambda$ be as in eq.~\eqref{eq:N-def}, with the Kraus representation eq.~\eqref{eq:Kraus}.
Define
\begin{equation}\label{eq:AB-def}
  A^{(\lambda)}_{k,b}=\tr\bigl(\widetilde P_k\,\D_\lambda(\widetilde P'_b)\bigr),\qquad
  B^{(\lambda)}_{k,b}=\tr\bigl(\widetilde P_k\,\E_\lambda^*(\widetilde P'_b)\bigr),
  \qquad k\in[K],\ b\in[K'],
\end{equation}
where $\{\widetilde P_k\}_{k \in [K]}$ and $\{\widetilde P'_b\}_{b \in [K']}$ are the block projectors of $\cH_{\mathbf{f}}$ and $\cH_{\mathbf{g}}$ of Definition~\ref{def:block_pinching}.
Then $A^{(\lambda)}_{k,b}\ge0$, $B^{(\lambda)}_{k,b}\ge0$,
\begin{equation}\label{eq:AB-constraints}
  \sum_{k=1}^{K}A^{(\lambda)}_{k,b}=d'_b\quad\text{for every }b\in[K'],
  \qquad
  \sum_{b=1}^{K'}B^{(\lambda)}_{k,b}=d_k\quad\text{for every }k\in[K],
\end{equation}
and for every probability vector $q$,
\begin{equation}\label{eq:decoupled-bound}
  \tr\bigl(\mu\,(\cN_\lambda\otimes\Id)(\sigma_q)\bigr)\;\le\;
  \sum_{k=1}^{K}\sum_{b=1}^{K'}\frac{q_k}{d_k^{2}}\,A^{(\lambda)}_{k,b}B^{(\lambda)}_{k,b},
\end{equation}
with $\mathbf{f} = (d_1,\dots,d_K)$ and $\mathbf{g} = (d'_1,\dots,d'_{K'})$.
\end{lemma}

\begin{proof}
Since it is irrelevant in this proof, we omit the parameter $\lambda$.

\emph{Positivity.} $\D$ is CP and $\widetilde P'_b\ge0$, so $\D(\widetilde P'_b)\ge0$ and $A_{k,b}=\tr(\widetilde P_k\D(\widetilde P'_b))\ge0$. Likewise $\E^*$ is CP and $B_{k,b}\ge0$.

\emph{Constraints.} Since $\sum_k\widetilde P_k=\mathds{1}_{\cH_{\mathbf{f}}}$ and $\D$ is trace preserving, $\sum_kA_{k,b}=\tr\bigl(\D(\widetilde P'_b)\bigr)=\tr(\widetilde P'_b)=d'_b$.
Since $\sum_b\widetilde P'_b=\mathds{1}_{\cH_{\mathbf{g}}}$ and $\E^*$ is unital (as $\E$ is a channel), $\sum_bB_{k,b}=\tr\bigl(\widetilde P_k\E^*(\mathds{1}_{\cH_{\mathbf{g}}})\bigr) =\tr(\widetilde P_k)=d_k$.

\emph{The bound.} By Lemma~\ref{lem:sim-statistic} with $K_i=D_a\widetilde P'_bE_c$, $i=(a,b,c)$,
\begin{equation*}
  \tr\bigl(\mu(\cN\otimes\Id)(\sigma_q)\bigr)
  =\sum_k\frac{q_k}{d_k^2}\sum_{a,b,c}\bigl|\tr(\widetilde P_kD_a\widetilde P'_bE_c)\bigr|^2 .
\end{equation*}
Fix $k,a,b,c$. Using $\widetilde P_k^2=\widetilde P_k$, $(\widetilde P'_b)^2=\widetilde P'_b$ and
cyclicity of the trace,
\begin{equation*}
  \tr(\widetilde P_kD_a\widetilde P'_bE_c)
  =\tr\Bigl(\bigl(\widetilde P'_bD_a^*\widetilde P_k\bigr)^{*}\,
  \bigl(\widetilde P'_bE_c\widetilde P_k\bigr)\Bigr).
\end{equation*}
The Cauchy--Schwarz inequality for the Hilbert--Schmidt inner product, $|\tr(x^*y)|^2\le\tr(x^*x)\tr(y^*y)$, applied with $x=\widetilde P'_bD_a^*\widetilde P_k$ and $y=\widetilde P'_bE_c\widetilde P_k$, gives
\begin{equation*}
  \bigl|\tr(\widetilde P_kD_a\widetilde P'_bE_c)\bigr|^2
  \le\tr\bigl(\widetilde P_kD_a\widetilde P'_bD_a^*\bigr)
  \tr\bigl(\widetilde P_kE_c^*\widetilde P'_bE_c\bigr).
\end{equation*}
The right-hand side factorises over $a$ and $c$ once $b$ and $k$ are fixed, so
\begin{equation*}
  \sum_{a,b,c}\bigl|\tr(\widetilde P_kD_a\widetilde P'_bE_c)\bigr|^2
\le\sum_{b}\Bigl(\sum_a\tr\bigl(\widetilde P_kD_a\widetilde P'_bD_a^*\bigr)\Bigr)
  \Bigl(\sum_c\tr\bigl(\widetilde P_kE_c^*\widetilde P'_bE_c\bigr)\Bigr)
  =\sum_bA_{k,b}B_{k,b},
\end{equation*}
where we used $\sum_aD_a\widetilde P'_bD_a^*=\D(\widetilde P'_b)$ and
$\sum_cE_c^*\widetilde P'_bE_c=\E^*(\widetilde P'_b)$.
Multiplying by $q_k/d_k^2$ and summing over $k$ gives eq.~\eqref{eq:decoupled-bound}.
\end{proof}

\begin{remark}
Note that eq.~\eqref{eq:AB-constraints} is a pair of \emph{complementary} normalisations: $A$ is normalised along columns (indexed by the blocks of $\cG$) and $B$ along rows (indexed by the blocks of $\F$).
This asymmetry is not an artefact: $A$ carries the decoder's trace preservation and $B$ the unitality of the adjoint of the encoder.
It is precisely this asymmetry that forces the extremal configurations to be \emph{injective partial matchings} between the two block sets.
\end{remark}

Everything now reduces to a finite-dimensional optimisation which we solve exactly.

\begin{lemma}[Extremal configurations]\label{lem:bilinear}
Let $c=(c_k)_{k\in[K]}\in\mathbb R^K_{\ge0}$, $d = (d_k)_{k \in [K]}\in\mathbb R^{K}_{>0}$ and $d' = (d'_b)_{b \in [K']}\in\mathbb R^{K'}_{>0}$.
Put
\begin{equation*}
  \mathcal A=\Bigl\{A\in\mathbb R^{K\times K'}_{\ge0}\ :\ \sum_{k=1}^KA_{k,b}=d'_b,\ \ \forall b \in [K']\Bigr\},
  \qquad
  \mathcal B=\Bigl\{B\in\mathbb R^{K\times K'}_{\ge0}\ :\ \sum_{b=1}^{K'}B_{k,b}=d_k,\ \ \forall k \in [K]\Bigr\}.
\end{equation*}
Then
\begin{equation}\label{eq:bilinear-max}
  \max_{A\in\mathcal A,\,B\in\mathcal B}\ \sum_{k,b}c_k\,A_{k,b}B_{k,b}
  \;=\;\max\Bigl\{\ \sum_{b\in S}c_{\phi(b)}\,d'_b\,d_{\phi(b)}\ :\
  S\subseteq[K'],\ \phi:S\to[K]\ \text{injective}\ \Bigr\}.
\end{equation}
\end{lemma}

\begin{proof}
Write $f(A,B)=\sum_{k,b}c_kA_{k,b}B_{k,b}$. The sets $\mathcal A$ and $\mathcal B$ are non-empty, both contain the matrices with entries $d'_b/K$ and $d_k/K'$ respectively, and are compact convex polytopes, so the maximum of the continuous function $f$ on $\mathcal A\times\mathcal B$ is attained at a pair $(A^\star,B^\star)$, with value $M$.

\smallskip\noindent\emph{(i) The maximum is attained at a pair of vertices.} The map $A\mapsto f(A,B^\star)$ is linear, hence attains its maximum over the polytope $\mathcal A$ at an extreme point $A^\sharp$; then $f(A^\sharp,B^\star)\ge f(A^\star,B^\star)=M$, so $f(A^\sharp,B^\star)=M$.
Likewise $B\mapsto f(A^\sharp,B)$ is linear and attains its maximum at an extreme point $B^\sharp$ with $f(A^\sharp,B^\sharp)=M$.

\smallskip\noindent\emph{(ii) Description of the vertices.} $\mathcal A=\prod_{b\in[K']}\bigl(d'_b \Delta_{K-1}\bigr)$ is a product of scaled simplices ($\Delta_{K-1}$ the probability simplex on $[K]$), and the extreme points of a product of compact convex sets are exactly the products of extreme points.
Hence the extreme points of $\mathcal A$ are the matrices $A_{k,b}=d'_b\,\delta_{k,\phi(b)}$ for functions $\phi:[K']\to[K]$.
Symmetrically, the extreme points of $\mathcal B$ are $B_{k,b}=d_k\,\delta_{b,\psi(k)}$ for functions $\psi:[K]\to[K']$.

\smallskip\noindent\emph{(iii) Evaluation.} For such a pair,
\begin{equation*}
  f(A,B)=\sum_{k,b}c_k\,d'_b\,d_k\,\delta_{k,\phi(b)}\delta_{b,\psi(k)}
  =\sum_{b\in S}c_{\phi(b)}\,d'_b\,d_{\phi(b)},\qquad
  S=\{b\in[K']\ :\ \psi(\phi(b))=b\}.
\end{equation*}
The restriction $\phi|_S$ is injective: if $b_1,b_2\in S$ and $\phi(b_1)=\phi(b_2)$, then $b_1=\psi(\phi(b_1))=\psi(\phi(b_2))=b_2$.
This proves ``$\le$'' in eq.~\eqref{eq:bilinear-max}.

\smallskip\noindent\emph{(iv) Converse.} Given $S\subseteq[K']$ and an injective $\phi:S\to[K]$,
extend $\phi$ to $[K']$ arbitrarily and define $\psi:[K]\to[K']$ by $\psi(\phi(b))=b$ for $b\in S$
(well defined by injectivity) and arbitrarily elsewhere.
The corresponding vertices satisfy
$\{b:\psi(\phi(b))=b\}\supseteq S$, so $f(A,B)\ge\sum_{b\in S}c_{\phi(b)}d'_bd_{\phi(b)}$ (all terms are non-negative).
This gives ``$\ge$''.
\end{proof}

\begin{lemma}\label{lem:key}
Let $d=(d_1,\dots,d_K)\in\mathbb N^{K}_{>0}$ and $d'=(d'_1,\dots,d'_{K'})\in\mathbb N^{K'}_{>0}$.
Fix $p\in[1,+\infty)$ and let $q^{(p)}$ be as in eq.~\eqref{eq:tilt}. Then for all $A\in\mathcal A$, $B\in\mathcal B$ (with the notation of Lemma~\ref{lem:bilinear}),
\begin{equation}\label{eq:key}
  \sum_{k,b}\frac{q^{(p)}_k}{d_k^{2}}A_{k,b}B_{k,b}\;\le\;\frac{\lVert d'\rVert_p}{\lVert d\rVert_p}.
\end{equation}
\end{lemma}

\begin{proof}
Apply Lemma~\ref{lem:bilinear} with $c_k=q^{(p)}_k/d_k^2=d_k^{\,p-2}/\lVert d\rVert_p^{\,p}$.
For $S\subseteq[K']$ and $\phi:S\to[K]$ injective, the corresponding value is
\begin{equation*}
  \sum_{b\in S}\frac{d_{\phi(b)}^{\,p-2}}{\lVert d\rVert_p^{\,p}}\,d'_b\,d_{\phi(b)}
  =\frac{1}{\lVert d\rVert_p^{\,p}}\sum_{b\in S}d'_b\,d_{\phi(b)}^{\,p-1}.
\end{equation*}

\emph{Case $p=1$.} The exponent $p-1$ vanishes, so the value equals $\frac{1}{\lVert d\rVert_1}\sum_{b\in S}d'_b\le\frac{\lVert d'\rVert_1}{\lVert d\rVert_1}$, which is eq.~\eqref{eq:key}.

\emph{Case $p\in(1,+\infty)$.} Let $p'=\frac{p}{p-1}$ be the Hölder conjugate, so that $(p-1)p'=p$.
Hölder's inequality on the index set $S$ gives
\begin{equation*}
  \sum_{b\in S}d'_b\,d_{\phi(b)}^{\,p-1}
  \;\le\;\Bigl(\sum_{b\in S}(d'_b)^{p}\Bigr)^{1/p}
  \Bigl(\sum_{b\in S}d_{\phi(b)}^{\,(p-1)p'}\Bigr)^{1/p'}
  =\Bigl(\sum_{b\in S}(d'_b)^{p}\Bigr)^{1/p}
  \Bigl(\sum_{b\in S}d_{\phi(b)}^{\,p}\Bigr)^{\frac{p-1}{p}} .
\end{equation*}
Now $\sum_{b\in S}(d'_b)^p\le\lVert d'\rVert_p^{\,p}$, and \emph{because $\phi$ is injective} the indices $\phi(b)$, $b\in S$, are pairwise distinct, whence $\sum_{b\in S}d_{\phi(b)}^{\,p}\le \sum_{k\in[K]}d_k^{\,p}=\lVert d\rVert_p^{\,p}$.
Therefore
\begin{equation*}
  \sum_{b\in S}d'_b\,d_{\phi(b)}^{\,p-1}\;\le\;\lVert d'\rVert_p\,\lVert d\rVert_p^{\,p-1},
\end{equation*}
and the value is at most $\lVert d'\rVert_p\lVert d\rVert_p^{p-1}/\lVert d\rVert_p^{p} =\lVert d'\rVert_p/\lVert d\rVert_p$.
Since by Lemma~\ref{lem:bilinear} every admissible $(A,B)$ gives a value at most the maximum over such $(S,\phi)$, eq.~\eqref{eq:key} follows.
\end{proof}
We can now prove item~(ii) of Theorem~\ref{thm:main_UA_SR}.
\begin{proof}[Proof of item~(ii) of Theorem~\ref{thm:main_UA_SR}]
By eq.~\eqref{eq:reduction_to_pinchings}, it suffices to prove eq.~\eqref{eq:bound_error_approx_case} for the block pinchings $\Pi_{\mathbf{f}}$ and $\Pi_{\mathbf{g}}$, since $\la(\Pi_{\mathbf{f}}) = \la(\F) = \mathbf{f}$ and $\la(\Pi_{\mathbf{g}}) = \la(\cG) = \mathbf{g}$.

Fix $p\in[1,+\infty)$, $(p_\lambda)_{\lambda \in [\Lambda]}$ a probability distribution and channels $\{(\D_\lambda,\E_\lambda) \mid \lambda \in [\Lambda]\}$, with $\E_\lambda : \cL(\cH_{\mathbf{f}}) \to \cL(\cH_{\mathbf{g}})$ and $\D_\lambda : \cL(\cH_{\mathbf{g}}) \to \cL(\cH_{\mathbf{f}})$.
We first show that 
\begin{equation}\label{eq:lower_bound_single_copy}
    \frac{1}{2}\Bigl\|\Pi_{\mathbf{f}} - \sum_{\lambda \in [\Lambda]}p_\lambda\D_\lambda \Pi_{\mathbf{g}} \E_\lambda\Bigr\|_\diamond \geq 1 - \infnormratio{\cG}{\F}{}{},
\end{equation}
that is, we show the lower bound
\begin{equation}
    \deltaSR{1}{1}{\cG}{\F} \geq 1 - \infnormratio{\cG}{\F}{}{}.
\end{equation}
Let $\cN_\lambda$ be as in eq.~\eqref{eq:N-def}, and let $\sigma^{(p)}$, $\mu$ be as in Definition~\ref{def:ansatz}.
By Lemma~\ref{lem:mu-valid}, $0\le\mu\le\mathds{1}$, and $\sigma^{(p)}$ is a state.
Hence, Lemma~\ref{lemma:channel_Holevo} applies to the channels $\Pi_{\mathbf{f}}$ and $\sum_{\lambda \in [\Lambda]}p_\lambda\cN_\lambda$ on $\cL(\cH_{\mathbf{f}})$ with the auxiliary system $\cH_R$:
\begin{equation*}
  \frac{1}{2}\Bigl\|\Pi_{\mathbf{f}}-\sum_{\lambda \in [\Lambda]}p_\lambda\cN_\lambda\Bigr\|_\diamond\;\ge\;
  \tr\bigl(\mu(\Pi_{\mathbf{f}}\otimes\Id)(\sigma^{(p)})\bigr)-\sum_{\lambda \in [\Lambda]}p_\lambda\tr\bigl(\mu(\cN_\lambda\otimes\Id)(\sigma^{(p)})\bigr).
\end{equation*}
By Lemma~\ref{lem:ideal} the first term equals $1$.
For the second term, let $\lambda \in [\Lambda]$.
By Lemma~\ref{lem:cauchy_schwarz}, $\tr\bigl(\mu(\cN_\lambda\otimes\Id)(\sigma^{(p)})\bigr)$ is at most $\sum_{k,b}\frac{q^{(p)}_k}{d_k^2}A^{(\lambda)}_{k,b}B^{(\lambda)}_{k,b}$ for matrices $A^{(\lambda)},B^{(\lambda)}$ satisfying eq.~\eqref{eq:AB-constraints}, i.e.\ $A^{(\lambda)}\in\mathcal A$ and $B^{(\lambda)}\in\mathcal B$ in the notation of Lemma~\ref{lem:bilinear}.
By Lemma~\ref{lem:key} this is at most $\|\la(\cG)\|_p/\|\la(\F)\|_p$.
Hence, for every $\lambda \in [\Lambda]$,
\begin{equation*}
    \tr\bigl(\mu(\cN_\lambda\otimes\Id)(\sigma^{(p)})\bigr) \leq \frac{\|\la(\cG)\|_p}{\|\la(\F)\|_p},
\end{equation*}
which is eq.~\eqref{eq:target_simulated_statistic}.
Therefore,
\begin{equation*}
  \frac{1}{2}\Bigl\|\Pi_{\mathbf{f}}-\sum_{\lambda \in [\Lambda]}p_\lambda\cN_\lambda\Bigr\|_\diamond\;\ge\;1-\sum_{\lambda \in [\Lambda]} p_\lambda\frac{\|\la(\cG)\|_p}{\|\la(\F)\|_p} = 1 - \frac{\|\la(\cG)\|_p}{\|\la(\F)\|_p}.
\end{equation*}
This proves the pointwise bound for finite $p$, since $\cN_\lambda = \D_\lambda\circ\Pi_{\mathbf{g}}\circ\E_\lambda$ by eq.~\eqref{eq:N-def}.

For any fixed finite vector $u$, we have $\|u\|_p \to \|u\|_\infty$ as $p \to +\infty$.
Taking this limit in the preceding pointwise bound therefore yields $\tfrac{1}{2}\|\Pi_{\mathbf{f}}-\sum_{\lambda \in [\Lambda]}p_\lambda\D_\lambda\Pi_{\mathbf{g}}\E_\lambda\|_\diamond\ge1-\|\la(\cG)\|_\infty/\|\la(\F)\|_\infty$.

Since the preceding bound holds for every $p \in [1,+\infty]$ separately, it holds with the right-hand side maximised over $p$, i.e.\ with $\|\la(\cG)\|_p/\|\la(\F)\|_p$ replaced by its infimum over $p\in[1,+\infty]$.
Taking the infimum over $\Lambda$, over the probability distributions $(p_\lambda)_{\lambda \in [\Lambda]}$ and over the pairs of channels $\{(\D_\lambda,\E_\lambda)\}_{\lambda \in [\Lambda]}$, and using eq.~\eqref{eq:reduction_to_pinchings}, we obtain 
\begin{equation}
\deltaSR{1}{1}{\cG}{\F} \geq 1 - \infnormratio{\cG}{\F}{}{}.
\end{equation}

We can now show eq.~\eqref{eq:bound_error_approx_case} in full generality by noticing that $\deltaSR{k}{n}{\cG}{\F} = \deltaSR{1}{1}{\cG^{\otimes n}}{\F^{\otimes k}}$. 
Hence
\begin{align*}
    \deltaSR{k}{n}{\cG}{\F} &=  \deltaSR{1}{1}{\cG^{\otimes n}}{\F^{\otimes k}} \\
    &\geq 1 - \infnormratio{\cG^{\otimes n}}{\F^{\otimes k}}{}{} \\
    &= 1 - \infnormratio{\cG}{\F}{n}{k},
\end{align*}
where the last equality follows by Proposition~\ref{prop:multiplicativity}.
\end{proof}

Now that item~(ii) of Theorem~\ref{thm:main_UA_SR} is proven, we can prove the converse part of item~(i).

\begin{proof}[Proof of the converse part of item~(i) of Theorem~\ref{thm:main_UA_SR}]
Let $\F$ and $\cG$ be idempotent channels, with $\la(\F) \neq (1)$, and let $P_\F = [1,+\infty]$ if $\F$ is not classical, i.e.\ if $\la(\F) \neq (1,\dots,1)$, and $P_\F = [1,+\infty)$ otherwise.
In both cases, $\log(\|\la(\F)\|_p) > 0$ for all $p \in P_\F$, and, following the convention of Section~\ref{sec:emul_capa}, the infimum in eq.~\eqref{eq:emul_capa} is taken over $p \in P_\F$.

Let $r$ be an achievable rate in the shared randomness regime, and let $(k_\nu)_\nu$ and $(n_\nu)_\nu$ be two sequences of positive integers satisfying eq.~\eqref{eq:def_achievable_rate} with $\delta_\nu = \deltaSR{k_\nu}{n_\nu}{\cG}{\F} \to 0$.
Let $p \in P_\F$.
By item~(ii) of Theorem~\ref{thm:main_UA_SR}, for every $\nu$,
\begin{equation*}
    \frac{\|\la(\cG)\|_p^{n_\nu}}{\|\la(\F)\|_p^{k_\nu}} \geq 1 - \delta_\nu.
\end{equation*}
For $\nu$ large enough, $\delta_\nu < 1$, and, taking the logarithm and dividing by $n_\nu\log(\|\la(\F)\|_p) > 0$, we get
\begin{equation*}
    \frac{k_\nu}{n_\nu} \leq \frac{\log(\|\la(\cG)\|_p)}{\log(\|\la(\F)\|_p)} - \frac{\log(1-\delta_\nu)}{n_\nu\log(\|\la(\F)\|_p)}.
\end{equation*}
Since $\delta_\nu \to 0$ and $n_\nu \to +\infty$, the last term vanishes as $\nu \to \infty$, so that
\begin{equation*}
    r \leq \liminf_{\nu\to\infty}\frac{k_\nu}{n_\nu} \leq \frac{\log(\|\la(\cG)\|_p)}{\log(\|\la(\F)\|_p)}.
\end{equation*}
As this holds for every $p \in P_\F$ and every achievable rate $r$, and as every rate achievable in the unassisted regime is also achievable with shared randomness, we obtain
\begin{equation}\label{eq:upper_bound_rate_UA}
    \CUA{\cG}{\F} \leq \CSR{\cG}{\F} \leq \inf_{p \in P_\F}\frac{\log(\|\la(\cG)\|_p)}{\log(\|\la(\F)\|_p)}.
\end{equation}

Finally, by Section~\ref{subsec:emulation}, the zero-error emulation capacity in each regime is at most the corresponding emulation capacity, and the zero-error unassisted emulation capacity is at most the zero-error shared randomness emulation capacity.
In Section~\ref{sec:emul_capa}, we showed that, for every rate $r$ strictly below the right-hand side of eq.~\eqref{eq:upper_bound_rate_UA}, there is a pair $(k,n)$ with $r \leq k/n$ such that $\F^{\otimes k}$ can be emulated exactly by $\cG^{\otimes n}$ without assistance.
Hence, the zero-error unassisted emulation capacity is at least the right-hand side of eq.~\eqref{eq:upper_bound_rate_UA}, so that the four capacities coincide with it.
This ends the proof of item~(i) of Theorem~\ref{thm:main_UA_SR}.
\end{proof}

Finally, item~(ii) yields the strong converse of item~(iii) of Theorem~\ref{thm:main_UA_SR}.
Since the normalised diamond distance between two channels is at most $1$, item~(iii) states that if one tries to emulate $\F$ with $\cG$, with or without shared randomness, at a rate greater than the emulation capacity, the error tends to its maximal value, exponentially fast in the blocklength.
As the lower bound of eq.~\eqref{eq:bound_error_approx_case} holds at every finite blocklength, it suffices to evaluate it at a value of $p$ which almost achieves the infimum in eq.~\eqref{eq:emul_capa}.

\begin{proof}[Proof of item~(iii) of Theorem~\ref{thm:main_UA_SR}]
Let $\F$ and $\cG$ be idempotent channels with $\la(\F) \neq (1)$.
By item~(i), $\CUA{\cG}{\F} = \CSR{\cG}{\F}$, and we write
\begin{equation*}
    C = \CUA{\cG}{\F} = \CSR{\cG}{\F} = \inf_{p \in [1,+\infty]}\frac{\log(\|\la(\cG)\|_p)}{\log(\|\la(\F)\|_p)}.
\end{equation*}
Since $\deltaUA{k}{n}{\cG}{\F} \geq \deltaSR{k}{n}{\cG}{\F}$ for all $k,n \in \N_{>0}$, it suffices to prove the first inequality of item~(iii) in the shared randomness regime.
Throughout this proof, we take the logarithm in base $2$.

Let $\varepsilon > 0$.
We first fix a value of $p$ which almost achieves the infimum defining $C$.
Since $\la(\F) \neq (1)$, we have $\|\la(\F)\|_p > 1$ for all $p \in [1,+\infty)$, and the ratio in the infimum is well-defined for these values of $p$ (the value $p = +\infty$ being also allowed when $\|\la(\F)\|_\infty \geq 2$).
Hence, by definition of the infimum, there is $p^\star \in [1,+\infty]$ such that $\|\la(\F)\|_{p^\star} > 1$ and
\begin{equation}\label{eq:choice_p_star}
    \frac{\log(\|\la(\cG)\|_{p^\star})}{\log(\|\la(\F)\|_{p^\star})} \leq C + \frac{\varepsilon}{2}.
\end{equation}
Note that $p^\star$ depends only on $\F$, $\cG$ and $\varepsilon$.
We write $a = \log(\|\la(\F)\|_{p^\star}) > 0$ and $g = \log(\|\la(\cG)\|_{p^\star}) \geq 0$, so that eq.~\eqref{eq:choice_p_star} reads $g \leq (C + \varepsilon/2)a$.

Let $k,n \in \N_{>0}$ with $k/n \geq C + \varepsilon$.
Multiplying this inequality by $na > 0$, we get
\begin{equation*}
    ka \geq n(C + \varepsilon)a \geq ng + \frac{\varepsilon}{2}na,
\end{equation*}
where we used eq.~\eqref{eq:choice_p_star} in the second inequality.
Therefore,
\begin{equation*}
    \frac{\|\la(\cG)\|_{p^\star}^{n}}{\|\la(\F)\|_{p^\star}^{k}} = 2^{ng - ka} \leq 2^{-\frac{\varepsilon}{2}a\,n}.
\end{equation*}
Now, by item~(ii), since the infimum over $p$ in eq.~\eqref{eq:bound_error_approx_case} is upper-bounded by the value of the ratio at $p = p^\star$,
\begin{equation*}
    \deltaSR{k}{n}{\cG}{\F} \geq 1 - \inf_{p\in[1,+\infty]}\frac{\|\la(\cG)\|_p^{n}}{\|\la(\F)\|_p^{k}} \geq 1 - \frac{\|\la(\cG)\|_{p^\star}^{n}}{\|\la(\F)\|_{p^\star}^{k}} \geq 1 - 2^{-cn},
\end{equation*}
with $c = \frac{\varepsilon}{2}\log(\|\la(\F)\|_{p^\star}) > 0$, which depends only on $\varepsilon$, $\F$ and $\cG$.
This proves item~(iii) in the shared randomness regime, and thus in the unassisted regime.

Finally, let $(k_\nu)_\nu$ and $(n_\nu)_\nu$ be two sequences of positive integers with $n_\nu \to +\infty$ and $\liminf_\nu k_\nu/n_\nu > C$.
Let $\varepsilon > 0$ be such that $\liminf_\nu k_\nu/n_\nu > C + \varepsilon$, so that $k_\nu/n_\nu \geq C + \varepsilon$ for all $\nu$ large enough.
For such $\nu$, since the normalised diamond distance between two channels is at most $1$,
\begin{equation*}
    1 - 2^{-cn_\nu} \leq \deltaSR{k_\nu}{n_\nu}{\cG}{\F} \leq \deltaUA{k_\nu}{n_\nu}{\cG}{\F} \leq 1,
\end{equation*}
so that both errors tend to $1$ exponentially fast in $n_\nu$.
\end{proof}

\section{An efficient algorithm for the emulation capacity}\label{sec:algorithm}

Since the single-letter expression giving the emulation capacity, both in the unassisted and in the shared-randomness-assisted settings, involves computing an infimum over an unbounded interval, it is not a priori clear whether this expression can actually be computed given any two idempotent channels.
It turns out that the infimum in the right-hand side of eq.~\eqref{eq:emul_capa} can be \emph{efficiently} computed.
We provide in this section an algorithm that, given two channels $\F$ and $\cG$ and an error parameter $\varepsilon > 0$, outputs an approximation of $\CUA{\cG}{\F}$ up to the error $\varepsilon$.
Note that since an algorithm computing $\lambda(\F)$ given $\F$ in time $O(d^6\log d)$, where $d = \dim \cH_\F$, is provided in~\cite{Fawzi24}, it suffices to construct an algorithm that, given two integer vectors $u,v$ and an error parameter $\varepsilon > 0$, computes an approximation of 
\begin{equation}
    \inf_{p \in [1,+\infty]} \frac{\log(\|u\|_p)}{\log(\|v\|_p)}.
\end{equation}
The relevance of such an algorithm goes therefore beyond the problem of computing the emulation capacities for idempotent channels.

Throughout this section, we write, for $u \in \N_{>0}^{m_u}$, $v \in \N_{>0}^{m_v}$ and $p \in [1,+\infty]$,
\begin{equation}\label{eq:def_ratio_R}
    R_{u,v}(p) = \frac{\log(\|u\|_p)}{\log(\|v\|_p)}, \qquad R_*(u,v) = \inf_{p \in [1,+\infty]} R_{u,v}(p),
\end{equation}
so that, by item~(i) of Theorem~\ref{thm:main_UA_SR}, $\CUA{\cG}{\F} = \CSR{\cG}{\F} = R_*(\la(\cG),\la(\F))$.
Here, $m_u$ and $m_v$ denote the number of entries of $u$ and $v$, and we write $M_u = \|u\|_\infty$, $M_v = \|v\|_\infty$ their largest entries.
Since the ratio in eq.~\eqref{eq:def_ratio_R} does not depend on the base of the logarithm, we take, in this section only, $\log$ to be the natural logarithm, and we write $\log_2$ for the binary logarithm.
We always assume that $v \neq (1)$, which corresponds to the assumption $\la(\F) \neq (1)$ of Theorem~\ref{thm:main_UA_SR}.
Then, $\|v\|_p > 1$ for every $p \in [1,+\infty)$, since $v$ has either an entry greater than or equal to $2$ or at least two entries equal to $1$.
However, $\|v\|_\infty = 1$ when all the entries of $v$ are equal to $1$, i.e.\ when $\F$ is a classical channel.
In that case, the denominator of $R_{u,v}(\infty)$ vanishes and we use the convention that the term $p = +\infty$ is omitted in the infimum of eq.~\eqref{eq:def_ratio_R}, which is consistent with the last row of Table~\ref{tab:lower_bounds}.

\subsection{Main result and model of computation}

We work in the standard bit model of computable analysis, see e.g.~\cite{Ko91}: the vectors $u$ and $v$ are given as explicit lists of positive integers written in binary, and the precision is given as an integer $b \in \N$.
Hence, the size of the input is measured by $b$ and by
\begin{equation}\label{eq:input_size}
    S = m_u + m_v + \sum_{i=1}^{m_u} \lceil \log_2(u_i + 1)\rceil + \sum_{j=1}^{m_v}\lceil \log_2(v_j+1)\rceil.
\end{equation}
Note that we will frequently use the following elementary bounds: for $x \in \N_{>0}^{m_x}$,
\begin{equation}\label{eq:elementary_bounds_S}
    \log(m_x) \leq m_x \leq S, \qquad \log\Bigl(\sum_{i=1}^{m_x} x_i\Bigr) \leq \log\Bigl(\prod_{i=1}^{m_x}(x_i+1)\Bigr) \leq \sum_{i=1}^{m_x}\lceil\log_2(x_i+1)\rceil \leq S,
\end{equation}
where we used that the natural logarithm is upper-bounded by the binary one on $[1,+\infty)$.
The main result of this section is the following theorem.

\begin{theorem}\label{thm:algorithm}
There is a deterministic algorithm which, given as input $u \in \N_{>0}^{m_u}$, $v \in \N_{>0}^{m_v}$ with $v \neq (1)$ and $b \in \N$, outputs a rational number $\widehat{R}$ such that
\begin{equation}\label{eq:guarantee_algorithm}
    |\widehat{R} - R_*(u,v)| \leq 2^{-b}.
\end{equation}
The bit complexity of this algorithm is polynomial in $S$ and $b$, where $S$ is defined in eq.~\eqref{eq:input_size}.
\end{theorem}

Combining this theorem with Proposition~\ref{prop:computation_shape}, we obtain that the emulation capacities of item~(i) of Theorem~\ref{thm:main_UA_SR} can be efficiently computed from the description of the channels.

\begin{corollary}\label{coro:algorithm_channels}
Let $\F : \cL(\cH_\F) \rightarrow \cL(\cH_\F)$ and $\cG : \cL(\cH_\cG) \rightarrow \cL(\cH_\cG)$ be two idempotent channels with $\la(\F) \neq (1)$, and let $b \in \N$.
One can compute a rational number $\widehat{R}$ such that $|\widehat{R} - \CUA{\cG}{\F}| = |\widehat{R} - \CSR{\cG}{\F}| \leq 2^{-b}$ by first computing the shape vectors $\la(\F)$ and $\la(\cG)$ with the algorithm of Proposition~\ref{prop:computation_shape}, and then running the algorithm of Theorem~\ref{thm:algorithm} on $(\la(\cG),\la(\F),b)$.
The first step has complexity polynomial in $\dim(\cH_\F)$ and $\dim(\cH_\cG)$ in the computational model of~\cite{Fawzi24}.
The second has bit complexity polynomial in $b$, $\dim(\cH_\F)$ and $\dim(\cH_\cG)$.
With these computational conventions, the overall complexity
is polynomial in $b$, $\dim(\cH_\F)$ and $\dim(\cH_\cG)$.
\end{corollary}

\begin{proof}
The correctness follows directly from item~(i) of Theorem~\ref{thm:main_UA_SR} and Theorem~\ref{thm:algorithm}.
For the complexity, recall from Proposition~\ref{prop:structure_idempotent_channel} that $\la(\F)$ has $K \leq \dim(\cH_\F)$ entries, and that each of them is the dimension $d_k$ of a subspace of $\cH_\F$, so that $d_k \leq \dim(\cH_\F)$.
Hence, the contribution of $\la(\F)$ to the input size $S$ of eq.~\eqref{eq:input_size} is at most $\dim(\cH_\F)(1 + \lceil\log_2(\dim(\cH_\F)+1)\rceil)$, and similarly for $\la(\cG)$.
Therefore, $S = O(\dim(\cH_\F)\log(\dim(\cH_\F)) + \dim(\cH_\cG)\log(\dim(\cH_\cG)))$ and the conclusion follows from Theorem~\ref{thm:algorithm}.
\end{proof}

Before we move on to the proof of Theorem~\ref{thm:algorithm}, we make the following remark on its scope.

\begin{remark}\label{rem:approximation_vs_exact}
Theorem~\ref{thm:algorithm} provides a certified approximation of the \emph{value} of $R_*(u,v)$.
It does not, by itself, decide symbolic questions such as whether the infimum in eq.~\eqref{eq:def_ratio_R} is attained at a given value of $p$, e.g.\ at $p = 2$ or at $p = +\infty$, nor whether $R_*(\la(\cG),\la(\F))$ is larger than a given rational number $k/n$.
Indeed, two distinct transcendental expressions can be arbitrarily close to each other, and deciding such questions would require additional structural information on the two shape vectors.
\end{remark}

The algorithm relies on standard subroutines of computer algebra, which we now recall.
First, given a real number $y \geq 0$ and an integer $k \in \N_{>0}$, we write $\lfloor y\rfloor_k$ for a rational number approximating $y$ from below with $k$ binary digits of accuracy, that is, such that
\begin{equation}\label{eq:truncation}
    0 \leq y - \lfloor y\rfloor_k \leq 2^{-k}.
\end{equation}
This is a slight abuse of notation, as such a number is not unique: whenever $\widetilde{y}$ is a rational number with $|\widetilde{y} - y| \leq 2^{-k-1}$, the number $\max\{0,\widetilde{y} - 2^{-k-1}\}$ satisfies eq.~\eqref{eq:truncation}, and we fix one such choice for each quantity we truncate.
We use this one-sided approximation instead of the exact truncation $2^{-k}\lfloor 2^k y\rfloor$ so that computing it requires no separation estimate between $y$ and the dyadic boundaries.
Also, we choose $\widetilde y$ to be a multiple of $2^{-k-1}$, so that $\lfloor y\rfloor_k$ also has at most $k+1$ binary digits after the binary point.
The only transcendental evaluations required by our algorithm are logarithms of positive rational numbers.
Certified approximations of these logarithms can be computed in time polynomial in the requested precision and in the binary description length of their arguments, see e.g.~\cite{Brent76}.
The quotients of logarithms appearing below are evaluated as exact quotients of rational approximations.
Their denominators approximate $\log(a)$ with $a\geq2$ from below to error at most $1/4$, and are therefore at least $\log(2)-1/4>1/4$.

We will also need these two technical lemmas.

\begin{lemma}[Real root isolation, see e.g.~\cite{SagraloffMehlhorn16}]\label{lem:root_isolation}
Let $Q \in \mathbb{Q}[X]$ be a non-zero polynomial of degree $\mathrm{d}$ whose coefficients have binary length at most $\tau$.
\begin{enumerate}[label=(\roman*)]
    \item The square-free part $Q_{\mathrm{sf}} = Q/\gcd(Q,Q')$ of $Q$ can be computed in time polynomial in $\mathrm{d}$ and $\tau$, and its degree and the binary length of its coefficients are polynomial in $\mathrm{d}$ and $\tau$. Moreover, $Q_{\mathrm{sf}}$ is square-free and has the same real roots as $Q$, see e.g.~\cite{vzGG13}.
    \item If $Q$ is square-free, then, given $\kappa \in \N$, one can compute in time polynomial in $\mathrm{d}$, $\tau$ and $\kappa$ a list of pairwise disjoint intervals with rational endpoints, each of width at most $2^{-\kappa}$, such that every interval contains exactly one real root of $Q$ and every real root of $Q$ belongs to one of these intervals.
\end{enumerate}
\end{lemma}

\begin{lemma}\label{lem:quotient_perturbation}
Let $n,\delta,\tilde n,\tilde\delta$ be real numbers with $0 \leq n \leq S$ and $\delta \geq 1/2$, and let $\tau \in (0,1/4]$ be such that $|\tilde n - n| \leq \tau$ and $|\tilde\delta - \delta| \leq \tau$.
Then $\tilde\delta \geq 1/4$ and
\begin{equation}\label{eq:quotient_perturbation}
    \Bigl|\frac{\tilde n}{\tilde\delta} - \frac{n}{\delta}\Bigr| \leq (4+8S)\tau.
\end{equation}
\end{lemma}

\begin{proof}
Since $\tau \leq 1/4$ and $\delta \geq 1/2$, we have $\tilde\delta \geq \delta - \tau \geq 1/4 > 0$. Hence
\begin{equation*}
    \Bigl|\frac{\tilde n}{\tilde\delta} - \frac{n}{\delta}\Bigr| = \Bigl|\frac{\tilde n - n}{\tilde\delta} + n\,\frac{\delta - \tilde\delta}{\delta\tilde\delta}\Bigr| \leq \frac{\tau}{\tilde\delta} + \frac{n\tau}{\delta\tilde\delta} \leq 4\tau + \frac{S\tau}{(1/2)(1/4)} = (4+8S)\tau. \qedhere
\end{equation*}
\end{proof}

\subsection{Reduction to a bounded range of \texorpdfstring{$p$}{p}}\label{subsec:compactification}

The first step of the algorithm consists in reducing the optimisation over the unbounded interval $[1,+\infty]$ to an optimisation over a bounded interval $[1,P_{\max}]$.
To that end, we change the variable $p \in [1,+\infty]$ into $t = 1/p \in [0,1]$.
Throughout this subsection and the next one, given $x \in \N_{>0}^{m_x}$, we write
\begin{equation}\label{eq:def_Z_x}
    Z_x(p) = \sum_{i=1}^{m_x}\Bigl(\frac{x_i}{M_x}\Bigr)^{p}, \qquad p \in [1,+\infty),
\end{equation}
for the sum of the $p$-th powers of the entries of $x$, normalised by its largest entry.
Every term of this sum is at most $1$, and at least one of them equals $1$, so that
\begin{equation}\label{eq:bounds_Z_x}
    1 \leq Z_x(p) \leq m_x, \qquad p \in [1,+\infty),
\end{equation}
and, factoring out $M_x$,
\begin{equation}\label{eq:log_norm_via_Z}
    \log(\|x\|_p) = \log(M_x) + \frac{1}{p}\log(Z_x(p)).
\end{equation}
Furthermore, given a vector $x \in \N_{>0}^{m_x}$, we define the function $A_x : [0,1] \rightarrow \R$ by
\begin{equation}\label{eq:def_A_x}
    A_x(t) = t\log\Bigl(\sum_{i=1}^{m_x} x_i^{1/t}\Bigr) \quad \text{for } t \in (0,1], \qquad A_x(0) = \log(M_x),
\end{equation}
so that for all $p \in [1,+\infty]$, $A_x(1/p) = \log(\|x\|_p)$, with the convention $1/(+\infty) = 0$.
Equivalently, by eq.~\eqref{eq:log_norm_via_Z},
\begin{equation}
A_x(t) = \log(M_x) + t\log(Z_x(1/t))
\end{equation}
for $t \in (0,1]$, which makes the value $A_x(0) = \log(M_x)$ the natural extension of
$A_x$ at $t = 0$.
We also write, for $t \in (0,1]$, $w_x(t)$ for the probability vector with coordinates
\begin{equation}\label{eq:gibbs_weights}
    w_{x,i}(t) = \frac{x_i^{1/t}}{\sum_{j=1}^{m_x} x_j^{1/t}} = \frac{(\tfrac{x_i}{M_x})^{1/t}}{Z_x(1/t)}, \quad i \in [m_x],
\end{equation}
and $H(w) = -\sum_i w_i\log(w_i)$ for the Shannon entropy of a probability vector $w$.

\begin{lemma}\label{lem:properties_A_x}
Let $x \in \N_{>0}^{m_x}$.
The function $A_x$ is continuous and non-decreasing on $[0,1]$, differentiable on $(0,1]$ and, for all $t \in (0,1]$,
\begin{equation}\label{eq:derivative_A_x}
    A_x'(t) = H(w_x(t)) \in [0,\log(m_x)].
\end{equation}
In particular, for all $t \in [0,1]$, $\log(M_x) \leq A_x(t) \leq \log\bigl(\sum_{i=1}^{m_x} x_i\bigr)$. 
\end{lemma}

\begin{proof}
Let $a_i = \log(x_i/M_x) \leq 0$ for $i \in [m_x]$, so that, by eqs.~\eqref{eq:def_Z_x}
and~\eqref{eq:log_norm_via_Z}, for $t \in (0,1]$,
\begin{equation*}
    Z_x(1/t) = \sum_{i=1}^{m_x}e^{a_i/t}, \qquad w_{x,i}(t) = \frac{e^{a_i/t}}{Z_x(1/t)}.
\end{equation*}
The derivative of $t \mapsto Z_x(1/t)$ is given by
\begin{equation}
    \tfrac{d}{dt}Z_x(1/t) = -\sum_i \frac{a_i}{t^2}e^{a_i/t},
\end{equation}
for $t \in (0,1]$. Hence, we obtain the derivative of $A_x$:
\begin{align*}
    A_x'(t) &= \log(Z_x(1/t)) - \frac{1}{t}\sum_{i=1}^{m_x} w_{x,i}(t)a_i\\
    &= -\sum_{i=1}^{m_x}w_{x,i}(t)\Bigl(\frac{a_i}{t} - \log(Z_x(1/t))\Bigr) \\
    &= -\sum_{i=1}^{m_x}w_{x,i}(t)\log(w_{x,i}(t)) \\
    &= H(w_x(t)),
\end{align*}
where we used $\sum_i w_{x,i}(t) = 1$ in the second equality.
This proves eq.~\eqref{eq:derivative_A_x}, since the entropy of a probability vector with $m_x$ coordinates lies in $[0,\log(m_x)]$.
For the continuity at $t = 0$, eq.~\eqref{eq:bounds_Z_x} gives
\begin{equation}
|A_x(t) - \log(M_x)| = t\log(Z_x(1/t)) \leq t\log(m_x) \underset{t \to 0}\to 0.
\end{equation}
Finally, $A_x$ is non-decreasing as its derivative is non-negative, so that $A_x(0) \leq A_x(t) \leq A_x(1) = \log(\sum_{i=1}^{m_x} x_i)$.
\end{proof}

We can now control the behaviour of $R_{u,v}(p)$ for large $p$.

\begin{lemma}[Tail bound]\label{lem:tail_bound}
Let $u \in \N_{>0}^{m_u}$ and $v \in \N_{>0}^{m_v}$ with $M_v \geq 2$.
Then, for all $p \in [1,+\infty)$,
\begin{equation}\label{eq:tail_bound}
    |R_{u,v}(p) - R_{u,v}(\infty)| \leq \frac{L_{u,v}}{p}, \quad \text{with} \quad L_{u,v} = \frac{\log(m_u)}{\log(M_v)} + \frac{\log(\sum_{i=1}^{m_u}u_i)\log(m_v)}{\log(M_v)^2}.
\end{equation}
Furthermore, $L_{u,v} \leq \overline{L} = 2S + 4S^2$, with $S$ defined in eq.~\eqref{eq:input_size}.
\end{lemma}

\begin{proof}
Let $Q(t) = A_u(t)/A_v(t)$ for $t \in [0,1]$, so that $Q(1/p) = R_{u,v}(p)$ for $p \in [1,+\infty]$.
By Lemma~\ref{lem:properties_A_x}, $A_v(t) \geq \log(M_v) > 0$ for all $t \in [0,1]$, so that $Q$ is well-defined and continuous on $[0,1]$ and differentiable on $(0,1]$ with
\begin{equation*}
    |Q'(t)| = \Bigl|\frac{A_u'(t)}{A_v(t)} - \frac{A_u(t)A_v'(t)}{A_v(t)^2}\Bigr| \leq \frac{\log(m_u)}{\log(M_v)} + \frac{\log(\sum_i u_i)\log(m_v)}{\log(M_v)^2} = L_{u,v},
\end{equation*}
where we used $0 \leq \log(M_u) \leq A_u(t) \leq \log(\sum_{i=1}^{m_u} u_i)$, $0 \leq A_u'(t) \leq \log(m_u)$, $\log(M_v) \leq A_v(t)$ and $0 \leq A_v'(t) \leq \log(m_v)$ from Lemma~\ref{lem:properties_A_x}.
By the mean value theorem, $|R_{u,v}(p) - R_{u,v}(\infty)| = |Q(1/p) - Q(0)| \leq L_{u,v}/p$.
Finally, since $\log(M_v) \geq \log(2) \geq 1/2$, the bounds of eq.~\eqref{eq:elementary_bounds_S} give $L_{u,v} \leq 2S + 4S^2$.
\end{proof}

The case $M_v = 1$, which is excluded in Lemma~\ref{lem:tail_bound}, is simple.

\begin{lemma}[Classical target]\label{lem:classical_target}
Let $u \in \N_{>0}^{m_u}$ and $v = (1,\dots,1) \in \N_{>0}^{m_v}$ with $m_v \geq 2$.
Then $p \mapsto R_{u,v}(p)$ is non-decreasing on $[1,+\infty)$, so that
\begin{equation}
    R_*(u,v) = R_{u,v}(1) = \frac{\log(\sum_{i=1}^{m_u}u_i)}{\log(m_v)}.
\end{equation}
\end{lemma}

\begin{proof}
For $p \in [1,+\infty)$, $\log(\|v\|_p) = \frac{1}{p}\log(m_v)$, so that
\begin{equation*}
    R_{u,v}(p) = \frac{p\log(\|u\|_p)}{\log(m_v)} = \frac{\log(\sum_i u_i^p)}{\log(m_v)},
\end{equation*}
which is non-decreasing in $p$ since $u_i \geq 1$ for every $i$.
\end{proof}

Therefore, when $M_v \geq 2$, Lemma~\ref{lem:tail_bound} allows us to discard the
large values of $p$, as follows.
Let $\varepsilon = 2^{-b}\in (0,1]$ be the precision that we target for $\widehat{R}$ in eq.~\eqref{eq:guarantee_algorithm} and let $J \in \N$ be the smallest integer such that $2^J \geq 8\overline{L}/\varepsilon$, with $\overline{L}$ as in Lemma~\ref{lem:tail_bound}.
Then, for every $p \geq 2^J$, eq.~\eqref{eq:tail_bound} gives
$|R_{u,v}(p) - R_{u,v}(\infty)| \leq \overline{L}/2^J \leq \varepsilon/8$, so that
\begin{equation}\label{eq:tail_enclosure}
    \inf_{p \in [2^J,+\infty]}R_{u,v}(p) \in \bigl[R_{u,v}(\infty) - \tfrac{\varepsilon}{8},\ R_{u,v}(\infty)\bigr].
\end{equation}
It is therefore enough to approximate $R_{u,v}(p)$ for $p \geq 2^J$ by $R_{u,v}(\infty) = \log(M_u)/\log(M_v)$, which is a ratio of logarithms of two positive integers, and to minimise $R_{u,v}$ only over the bounded interval $[1,2^J]$.
The latter is the union of the $J$ intervals of the form $[2^j,2^{j+1}]$, for $j \in \{0,\dots,J-1\}$, which we call \emph{dyadic intervals}.
Although the threshold $2^J$ itself is exponentially large in $b$, the number $J$ of these dyadic intervals is only $O(\log_2(S) + b)$, since $\overline{L} = 2S + 4S^2$ is polynomial in $S$.
It then suffices to approximate the minimum of $R_{u,v}$ on each of these dyadic intervals, which is the object of the next subsection, and then to extract the least of these approximants.
The latter step is polynomial since there are only $O(\log_2(S) + b)$ such values to compare.

\subsection{Polynomial approximation on a dyadic interval}

Let us fix a dyadic interval $[P,2P]$, with $P = 2^j$, $j \in \N$, and let us write
$p = P\theta$ with $\theta \in [1,2]$.
By eq.~\eqref{eq:log_norm_via_Z}, we have, for every $x \in \N_{>0}^{m_x}$,
\begin{equation}\label{eq:dyadic_logsum}
    \log(\|x\|_{P\theta}) = \log(M_x) + \frac{1}{P\theta}\log(Z_{x,P}(\theta)),
\end{equation}
where
\begin{equation*}
    Z_{x,P}(\theta) = Z_x(P\theta) = \sum_{i=1}^{m_x} e^{-z_i\theta},
    \quad z_i = P\log\Bigl(\frac{M_x}{x_i}\Bigr) \geq 0,
\end{equation*}
and $Z_{x,P}(\theta) \in [1,m_x]$ for all $\theta \in [1,2]$, by eq.~\eqref{eq:bounds_Z_x}.
Hence,
\begin{equation}\label{eq:ratio_dyadic}
    R_{u,v}(P\theta) = \frac{N_{P}(\theta)}{D_{P}(\theta)}, 
\end{equation}
with
\begin{align*}
    N_{P}(\theta) &= \theta\log(M_u) + \frac{1}{P}\log(Z_{u,P}(\theta)), \\
    D_{P}(\theta) &= \theta\log(M_v) + \frac{1}{P}\log(Z_{v,P}(\theta)),
\end{align*}
where we multiplied both the numerator and the denominator of $R_{u,v}(P\theta)$ by $\theta$.
The functions $\theta \mapsto \log(Z_{x,P}(\theta))$ are analytic, but they are not polynomials.
We now show that they can be uniformly approximated on $[1,2]$ by polynomials with rational coefficients, whose degree and coefficient binary lengths are polynomial in $S$ and $\log_2(1/\eta)$, where $\eta$ is the required precision.
The construction below works for an arbitrary $\eta \in (0,1]$, and $\eta$ will be fixed in step 1 of Algorithm~\ref{alg:rate} as $\eta = \varepsilon/[256(S+1)]$.
It has to be smaller than the target precision $\varepsilon$ by a factor polynomial in $S$, because $\eta$ controls the error on $\log(\|u\|_p)$ and on $\log(\|v\|_p)$ separately, whereas an error $\eta$ on each of them induces an error of order $S\eta$ on their ratio, by Lemma~\ref{lem:uniform_rational_approx}.

The idea of the following algorithm is to discard the exponentially small terms in $Z_{x,P}$, to truncate the Taylor series of the remaining exponentials and to truncate the series of the logarithm.
This is formalised in Algorithm~\ref{alg:logsum}, whose correctness and complexity are proven in Lemma~\ref{lem:logsum}.

\begin{algorithm}[H]
\caption{\texttt{LogSumPolynomial}$(x,P,\eta)$}\label{alg:logsum}
\textbf{Input:} a vector $x \in \N_{>0}^{m_x}$ written in binary, an integer $P\geq1$ written in binary and promised to be a power
of two, for which we write $P=2^j$, and a rational number $\eta \in (0,1]$. \\
\textbf{Output:} a polynomial $B \in \mathbb{Q}[X]$ such that $\sup_{\theta \in [1,2]}|B(\theta) - \log(Z_{x,P}(\theta))| \leq \eta$.
\begin{enumerate}[label=\arabic*.,leftmargin=2em,itemsep=2pt,topsep=4pt]
    \item Set $h \gets \eta/(64m_x)$ and $T \gets \lceil\log_2(1/h)\rceil$, so that $2^{-T} \leq h$.
    \item For each $i \in [m_x]$, set $r_i \gets 0$ if $x_i = M_x$, and $r_i \gets \bigl\lfloor P\log(M_x/x_i)\bigr\rfloor_T$ otherwise, so that $0 \leq r_i$ and $|r_i - P\log(M_x/x_i)| \leq 2^{-T} \leq h$.
    \item Set $\mathcal{J} \gets \{i \in [m_x] \mid r_i \leq T + 1\}$.
    \item Set $d \gets 12(T+1)$ and
    \begin{equation*}
        Z_{\mathrm{pol}} \gets \sum_{i \in \mathcal{J}}\sum_{\ell=0}^{d}\frac{(-r_iX)^\ell}{\ell!} \in \Q[X].
    \end{equation*}
    \item Set $q \gets 4m_x(T-1)$ and $c \gets \bigl\lfloor\log(2m_x)\bigr\rfloor_{K}$ with $K = \lceil\log_2(8/\eta)\rceil$, so that $|c - \log(2m_x)| \leq 2^{-K} \leq \eta/8$.
    \item Compute, using exact arithmetic in $\mathbb{Q}$, the list of coefficients in the monomial basis of the polynomial
    \begin{equation*}
        B = c - \sum_{\ell=1}^{q}\frac{1}{\ell}\Bigl(1 - \frac{Z_{\mathrm{pol}}}{2m_x}\Bigr)^{\ell} \ \in \mathbb{Q}[X],
    \end{equation*}
    and \textbf{return} it.
\end{enumerate}
\end{algorithm}

\noindent
In step 1, since $h$ is rational, $T$ is obtained by comparing $1/h$ with the successive powers of two, that is, by integer comparisons on the binary expansion of $1/h$; no comparison between transcendental numbers is involved.

\begin{lemma}\label{lem:logsum}
Let $x \in \N_{>0}^{m_x}$, $P = 2^j$ with $j \in \N$ and $\eta \in (0,1]$ be rational.
Then, the polynomial $B$ returned by Algorithm~\ref{alg:logsum} satisfies
\begin{equation}\label{eq:logsum_guarantee}
    \sup_{\theta \in [1,2]}|B(\theta) - \log(Z_{x,P}(\theta))| \leq \eta.
\end{equation}
Furthermore, $\deg(B) = O(m_xT^2)$ with $T = O(\log(m_x) + \log(1/\eta))$, and Algorithm~\ref{alg:logsum} has bit complexity polynomial in $S$, $j$, and the binary description length of $\eta$.
\end{lemma}

\begin{proof}
We write $z_i = P\log(M_x/x_i)$ for $i \in [m_x]$, so that $|r_i - z_i| \leq h$.
We proceed in three steps. First, we bound the sum of discarded terms $\sum_{i \notin \mathcal{J}} e^{-z_i\theta}$. Second, we bound the distance between $Z_{x,P}(\theta)$ and $Z_{\mathrm{pol}}(\theta)$, which yields an upper bound on $|\log(Z_{x,P}(\theta)) - \log(Z_{\mathrm{pol}}(\theta))|$. Third, truncating the Taylor series of the logarithm yields an upper bound on $|B(\theta) - \log(Z_{\mathrm{pol}}(\theta))|$, which finally yields eq.~\eqref{eq:logsum_guarantee} by triangle inequality.

\emph{Discarded terms.}
If $i \notin \mathcal{J}$, then 
\begin{equation*}
z_i \geq r_i - h > T + 1 - h \geq T,
\end{equation*}
so that, for $\theta \in [1,2]$,
\begin{equation*}
0 \leq e^{-z_i\theta} \leq e^{-T} \leq 2^{-T} \leq h= \eta/(64m_x).
\end{equation*}
Hence the discarded terms contribute at most $m_x\eta/(64m_x) = \eta/64$ to $Z_{x,P}(\theta)$, that is
\begin{equation}\label{eq:bound_discarded_temrs}
0 \leq \sum_{i \notin \mathcal{J}} e^{-z_i\theta} \leq \frac{\eta}{64}.
\end{equation}

\emph{Approximating $Z_{x,P}$ by $Z_{\mathrm{pol}}$:}
Let $\theta \in [1,2]$,
\begin{align}
    \left|Z_{x,P}(\theta) - Z_{\mathrm{pol}}(\theta)\right| &= \left|\sum_{i=1}^{m_x} e^{-z_i\theta} - \sum_{i \in \mathcal{J}}\sum_{\ell=0}^{d}\frac{(-r_i\theta)^\ell}{\ell!}\right| \\
    &\leq \sum_{i \notin \mathcal{J}} e^{-z_i \theta} + \left|\sum_{i \in \mathcal{J}}\bigl(e^{-z_i \theta}- \sum_{\ell = 0}^d\frac{(-r_i\theta)^\ell}{\ell!}\bigr)\right| \\
    &\leq \frac{\eta}{64} + \sum_{i \in \mathcal{J}} \left|e^{-z_i \theta}- \sum_{\ell = 0}^d\frac{(-r_i\theta)^\ell}{\ell!}\right|,
\end{align}
where the last inequality follows from eq.~\eqref{eq:bound_discarded_temrs} and the triangle inequality.
Let $i \in \mathcal{J}$ and $\theta \in [1,2]$.
By the mean value theorem, we have $|e^{-z_i\theta} - e^{-r_i\theta}| \leq \theta|z_i - r_i| \leq 2h = \eta/(32m_x)$. Hence,
\begin{equation}\label{eq:term_by_term_bound_intermediate}
    \left|e^{-z_i\theta} - \sum_{\ell = 0}^d\frac{(-r_i\theta)^\ell}{\ell!}\right| \leq \left|e^{-z_i\theta} - e^{-r_i \theta}\right| + \left|e^{-r_i \theta} - \sum_{\ell = 0}^d\frac{(-r_i\theta)^\ell}{\ell!}\right| \leq \frac{\eta}{32m_x} + \left|e^{-r_i \theta} - \sum_{\ell = 0}^d\frac{(-r_i\theta)^\ell}{\ell!}\right|.
\end{equation}
It suffices now to upper bound the right-most term in eq.~\eqref{eq:term_by_term_bound_intermediate}. 
We have $r_i \theta \in [0,2(T+1)]$, since $r_i \leq T+1$ for $i \in \mathcal{J}$ and $\theta \leq 2$.
We write $V = 2(T+1)$.
Let $y \in [0,V]$.
By the Taylor formula with Lagrange remainder applied to $y \mapsto e^{-y}$ at the order $d$, there exists $\xi \in [0,y]$ such that
\begin{equation*}
    e^{-y} - \sum_{\ell=0}^d\frac{(-y)^\ell}{\ell!} = (-1)^{d+1}\frac{e^{-\xi}}{(d+1)!}y^{d+1}.
\end{equation*}
Since $0 \leq e^{-\xi}\leq 1$ and $0 \leq y \leq V$, and since $(d+1)! \geq ((d+1)/e)^{d+1}$, this gives
\begin{equation*}
    \Bigl|e^{-y} - \sum_{\ell=0}^d\frac{(-y)^\ell}{\ell!}\Bigr| \leq \frac{V^{d+1}}{(d+1)!} \leq \Bigl(\frac{eV}{d+1}\Bigr)^{d+1}.
\end{equation*}
By the choice $d = 12(T+1)$ made in step 4, we have $d + 1 > 12(T+1) = 6V \geq 2eV$, so that $eV/(d+1) \leq 1/2$ and
\begin{equation*}
    \Bigl|e^{-y} - \sum_{\ell=0}^d\frac{(-y)^\ell}{\ell!}\Bigr| \leq 2^{-(d+1)} \leq 2^{-T} \leq h = \frac{\eta}{64m_x},
\end{equation*}
where the last two inequalities follow from $d+1 \geq T$ and from the choice $T = \lceil\log_2(1/h)\rceil$ made in step 1.
Summing the two contributions over $i \in \mathcal{J}$ and adding the discarded terms, we obtain
\begin{equation*}
    \sup_{\theta \in [1,2]}|Z_{\mathrm{pol}}(\theta) - Z_{x,P}(\theta)| \leq m_x\Bigl(\frac{\eta}{32m_x} + \frac{\eta}{64m_x}\Bigr) + \frac{\eta}{64} = \frac{\eta}{16}.
\end{equation*}
In particular, since $Z_{x,P}(\theta) \in [1,m_x]$ and $\eta \leq 1$, we have $Z_{\mathrm{pol}}(\theta) \in [1/2,2m_x]$ for all $\theta \in [1,2]$.
As $\log$ is $2$-Lipschitz on $[1/2,+\infty)$, i.e.\ for all $x,y \in [1/2,+\infty)$, $|\log(x) - \log(y)| \leq 2|x-y|$, which follows from the mean value theorem, we obtain $|\log(Z_{\mathrm{pol}}(\theta)) - \log(Z_{x,P}(\theta))| \leq \eta/8$.

\emph{Series of the logarithm.}
For $z \in [1/2,2m_x]$, we let $y = 1 - z/(2m_x) \in [0,1-1/(4m_x)]$ and
\begin{equation*}
    \log(z) = \log(2m_x) + \log(1 - y) = \log(2m_x) - \sum_{\ell=1}^{\infty}\frac{y^\ell}{\ell}.
\end{equation*}
The tail of this series satisfies
\begin{equation*}
    \sum_{\ell > q}\frac{y^\ell}{\ell} \leq \frac{y^{q+1}}{1-y} \leq 4m_x\Bigl(1-\frac{1}{4m_x}\Bigr)^{q} \leq 4m_xe^{-q/(4m_x)} \leq 4m_x 2^{-\lceil\log_2(32m_x/\eta)\rceil} \leq \frac{\eta}{8}.
\end{equation*}
Applying this bound with $z = Z_{\mathrm{pol}}(\theta)$, and taking into account the error $\eta/8$ on the constant $c$, we get that for all $\theta \in [1,2]$, $|B(\theta) - \log(Z_{\mathrm{pol}}(\theta))| \leq \eta/4$.
Together with the previous step, this proves eq.~\eqref{eq:logsum_guarantee}, with the error $3\eta/8 \leq \eta$.

\emph{Complexity.}
We have $\deg(Z_{\mathrm{pol}}) \leq d = O(T)$ and $\deg(B) \leq dq = O(m_xT^2)$.
Then, let us bound the cost and the size of the numbers $r_i$ computed at step 2.
Multiplying by $P = 2^j$ is an exact shift of the binary point, but it multiplies absolute errors by $2^j$; hence $r_i$ is obtained by computing $\log(M_x/x_i)$ with an absolute error at most $2^{-(T+j+1)}$, which takes time polynomial in $T + j$ and in the binary lengths of $M_x$ and $x_i$~\cite{Brent76}, followed by a shift and a truncation.
Since $T = \lceil\log_2(64m_x/\eta)\rceil = O(\log_2(m_x) + \log_2(1/\eta))$ and $m_x \leq S$ by eq.~\eqref{eq:elementary_bounds_S}, this is polynomial in $S$, $j$ and $\log_2(1/\eta)$.
Moreover, $0 \leq r_i \leq P\log(M_x/x_i) \leq 2^jS$, again by eq.~\eqref{eq:elementary_bounds_S}, and $r_i$ has at most $T+1$ binary digits after the binary point, so that its binary length is $O(j + \log_2(S) + T)$, which is polynomial in $S$, $j$ and $\log_2(1/\eta)$ as well.
The computation of these approximations is polynomial in $S$, $j$ and $\log_2(1/\eta)$~\cite{Brent76}.
The coefficients of $Z_{\mathrm{pol}}$ are sums of at most $m_x$ terms of the form $(-r_i)^\ell/\ell!$ with $\ell \leq d$, so that their binary length is polynomial in $S$, $j$ and $\log_2(1/\eta)$.
Finally, the expansion of the $q$ powers of $1 - Z_{\mathrm{pol}}/(2m_x)$ involves polynomials of degree at most $dq$.
Since the coefficients of $Z_{\mathrm{pol}}$ admit a common denominator of polynomial binary length, these powers and their weighted sum can be computed in exact rational arithmetic using numbers of polynomial binary length.
The degrees, coefficient lengths and arithmetic operation counts are all polynomial in $S$, $j$ and $\log_2(1/\eta)$.
\end{proof}

Applying Algorithm~\ref{alg:logsum} to $u$ and to $v$ provides polynomial approximations $B_{u,P}$ and $B_{v,P}$ of $\log(Z_{u,P})$ and $\log(Z_{v,P})$.
Together with rational approximations $c_u$, $c_v$ of $\log(M_u)$ and $\log(M_v)$, we obtain a rational function
\begin{equation}\label{eq:rational_approximant}
    \widetilde{R}_P(\theta) = \frac{\widetilde{N}_P(\theta)}{\widetilde{D}_P(\theta)}, \quad \text{with} \quad \widetilde{N}_P(\theta) = \theta c_u + \frac{1}{P}B_{u,P}(\theta), \quad \widetilde{D}_P(\theta) = \theta c_v + \frac{1}{P}B_{v,P}(\theta),
\end{equation}
which approximates $R_{u,v}(P\theta)$ uniformly for $\theta \in [1,2]$, as we now show.

\begin{lemma}\label{lem:uniform_rational_approx}
Let $u \in \N_{>0}^{m_u}$ and $v \in \N_{>0}^{m_v}$ with $M_v \geq 2$, let $\eta \in (0,1/8]$ and let $P = 2^j$ with $j \in \N$.
Suppose that $|c_u - \log(M_u)| \leq \eta$, $|c_v - \log(M_v)| \leq \eta$ and that $B_{u,P}$ and $B_{v,P}$ satisfy eq.~\eqref{eq:logsum_guarantee}.
Then, for all $\theta \in [1,2]$, $\widetilde{D}_P(\theta)/\theta \geq 1/4 > 0$ and
\begin{equation}\label{eq:uniform_rational_approx}
    |\widetilde{R}_P(\theta) - R_{u,v}(P\theta)| \leq (8 + 16S)\eta.
\end{equation}
\end{lemma}

\begin{proof}
Let $\theta \in [1,2]$ and $p = P\theta$.
By eqs.~\eqref{eq:dyadic_logsum} and~\eqref{eq:ratio_dyadic}, $N_P(\theta)/\theta = \log(\|u\|_{p})$ and $D_P(\theta)/\theta = \log(\|v\|_p)$.
Since $P\theta \geq 1$, we have
\begin{equation*}
    \Bigl|\frac{\widetilde{N}_P(\theta)}{\theta} - \log(\|u\|_p)\Bigr| \leq |c_u - \log(M_u)| + \frac{1}{P\theta}|B_{u,P}(\theta) - \log(Z_{u,P}(\theta))| \leq 2\eta,
\end{equation*}
and similarly $|\widetilde{D}_P(\theta)/\theta - \log(\|v\|_p)| \leq 2\eta$.
By Lemma~\ref{lem:properties_A_x} and eq.~\eqref{eq:elementary_bounds_S}, $\log(\|v\|_p) \geq \log(M_v) \geq \log(2) \geq 1/2$ and $0 \leq \log(\|u\|_p) \leq S$.
Hence, $\widetilde{D}_P(\theta)/\theta \geq 1/2 - 2\eta \geq 1/4$.
Lemma~\ref{lem:quotient_perturbation}, applied with $n = \log(\|u\|_p)$, $\delta = \log(\|v\|_p)$, $\tilde n = \widetilde N_P(\theta)/\theta$, $\tilde\delta = \widetilde D_P(\theta)/\theta$ and $\tau = 2\eta \leq 1/4$, gives $\widetilde D_P(\theta)/\theta \geq 1/4$ and eq.~\eqref{eq:uniform_rational_approx}, since $(4+8S)\cdot2\eta = (8+16S)\eta$.
\end{proof}

\subsection{Certified minimisation of a rational function}

It remains to minimise the rational function $\widetilde{R}_P$ over $[1,2]$.
Polynomial and rational optimisation can be approached using methods of real algebraic geometry~\cite{BPR06} or semidefinite programming~\cite{JibeteanDeKlerk06}.
For the univariate problem considered here, we give an explicit subroutine based on real root isolation.
It returns a rational interval of prescribed width $\zeta$ containing the minimum, with an explicit Lipschitz bound controlling the required root approximation accuracy.

Since the numerator and denominator of $\widetilde{R}_P$ are polynomials with rational coefficients and its denominator is positive on $[1,2]$, its minimum is attained either at one of the endpoints $\theta \in \{1,2\}$ or at a real root in $(1,2)$ of the numerator of the derivative of $\widetilde{R}_P$, which we write
\begin{equation}\label{eq:critical_polynomial}
    K_P = \widetilde{N}_P'\widetilde{D}_P - \widetilde{N}_P\widetilde{D}_P'.
\end{equation}
This leads to Algorithm~\ref{alg:ratmin}.

\begin{algorithm}[H]
\caption{\texttt{RationalMinimum}$(N,D,\zeta)$}\label{alg:ratmin}
\textbf{Input:} two polynomials $N, D \in \mathbb{Q}[X]$ such that $D(\theta) \geq \theta/4$ for all $\theta \in [1,2]$, and a rational number $\zeta > 0$.\\
\textbf{Output:} an interval $[\ell,h]$ with rational endpoints, of width at most $\zeta$, containing $\min_{\theta \in [1,2]} N(\theta)/D(\theta)$.
\begin{enumerate}[label=\arabic*.,leftmargin=2em,itemsep=2pt,topsep=4pt]
    \item Form $K \gets N'D - ND'$. \textbf{If} $K = 0$, \textbf{return} $\bigl[N(1)/D(1),\,N(1)/D(1)\bigr]$.
    \item Set $\Lambda \gets 16\cdot 2^{\deg(K)}\|K\|_{\ell_1}$, where $\|K\|_{\ell_1}$ denotes the sum of the absolute values of the coefficients of $K$, and $\kappa \gets \max\bigl\{0,\lceil\log_2(2\Lambda/\zeta)\rceil\bigr\}$, so that $2^{-\kappa} \leq \zeta/(2\Lambda)$.
    \item Compute the square-free part $K_{\mathrm{sf}} \gets K/\gcd(K,K')$ and isolate its real roots with Lemma~\ref{lem:root_isolation}, with isolating intervals of width at most $2^{-\kappa}$.
    \item Form the finite set of candidates
    \begin{equation*}
        \Theta \gets \{1,2\} \cup \bigl\{\max\{\alpha,1\} \ \bigm| \ [\alpha,\beta] \text{ an isolating interval of step 3 meeting } [1,2]\bigr\} \subseteq [1,2].
    \end{equation*}
    \item Compute $m \gets \min_{\theta \in \Theta}N(\theta)/D(\theta)$ in exact rational arithmetic and \textbf{return} $[m - \zeta/2,\ m]$.
\end{enumerate}
\end{algorithm}

\begin{lemma}\label{lem:rational_minimum}
Let $N, D \in \mathbb{Q}[X]$ be such that $D(\theta) \geq \theta/4$ for all $\theta \in [1,2]$, and let $\zeta > 0$ be rational.
Then Algorithm~\ref{alg:ratmin} returns an interval of width at most $\zeta$ containing $\min_{\theta \in [1,2]}N(\theta)/D(\theta)$.
Its bit complexity is polynomial in the degrees of $N$ and $D$, in the binary length of their coefficients and in the binary description length of $\zeta$.
\end{lemma}

\begin{proof}
Let $\mu = \min_{\theta\in[1,2]}N(\theta)/D(\theta)$, which exists by continuity, since $D$ does not vanish on $[1,2]$.
The derivative of $N/D$ is $K/D^2$, so that, if $\mu$ is attained at $\theta^\star \in (1,2)$, then $K(\theta^\star) = 0$.
Hence $\mu$ is the minimum of the values of $N/D$ at the endpoints $1,2$ and at the real roots of $K$ in $(1,2)$, which are the real roots of $K_{\mathrm{sf}}$ by Lemma~\ref{lem:root_isolation}(i).
If $K = 0$, then $N/D$ is constant on $[1,2]$ and $\mu = N(1)/D(1)$, which is returned at step 1.

Suppose now that $K \neq 0$. For $\theta \in [1,2]$ and a polynomial $Q$, we have $|Q(\theta)| \leq 2^{\deg(Q)}\|Q\|_{\ell_1}$, so that, since $D(\theta) \geq \theta/4 \geq 1/4$,
\begin{equation}\label{eq:lipschitz_rational}
    \Bigl|\Bigl(\frac{N}{D}\Bigr)'(\theta)\Bigr| = \frac{|K(\theta)|}{D(\theta)^2} \leq 16\cdot 2^{\deg(K)}\|K\|_{\ell_1} = \Lambda
\end{equation}
for all $\theta \in [1,2]$. In particular $N/D$ is $\Lambda$-Lipschitz on $[1,2]$, and the binary length of $\Lambda$ is polynomial in the degrees and in the coefficient lengths of $N$ and $D$.

Since $\Theta \subseteq [1,2]$, we have $m \geq \mu$.
Conversely, let $\theta^\star \in [1,2]$ be such that $N(\theta^\star)/D(\theta^\star) = \mu$.
If $\theta^\star \in \{1,2\}$, then $m \leq \mu$.
Otherwise $K(\theta^\star) = 0$, hence $K_{\mathrm{sf}}(\theta^\star) = 0$, so that $\theta^\star$ belongs to one of the isolating intervals, say $[\alpha_0,\beta_0]$, of step 3, which therefore meets $[1,2]$.
The corresponding candidate $\theta_0 = \max\{\alpha_0,1\}$ is the projection of $\alpha_0$ onto $[1,2]$, and $\theta^\star \in [1,2]$, so that $|\theta_0 - \theta^\star| \leq |\alpha_0 - \theta^\star| \leq \beta_0 - \alpha_0 \leq 2^{-\kappa} \leq \zeta/(2\Lambda)$.
Since $N/D$ is $\Lambda$-Lipschitz on $[1,2]$, we have
\begin{align*}
    0 \leq m-\mu &\leq \frac{N(\theta_0)}{D(\theta_0)} - \frac{N(\theta^\star)}{D(\theta^{\star})} \\
    &\leq \Lambda |\theta_0 - \theta^{\star}| \leq \frac{\zeta}{2}.
\end{align*}
Hence $\mu \in [m - \zeta/2, m]$, which is the returned interval, of width $\zeta/2 \leq \zeta$.

For the complexity, the polynomial $K$, its derivative, the greatest common divisor $\gcd(K,K')$, the quotient $K_{\mathrm{sf}}$ and the rational number $\Lambda$ can be computed in time polynomial in the degrees and coefficient lengths of $N$ and $D$, and so can the binary lengths of their coefficients.
By Lemma~\ref{lem:root_isolation}(ii), the root isolation takes time polynomial in the degree and coefficient length of $K_{\mathrm{sf}}$ and in $\kappa = O(1+ \log_2(\max\{1,\Lambda\}) + \log_2(\max\{1,1/\zeta\}))$, where the maxima in the definition of $\kappa$ ensure the logarithms are non-negative.
Finally, there are at most $\deg(K)+2$ candidates, and each evaluation of $N(\theta)/D(\theta)$ is an exact rational computation of polynomial size.
\end{proof}

\subsection{The algorithm and the proof of Theorem~\ref{thm:algorithm}}

We are now ready to present the full algorithm, which combines the reduction of Section~\ref{subsec:compactification}, the polynomial approximations of Algorithm~\ref{alg:logsum} on each dyadic interval and Algorithm~\ref{alg:ratmin}.

\begin{algorithm}[H]
\caption{\texttt{ApproximateRate}$(u,v,b)$}\label{alg:rate}
\textbf{Input:} two vectors $u \in \N_{>0}^{m_u}$, $v \in \N_{>0}^{m_v}$ written in binary, with $v \neq (1)$, and an integer $b \in \N$.\\
\textbf{Output:} an interval $I = [a,z]$ with rational endpoints containing $R_*(u,v)$, of width at most $2^{-b}$, and its midpoint $\widehat{R} = (a+z)/2$.
\begin{enumerate}[label=\arabic*.,leftmargin=2em,itemsep=2pt,topsep=4pt]
    \item Set $\varepsilon \gets 2^{-b}$, compute $S$ as in eq.~\eqref{eq:input_size}, and set
    \begin{equation*}
        \eta \gets \frac{\varepsilon}{256(S+1)}, \qquad k \gets b + 8 + \lceil\log_2(S+1)\rceil, \qquad\text{so that } 2^{-k} \leq \eta.
    \end{equation*}
    \item Set $\overline{L} \gets 2S + 4S^2$ and $J \gets b + 3 + \lceil\log_2(\overline{L})\rceil$, so that $2^J \geq 8\overline{L}/\varepsilon$.
    \item \textbf{If} $u = (1)$, \textbf{return} $I = [0,0]$ and $\widehat{R} = 0$.
    \item \textbf{If} $M_v = 1$, set $w \gets \bigl\lfloor\log(\sum_iu_i)\bigr\rfloor_k \big/ \bigl\lfloor\log(m_v)\bigr\rfloor_k$ and \textbf{return} $I = [w - \varepsilon/8, w + \varepsilon/8]$ and $\widehat{R} = w$.
    \item Set $c_u \gets \lfloor\log(M_u)\rfloor_k$, $c_v \gets \lfloor\log(M_v)\rfloor_k$ and $w_\infty \gets c_u/c_v$.
    \item Initialise $a \gets w_\infty - \varepsilon/4$ and $z \gets w_\infty + \varepsilon/4$.
    \item \textbf{For} $j = 0,\dots,J-1$:
    \begin{enumerate}[label=(\alph*),leftmargin=1.8em,itemsep=1pt,topsep=1pt]
        \item Set $B_{u} \gets$ \texttt{LogSumPolynomial}$(u,2^j,\eta)$ and $B_{v} \gets$ \texttt{LogSumPolynomial}$(v,2^j,\eta)$.
        \item Form $\widetilde{N}(\theta) \gets \theta c_u + 2^{-j}B_u(\theta)$ and $\widetilde{D}(\theta) \gets \theta c_v + 2^{-j}B_v(\theta)$.
        \item Compute $[\ell_j,h_j] \gets$ \texttt{RationalMinimum}$(\widetilde{N},\widetilde{D},\varepsilon/4)$ and update
        \begin{equation*}
            a \gets \min\{a,\ \ell_j - \varepsilon/8\}, \qquad z \gets \min\{z,\ h_j + \varepsilon/8\}.
        \end{equation*}
    \end{enumerate}
    \item \textbf{Return} $I = [a,z]$ and $\widehat{R} = (a+z)/2$.
\end{enumerate}
\end{algorithm}
\noindent
\noindent

\begin{proof}[Proof of Theorem~\ref{thm:algorithm}]
We show that Algorithm~\ref{alg:rate} satisfies the requirements of Theorem~\ref{thm:algorithm}.
Actually, we show the stronger statement that it outputs an interval of width at most $\varepsilon/2$ containing $R_*(u,v)$, so that its midpoint $\widehat{R}$ satisfies $|\widehat{R} - R_*(u,v)| \leq \varepsilon/4 \leq 2^{-b}$.

\emph{Trivial cases.}
If $u = (1)$, then $\|u\|_p = 1$ for all $p$ and $R_{u,v}(p) = 0$ for all $p$, so that $R_*(u,v) = 0$.
If $M_v = 1$, then $v = (1,\dots,1)$ with $m_v \geq 2$ since $v \neq (1)$, and $R_*(u,v) = \log(\sum_iu_i)/\log(m_v)$ by Lemma~\ref{lem:classical_target}.
By Lemma~\ref{lem:quotient_perturbation}, applied with $\tau = 2^{-k} \leq \eta \leq 1/4$ and with the bounds $\log(\sum_iu_i) \leq S$ of eq.~\eqref{eq:elementary_bounds_S} and $\log(m_v) \geq \log(2) \geq 1/2$, we get $|w - R_*(u,v)| \leq (4+8S)\eta \leq 8(S+1)\eta = \varepsilon/32$.
In both cases the returned interval contains $R_*(u,v)$ and has width at most $\varepsilon/4$.

\emph{General case.}
Suppose now that $M_v \geq 2$.
Since $[1,+\infty] = [2^J,+\infty] \cup \bigcup_{j=0}^{J-1}[2^j,2^{j+1}]$, we have
\begin{equation}\label{eq:decomposition_minimum}
    R_*(u,v) = \min\Bigl\{\mu_\infty,\ \min_{0 \leq j \leq J-1}\mu_j\Bigr\}, \quad \text{with} \quad \mu_\infty = \inf_{p \in [2^J,+\infty]}R_{u,v}(p) \ \text{ and } \ \mu_j = \min_{\theta\in[1,2]}R_{u,v}(2^j\theta).
\end{equation}
For $c \in \mathcal{C} = \{\infty\}\cup\{0,\dots,J-1\}$, let us introduce the intervals $[\alpha_c,\beta_c]$ implicitly considered by the algorithm, namely
\begin{equation*}
    [\alpha_\infty,\beta_\infty] = \Bigl[w_\infty - \frac{\varepsilon}{4},\ w_\infty + \frac{\varepsilon}{4}\Bigr], \qquad [\alpha_j,\beta_j] = \Bigl[\ell_j - \frac{\varepsilon}{8},\ h_j + \frac{\varepsilon}{8}\Bigr] \quad (0 \leq j \leq J-1),
\end{equation*}
where $[\ell_j,h_j]$ is the interval returned by \texttt{RationalMinimum} at step 7(c).
By construction, the values returned at step 8 are $a = \min_{c \in \mathcal{C}}\alpha_c$ and $z = \min_{c\in\mathcal{C}}\beta_c$.
We first show that, for every $c \in \mathcal{C}$, the interval $[\alpha_c,\beta_c]$ contains the corresponding term $\mu_c$ of eq.~\eqref{eq:decomposition_minimum} and has width at most $\varepsilon/2$.

Let $c = \infty$.
By Lemma~\ref{lem:tail_bound} and by the choice of $J$, we have $|R_{u,v}(p) - R_{u,v}(\infty)| \leq \overline{L}/p \leq \overline{L}/2^J \leq \varepsilon/8$ for all $p \geq 2^J$, so that $\mu_\infty \in [R_{u,v}(\infty) - \varepsilon/8,\, R_{u,v}(\infty)]$.
Moreover, $R_{u,v}(\infty) = \log(M_u)/\log(M_v)$, and $w_\infty = c_u/c_v$ with $|c_u - \log(M_u)| \leq 2^{-k} \leq \eta$ and $|c_v - \log(M_v)| \leq \eta$ by step 5.
Since $0 \leq \log(M_u) \leq S$ by eq.~\eqref{eq:elementary_bounds_S}, $\log(M_v) \geq \log(2) \geq 1/2$ and $\eta \leq 1/256 \leq 1/4$, Lemma~\ref{lem:quotient_perturbation} gives
\begin{equation*}
    |w_\infty - R_{u,v}(\infty)| \leq (4+8S)\eta \leq 8(S+1)\eta = \frac{\varepsilon}{32}.
\end{equation*}
Hence $\mu_\infty \in [w_\infty - \varepsilon/4, w_\infty + \varepsilon/4] = [\alpha_\infty,\beta_\infty]$, which has width $\varepsilon/2$.

Let now $c = j \in \{0,\dots,J-1\}$ and $P = 2^j$, and let $\widetilde{N}, \widetilde{D}$ be the polynomials formed at step 7(b), which are those of eq.~\eqref{eq:rational_approximant} for this value of $P$.
By Lemma~\ref{lem:logsum}, $B_u$ and $B_v$ satisfy eq.~\eqref{eq:logsum_guarantee}, and $|c_u - \log(M_u)| \leq \eta$, $|c_v - \log(M_v)| \leq \eta$ as above, with $\eta \leq 1/8$.
Hence Lemma~\ref{lem:uniform_rational_approx} applies and gives $\widetilde{D}(\theta) \geq \theta/4$ on $[1,2]$, as required by \texttt{RationalMinimum}, together with
\begin{equation*}
    \sup_{\theta \in [1,2]}\Bigl|\frac{\widetilde{N}(\theta)}{\widetilde{D}(\theta)} - R_{u,v}(P\theta)\Bigr| \leq (8+16S)\eta \leq \frac{16(S+1)\varepsilon}{256(S+1)} \leq \frac{\varepsilon}{8}.
\end{equation*}
Taking the minimum over $\theta \in [1,2]$ on both sides, the minimum of $\widetilde{N}/\widetilde{D}$ on $[1,2]$ is at distance at most $\varepsilon/8$ from $\mu_j$.
By Lemma~\ref{lem:rational_minimum}, $[\ell_j,h_j]$ contains that minimum and has width at most $\varepsilon/4$.
Hence $\mu_j \in [\ell_j - \varepsilon/8, h_j + \varepsilon/8] = [\alpha_j,\beta_j]$, which has width at most $\varepsilon/4 + \varepsilon/4 = \varepsilon/2$.

It remains to combine these enclosures.
Since $\min$ is non-decreasing in each of its arguments, $\mu_c \in [\alpha_c,\beta_c]$ for every $c \in \mathcal{C}$ implies
\begin{equation*}
    R_*(u,v) = \min_{c \in \mathcal{C}}\mu_c \in \Bigl[\min_{c\in\mathcal{C}}\alpha_c,\ \min_{c\in\mathcal{C}}\beta_c\Bigr] = [a,z],
\end{equation*}
which is why both endpoints are updated by a minimum in step 7(c).
Furthermore, if $a = \alpha_{c_0}$, then $z \leq \beta_{c_0} \leq \alpha_{c_0} + \varepsilon/2 = a + \varepsilon/2$.
Therefore, the returned interval $[a,z]$ contains $R_*(u,v)$ and has width at most $\varepsilon/2$, independently of the number $J+1$ of terms in eq.~\eqref{eq:decomposition_minimum}.

\emph{Complexity.}
The number of iterations of the loop is $J = O(\log_2(S) + b)$ by the choice of $\overline{L} = 2S + 4S^2$.
For each $j \leq J - 1$, the integer $j$ and $\log_2(1/\eta) = b + O(\log_2(S))$ are polynomial in $S$ and $b$.
By Lemma~\ref{lem:logsum}, the polynomials $B_u$ and $B_v$ are computed in time polynomial in $S$ and $b$, and so are their degrees and the binary lengths of their coefficients.
The same holds for $\widetilde{N}$ and $\widetilde{D}$, as $c_u$ and $c_v$ are computed in polynomial time~\cite{Brent76}.
Then, by Lemma~\ref{lem:rational_minimum}, each call to Algorithm~\ref{alg:ratmin} runs in time polynomial in $S$ and $b$.
Finally, steps 1 to 6 and step 8 are clearly polynomial in $S$ and $b$.
\end{proof}

\begin{remark}
As specified before Theorem~\ref{thm:algorithm}, the shape vectors are given as explicit lists of entries, so that $S$ counts their numbers of entries $m_u$ and $m_v$, and the complexity bound is polynomial in these numbers rather than in their binary lengths.
This is intrinsic to our algorithm rather than an artefact of the encoding: the number of terms $q = 4m_x(T-1)$ retained from the series of the logarithm in Algorithm~\ref{alg:logsum} is linear in $m_x$, because the argument $1 - Z_{\mathrm{pol}}(\theta)/(2m_x)$ of that series is only bounded away from $1$ by $1/(4m_x)$; consequently $\deg(B) = O(m_xT^2)$, and grouping equal entries and writing their multiplicities in binary would not help.
Theorem~\ref{thm:algorithm} therefore says nothing about succinct descriptions of shape vectors, such as $\la(\F^{\otimes k}) = \la(\F)^{\otimes k}$ described by $\la(\F)$ and $k$, which has $m^k$ entries, with $m$ the number of entries of $\lambda(\F)$.
This is not a limitation for the computation of emulation capacities: by Proposition~\ref{prop:multiplicativity}, $R_*(\la(\cG^{\otimes n}),\la(\F^{\otimes k})) = \frac{n}{k}R_*(\la(\cG),\la(\F))$, so that only the shape vectors of $\F$ and $\cG$ themselves are ever needed, and these have at most $\dim(\cH_\F)$ and $\dim(\cH_\cG)$ entries by Proposition~\ref{prop:structure_idempotent_channel}.
\end{remark}

\section{Entanglement-assisted emulation of idempotent channels}\label{sec:EA}

In this section, we prove Theorem~\ref{thm:main_EA}.
The proof relies on the following observation: in the presence of unlimited shared entanglement, every idempotent channel $\F$ is \emph{exactly} interconvertible, with a single use in each direction, with the noiseless classical channel $\Delta_{s}$ on an alphabet of size $s = \|\la(\F)\|_2^2$ (Corollary~\ref{coro:EA_equivalence_classical}).
The two directions of this interconversion are given by the teleportation~\cite{Bennett93teleportation} and the superdense coding~\cite{BennettWiesner92} protocols, applied block by block to the decomposition of Proposition~\ref{prop:structure_idempotent_channel}.
This reduces the entanglement-assisted emulation of idempotent channels to the entanglement-assisted emulation of noiseless classical channels with one another, which we solve exactly at finite blocklength in Lemma~\ref{lem:EA_classical_channels}.
In contrast with the unassisted and shared randomness regimes, where the whole family of $\ell_p$-norms of the shape vectors is needed, the emulation is thus governed by the single quantity $\|\la(\F)\|_2$.
This is reminiscent of the quantum reverse Shannon theorem~\cite{Bennett02,Bennett14}, according to which the asymptotic simulation of a channel, in the presence of free entanglement, is governed by a single quantity, as well as of its one-shot refinement in terms of the smooth max-information~\cite{Fang20}.
Theorem~\ref{thm:main_EA} is an exact, zero-error counterpart of this phenomenon for idempotent channels.

We prove item~(ii) of Theorem~\ref{thm:main_EA} first, and deduce items~(i) and~(iii) from it.
Note that eq.~\eqref{eq:error_approx_case_EA} involves the squares of the $\ell_2$-norms of the shape vectors $\lambda(\F)$ and $\lambda(\cG)$.

Throughout this section, we use the following notation.
Given an idempotent channel $\F$, we write
\begin{equation}\label{eq:def_s}
    s(\F) = \|\la(\F)\|_2^2 = \sum_{k=1}^K d_k^2,
\end{equation}
where $\la(\F) = (d_1,\dots,d_K)$.
Note that $s(\F)$ is a positive integer, and that $s(\F) = 1$ if and only if $\la(\F) = (1)$, since the entries of $\la(\F)$ are positive integers.
Furthermore, for $M \in \N_{>0}$, we regard the completely dephasing channel $\Delta_M : \cL(\C^M) \rightarrow \cL(\C^M)$ as a noiseless classical channel on the alphabet $[M]$.
Recall that $\la(\Delta_M) = (1,\dots,1) \in \N^M$, so that $s(\Delta_M) = M$, and that $\la(\Id_d) = (d)$, so that $s(\Id_d) = d^2$.

\subsection{Entanglement-assisted protocols and exact emulations}

We start by fixing some terminology on entanglement-assisted protocols, which allows us to compose them easily.

An entanglement-assisted protocol is a triple $\mathsf{P} = (\omega_{A_0B_0},\E,\D)$, where $\omega_{A_0B_0} \in \D(\cH_{A_0}\otimes \cH_{B_0})$ is a bipartite state on systems $A_0$, $B_0$ of arbitrary finite dimension, $\E : \cL(\cH_{\text{in}} \otimes \cH_{A_0}) \rightarrow \cL(\cH_{\text{in}}')$ is an encoding channel and $\D : \cL(\cH_{\text{out}}' \otimes \cH_{B_0}) \rightarrow \cL(\cH_{\text{out}})$ is a decoding channel, and where $\cH_{\text{in}}$, $\cH_{\text{out}}$, $\cH_{\text{in}}'$ and $\cH_{\text{out}}'$ are Hilbert spaces.
$\mathsf{P} = (\omega_{A_0B_0},\E,\D)$ maps $\cL(\cH_{\text{in}}') \rightarrow \cL(\cH_{\text{out}}')$ to $\cL(\cH_{\text{in}}) \rightarrow \cL(\cH_{\text{out}})$.
Given $\Psi : \cL(\cH_{\text{in}}') \rightarrow \cL(\cH_{\text{out}}')$, we write
\begin{equation}\label{eq:induced_map}
    \mathsf{P}[\Psi] = \D \circ (\Psi \otimes \Id_{B_0}) \circ (\E \otimes \Id_{B_0})\circ \T_{\omega_{A_0B_0}} : \cL(\cH_{\text{in}}) \rightarrow \cL(\cH_{\text{out}}),
\end{equation}
where $\T_{\omega_{A_0B_0}} : x \mapsto x \otimes \omega_{A_0B_0}$ for $x \in \cL(\cH_{\text{in}})$.
If $\Psi$ is a channel, $\mathsf{P}[\Psi]$ is a channel as well, which we call the channel \emph{induced} by the protocol $\mathsf{P}$ from $\Psi$.
Throughout, whenever we write $\mathsf{P}[\Psi]$, it is understood that the input and output spaces of $\Psi$ are those that the protocol $\mathsf{P}$ expects, namely $\cL(\cH_{\text{in}}')$ and $\cL(\cH_{\text{out}}')$.
The same convention applies to iterated expressions such as $\mathsf{P}_2[\mathsf{P}_1[\Theta]]$.
We say that a channel $\Phi$ can be \emph{exactly emulated} by a channel $\Psi$ with shared entanglement when $\Phi = \mathsf{P}[\Psi]$ for some entanglement-assisted protocol $\mathsf{P}$.

Entanglement-assisted protocols compose: given $\mathsf{P}_1 = (\omega_{A_1B_1},\E_1,\D_1)$ and $\mathsf{P}_2 = (\omega_{A_2B_2},\E_2,\D_2)$, and a channel $\Theta$, we have $\mathsf{P}_2[\mathsf{P}_1[\Theta]] = \mathsf{P}_3[\Theta]$ for the entanglement-assisted protocol $\mathsf{P}_3 = (\omega_{A_1B_1}\otimes\omega_{A_2B_2},\E_3,\D_3)$ given by
\begin{equation}\label{eq:composition_EA}
    \E_3 = \E_1 \circ (\E_2 \otimes \Id_{A_1}), \qquad \D_3 = \D_2 \circ (\D_1 \otimes \Id_{B_2}),
\end{equation}
where the shared state is on the systems $A_0 = A_2A_1$ and $B_0 = B_1B_2$.
This follows by expanding both sides with eq.~\eqref{eq:induced_map}, identifying as is customary the tensor products of Hilbert spaces that only differ by the order of their factors.
In particular, the shared state of the composed protocol is the tensor product of the shared states of $\mathsf{P}_1$ and $\mathsf{P}_2$.

With this notation, for any channels $\Phi : \cL(\cH_{\text{in}}) \rightarrow \cL(\cH_{\text{out}})$ and $\Psi : \cL(\cH_{\text{in}}') \rightarrow \cL(\cH_{\text{out}}')$, we write
\begin{equation}\label{eq:def_delta_EA_general}
    \deltaEA{1}{1}{\Psi}{\Phi} = \inf_{\mathsf{P}}\frac{1}{2}\|\Phi - \mathsf{P}[\Psi]\|_\diamond,
\end{equation}
where the infimum is taken over all the entanglement-assisted protocols $\mathsf{P}$ from maps $\cL(\cH_{\text{in}}') \rightarrow \cL(\cH_{\text{out}}')$ to $\cL(\cH_{\text{in}}) \rightarrow \cL(\cH_{\text{out}})$.
Then, for $k,n \in \N_{>0}$, by definition
\begin{equation}\label{eq:delta_kn_vs_delta_11}
    \deltaEA{k}{n}{\Psi}{\Phi} = \deltaEA{1}{1}{\Psi^{\otimes n}}{\Phi^{\otimes k}}.
\end{equation}

The maps induced by an entanglement-assisted protocol contract the diamond norm.

\begin{lemma}\label{lem:composition_EA}
Let $\mathsf{P} = (\omega_{A_0B_0},\E,\D)$ be an entanglement-assisted protocol.
Then, for every pair of linear maps $\Theta_1$, $\Theta_2$ from $\cL(\cH_{\text{in}}')$ to $\cL(\cH_{\text{out}}')$,
\begin{equation}\label{eq:contraction_EA}
    \|\mathsf{P}[\Theta_1] - \mathsf{P}[\Theta_2]\|_\diamond \leq \|\Theta_1 - \Theta_2\|_\diamond.
\end{equation}
\end{lemma}

\begin{proof}
By linearity, $\mathsf{P}[\Theta_1] - \mathsf{P}[\Theta_2] = \D\circ((\Theta_1 - \Theta_2)\otimes\Id_{B_0})\circ(\E\otimes\Id_{B_0})\circ\T_{\omega_{A_0B_0}}$.
The maps $\D$, $\E\otimes\Id_{B_0}$ and $\T_{\omega_{A_0B_0}}$ are channels, so that their diamond norm is equal to $1$, and the diamond norm is submultiplicative under composition and stable under tensor product with identity maps, i.e.\ $\|(\Theta_1 - \Theta_2)\otimes \Id_{B_0}\|_\diamond = \|\Theta_1 - \Theta_2\|_\diamond$ (see e.g.~\cite[Chapter~3]{watrous18}).
Eq.~\eqref{eq:contraction_EA} follows.
\end{proof}

As a consequence, the minimal error of entanglement-assisted emulation is monotone under exact entanglement-assisted emulation.

\begin{lemma}\label{lem:monotonicity_EA}
Let $\Phi$, $\Phi'$, $\Psi$ and $\Psi'$ be channels, and let $\mathsf{P}_1$, $\mathsf{P}_2$ be entanglement-assisted protocols such that $\Phi' = \mathsf{P}_1[\Phi]$ and $\Psi = \mathsf{P}_2[\Psi']$.
Then
\begin{equation}\label{eq:monotonicity_EA}
    \deltaEA{1}{1}{\Psi'}{\Phi'} \leq \deltaEA{1}{1}{\Psi}{\Phi}.
\end{equation}
\end{lemma}

\begin{proof}
Let $\mathsf{P}$ be an arbitrary entanglement-assisted protocol.
By eq.~\eqref{eq:composition_EA}, applied twice, there is an entanglement-assisted protocol $\mathsf{P}'$ such that $\mathsf{P}_1[\mathsf{P}[\Psi]] = \mathsf{P}_1[\mathsf{P}[\mathsf{P}_2[\Psi']]] = \mathsf{P}'[\Psi']$.
Then, by Lemma~\ref{lem:composition_EA},
\begin{equation*}
    \deltaEA{1}{1}{\Psi'}{\Phi'} \leq \frac{1}{2}\|\Phi' - \mathsf{P}'[\Psi']\|_\diamond = \frac{1}{2}\|\mathsf{P}_1[\Phi] - \mathsf{P}_1[\mathsf{P}[\Psi]]\|_\diamond \leq \frac{1}{2}\|\Phi - \mathsf{P}[\Psi]\|_\diamond.
\end{equation*}
Taking the infimum over $\mathsf{P}$ yields eq.~\eqref{eq:monotonicity_EA}.
\end{proof}

\subsection{The operational meaning of \texorpdfstring{$\|\la(\F)\|_2^2$}{the squared l2-norm of the shape vector}}

By Proposition~\ref{prop:idempotent_vs_pinching}, an idempotent channel $\F$ can be converted into the block pinching $\Pi_{\mathbf{f}}$ associated to its shape vector, and back.
As far as emulation is concerned, we may therefore work with block pinchings throughout this section, and it suffices to understand their entanglement-assisted emulation.
We now show that the block pinching $\Pi_{\mathbf{d}}$ of Definition~\ref{def:block_pinching} is exactly interconvertible with the noiseless classical channel $\Delta_S$, with $S = \|\mathbf{d}\|_2^2$, when shared entanglement is allowed.

We first recall the teleportation and superdense coding protocols, in a form adapted to our needs.
Let $d \geq 1$ be an integer and let $\{\ket{\nu}\}_{\nu \in \{0,\dots,d-1\}}$ be the standard basis of $\C^d$, whose indices are taken modulo $d$.
We consider the \emph{Heisenberg--Weyl} unitaries
\begin{equation}\label{eq:Weyl}
    U_{a,b} = X^aZ^b, \quad (a,b) \in \{0,\dots,d-1\}^2, \quad \text{with} \quad X\ket{\nu} = \ket{\nu+1}, \quad Z\ket{\nu} = e^{2i\pi\nu/d}\ket{\nu},
\end{equation}
which satisfy the orthogonality relations $\tr(U_{a,b}^*U_{a',b'}) = d\,\delta_{a,a'}\delta_{b,b'}$.
Here $X$ and $Z$ are the $d$-dimensional generalisations of the Pauli matrices $\sigma_x$ and $\sigma_z$, to which they reduce for $d = 2$; note that they are unitary but not Hermitian for $d \geq 3$.
As in Definition~\ref{def:ansatz}, we write $\ket{\phi_d} = \frac{1}{\sqrt{d}}\sum_{\nu=0}^{d-1}\ket{\nu}\otimes\ket{\nu} \in \C^d\otimes\C^d$ for the maximally entangled state in dimension $d$ and, for $(a,b) \in \{0,\dots,d-1\}^2$, we let
\begin{equation}\label{eq:Bell_basis}
    \ket{\phi_{a,b}} = (U_{a,b}\otimes\mathds{1})\ket{\phi_d}.
\end{equation}
Since, for every $x,y \in \cL(\C^d)$, $\bra{\phi_d}(x\otimes \mathds{1})^*(y\otimes\mathds{1})\ket{\phi_d} = \frac{1}{d}\tr(x^*y)$, the orthogonality relations of the unitaries $U_{a,b}$ imply that $\{\ket{\phi_{a,b}}\}_{a,b}$ is an orthonormal basis of $\C^d\otimes\C^d$, called the \emph{Bell basis}.
In particular, $\sum_{a,b}\proj{\phi_{a,b}} = \mathds{1}_{\C^d\otimes\C^d}$.

\begin{lemma}[Teleportation identity]\label{lem:teleportation}
Let $X$, $A$, $B$ be three systems of dimension $d$, identified with $\C^d$, and let $R$ be an arbitrary reference system.
Then, for all $x \in \cL(\cH_X\otimes\cH_R)$ and all $(a,b) \in \{0,\dots,d-1\}^2$,
\begin{equation}\label{eq:teleportation_one_outcome}
    (U_{a,b})_B\,(\bra{\phi_{a,b}}_{XA}\otimes\mathds{1}_{BR})\bigl(x_{XR}\otimes\proj{\phi_d}_{AB}\bigr)(\ket{\phi_{a,b}}_{XA}\otimes\mathds{1}_{BR})\,(U_{a,b})_B^* = \frac{1}{d^2}\,x_{BR},
\end{equation}
where $x_{BR}$ denotes the operator $x$ in which the system $X$ is relabelled $B$.
\end{lemma}

\begin{proof}
By linearity, it suffices to prove eq.~\eqref{eq:teleportation_one_outcome} for $x = \ketbra{\psi}{\varphi}$ with $\ket{\psi}, \ket{\varphi} \in \cH_X \otimes \cH_R$, and even to prove that for every $\ket{\psi} \in \cH_X\otimes\cH_R$,
\begin{equation*}
    (U_{a,b})_B(\bra{\phi_{a,b}}_{XA}\otimes\mathds{1}_{BR})\ket{\psi}_{XR}\ket{\phi_d}_{AB} = \frac{1}{d}\ket{\psi}_{BR}.
\end{equation*}
Let us first consider the case $(a,b) = (0,0)$. Writing $\ket{\psi} = \sum_{\nu}\ket{\nu}_X\otimes w_\nu$, where the $w_\nu \in \cH_R$ are vectors which need not be normalised, we get
\begin{equation*}
    (\bra{\phi_d}_{XA}\otimes\mathds{1}_{BR})\ket{\psi}_{XR}\ket{\phi_d}_{AB} = \frac{1}{d}\sum_{\nu,\mu,\kappa}\langle\mu|\nu\rangle_X\langle\mu|\kappa\rangle_A\ket{\kappa}_B\otimes w_\nu = \frac{1}{d}\sum_\nu\ket{\nu}_B\otimes w_\nu = \frac{1}{d}\ket{\psi}_{BR}.
\end{equation*}
In the general case, $\bra{\phi_{a,b}}_{XA} = \bra{\phi_d}_{XA}((U_{a,b})_X^*\otimes\mathds{1}_A)$, so that, applying the previous case to the vector $(U_{a,b})_X^*\ket{\psi}$, we obtain
\begin{equation*}
    (\bra{\phi_{a,b}}_{XA}\otimes\mathds{1}_{BR})\ket{\psi}_{XR}\ket{\phi_d}_{AB} = \frac{1}{d}(U_{a,b})^*_B\ket{\psi}_{BR},
\end{equation*}
and the unitary correction $(U_{a,b})_B$ yields the claim.
\end{proof}

Summing eq.~\eqref{eq:teleportation_one_outcome} over the $d^2$ outcomes $(a,b)$ shows that the teleportation protocol, i.e.\ a measurement of $XA$ in the Bell basis followed by the correction $U_{a,b}$ on $B$, implements the identity channel from $X$ to $B$, even on inputs which are entangled with an arbitrary reference system $R$.
We now use teleportation and superdense coding block by block.

\begin{prop}\label{prop:EA_equivalence_pinching}
Let $\mathbf{d} = (d_1,\dots,d_K) \in \N_{>0}^K$ and $S = \|\mathbf{d}\|_2^2 = \sum_{k=1}^Kd_k^2$.
Then there are entanglement-assisted protocols $\mathsf{P}$ and $\mathsf{P}'$, both with the shared state $\bigotimes_{k=1}^K\proj{\phi_{d_k}}_{A_kB_k}$, such that
\begin{equation}\label{eq:EA_equivalence_pinching}
    \Pi_{\mathbf{d}} = \mathsf{P}[\Delta_S] \qquad \text{and} \qquad \Delta_S = \mathsf{P}'[\Pi_{\mathbf{d}}],
\end{equation}
the protocol $\mathsf{P}$ being the teleportation protocol and $\mathsf{P}'$ the superdense coding protocol, applied block by block.
\end{prop}

\begin{proof}
We identify the classical alphabet $[S]$ with the disjoint union
\begin{equation*}
    \Sigma = \bigsqcup_{k=1}^K\{k\}\times\{0,\dots,d_k-1\}^2 = \{(k,a,b) \mid k \in [K],\ (a,b)\in\{0,\dots,d_k-1\}^2\},
\end{equation*}
which has cardinality $\sum_kd_k^2 = S$, so that $\Delta_S$ is the completely dephasing channel on $\cL(\C^\Sigma)$ in the basis $\{\ket{k,a,b}\}_{(k,a,b)\in\Sigma}$.
In both directions, Alice and Bob share the state $\omega_{A_0B_0} = \bigotimes_{k=1}^K\proj{\phi_{d_k}}_{A_kB_k}$ with $A_0 = A_1\cdots A_K$, $B_0 = B_1\cdots B_K$ and $\cH_{A_k} = \cH_{B_k} = \C^{d_k}$.
For $l \in [K]$, we write $A_{\bar l}$ for the systems $A_m$ with $m \neq l$, and likewise $B_{\bar l}$.
For $k \in [K]$ and $(a,b) \in \{0,\dots,d_k-1\}^2$, we write $U^{(k)}_{a,b}$ and $\ket{\phi^{(k)}_{a,b}}$ for the Heisenberg--Weyl unitaries and the Bell basis in dimension $d_k$, and, with $J_k : \C^{d_k} \rightarrow \cH_{\mathbf{d}}$ the isometry embedding $\C^{d_k}$ as the $k$-th summand of $\cH_{\mathbf{d}}$ as in Definition~\ref{def:block_pinching}, we define the orthogonal projector on $\cH_{\mathbf{d}}\otimes\cH_{A_0}$
\begin{equation}\label{eq:block_Bell_measurement}
    M_{k,a,b} = (J_k \otimes \mathds{1}_{A_k})\proj{\phi^{(k)}_{a,b}}(J_k\otimes\mathds{1}_{A_k})^*\otimes\mathds{1}_{A_{\bar k}}.
\end{equation}
Then $\sum_{a,b}M_{k,a,b} = (J_kJ_k^* \otimes \mathds{1}_{A_k})\otimes\mathds{1}_{A_{\bar k}} = \widetilde{P}_k\otimes\mathds{1}_{A_0}$, since the Bell basis is an orthonormal basis, and thus $\sum_{(k,a,b)\in\Sigma}M_{k,a,b} = \mathds{1}_{\cH_{\mathbf{d}}\otimes\cH_{A_0}}$, i.e.\ $\{M_{k,a,b}\}_{(k,a,b)\in\Sigma}$ is a projective measurement.
We define in the same way the projective measurement $\{N_{k,a,b}\}_{(k,a,b) \in \Sigma}$ on $\cH_{\mathbf{d}}\otimes\cH_{B_0}$ by replacing the systems $A_l$ by $B_l$.

\emph{Emulation of $\Pi_{\mathbf{d}}$ by $\Delta_S$ (teleportation).}
Alice measures her input together with $A_0$ with the measurement $\{M_{k,a,b}\}$, which first reveals the block $k$ of the input and then teleports its content with the pair $A_kB_k$, and she sends the outcome through $\Delta_S$. Bob applies the teleportation correction on $B_k$ and embeds $B_k$ in the $k$-th block.
Formally, let $\E : \cL(\cH_{\mathbf{d}}\otimes\cH_{A_0}) \rightarrow \cL(\C^\Sigma)$ and $\D : \cL(\C^\Sigma\otimes\cH_{B_0}) \rightarrow \cL(\cH_{\mathbf{d}})$ be defined by
\begin{align*}
    \E(y) &= \sum_{(k,a,b)\in\Sigma}\tr(M_{k,a,b}\,y)\proj{k,a,b},\\
    \D(y) &= \sum_{(k,a,b)\in\Sigma}J_kU^{(k)}_{a,b}\tr_{B_{\bar k}}\bigl((\bra{k,a,b}\otimes\mathds{1}_{B_0})y(\ket{k,a,b}\otimes\mathds{1}_{B_0})\bigr)U^{(k)*}_{a,b}J_k^*.
\end{align*}
The map $\E$ is a channel, since $\{M_{k,a,b}\}$ is a POVM, and $\D$ is a channel as it first measures the classical register in the basis $\{\ket{k,a,b}\}$ and then applies, conditionally on the outcome, a channel from $\cL(\cH_{B_0})$ to $\cL(\cH_{\mathbf{d}})$.
Let $R$ be a reference system and $x \in \cL(\cH_{\mathbf{d}}\otimes\cH_R)$.
Since the output of $\E$ is diagonal, $\Delta_S\circ\E = \E$, and we get
\begin{equation*}
    \bigl(\mathsf{P}[\Delta_S]\otimes\Id_R\bigr)(x) = \sum_{(k,a,b)\in\Sigma}(J_kU^{(k)}_{a,b}\otimes\mathds{1}_R)\,T_{k,a,b}\,(J_kU^{(k)}_{a,b}\otimes\mathds{1}_R)^*,
\end{equation*}
with $\mathsf{P} = (\omega_{A_0B_0},\E,\D)$ and $T_{k,a,b} = \tr_{\cH_{\mathbf{d}}A_0B_{\bar k}}\bigl((M_{k,a,b}\otimes\mathds{1}_{B_0R})(x\otimes\omega_{A_0B_0})\bigr) \in \cL(\cH_{B_k}\otimes\cH_R)$.
Now, $M_{k,a,b}$ acts as the identity on the systems $A_{\bar k}$, and $\tr_{A_{\bar k}B_{\bar k}}(\bigotimes_{l\neq k}\proj{\phi_{d_l}}) = 1$.
Hence, writing $x_{kk} = (J_k^*\otimes\mathds{1}_R)x(J_k\otimes\mathds{1}_R) \in \cL(\C^{d_k}\otimes\cH_R)$ and abbreviating
\begin{equation*}
    V_{k,a,b} = J_kU^{(k)}_{a,b}\otimes\mathds{1}_R, \qquad \Lambda_{k,a,b} = \bra{\phi^{(k)}_{a,b}}_{XA_k}\otimes\mathds{1}_{B_kR},
\end{equation*}
where $X \cong \C^{d_k}$ denotes the $k$-th block of the input, the $(k,a,b)$ term of the sum reads
\begin{equation*}
    V_{k,a,b}\,\Lambda_{k,a,b}\bigl(x_{kk}\otimes\proj{\phi_{d_k}}_{A_kB_k}\bigr)\Lambda_{k,a,b}^*\,V_{k,a,b}^* = \frac{1}{d_k^2}(J_k\otimes\mathds{1}_R)x_{kk}(J_k\otimes\mathds{1}_R)^*,
\end{equation*}
where the equality follows from Lemma~\ref{lem:teleportation}.
Summing over the $d_k^2$ pairs $(a,b)$ and over $k \in [K]$, we obtain
\begin{equation*}
    \bigl(\mathsf{P}[\Delta_S]\otimes\Id_R\bigr)(x) = \sum_{k=1}^K(\widetilde{P}_k\otimes\mathds{1}_R)x(\widetilde{P}_k\otimes\mathds{1}_R) = (\Pi_{\mathbf{d}}\otimes\Id_R)(x),
\end{equation*}
that is, $\Pi_{\mathbf{d}} = \mathsf{P}[\Delta_S]$.

\emph{Emulation of $\Delta_S$ by $\Pi_{\mathbf{d}}$ (superdense coding).}
To send the symbol $(k,a,b)$, Alice applies $U^{(k)}_{a,b}$ on her half $A_k$ of the $k$-th maximally entangled state and sends $A_k$ in the $k$-th block of $\Pi_{\mathbf{d}}$. Bob measures the received system together with $B_0$ with the measurement $\{N_{k,a,b}\}$.
Formally, let $\E' : \cL(\C^\Sigma\otimes\cH_{A_0}) \rightarrow \cL(\cH_{\mathbf{d}})$ and $\D' : \cL(\cH_{\mathbf{d}}\otimes\cH_{B_0}) \rightarrow \cL(\C^\Sigma)$ be defined by
\begin{align*}
    \E'(y) &= \sum_{(k,a,b)\in\Sigma}J_kU^{(k)}_{a,b}\tr_{A_{\bar k}}\bigl((\bra{k,a,b}\otimes\mathds{1}_{A_0})y(\ket{k,a,b}\otimes\mathds{1}_{A_0})\bigr)U^{(k)*}_{a,b}J_k^*,\\
    \D'(y) &= \sum_{(k,a,b)\in\Sigma}\tr(N_{k,a,b}\,y)\proj{k,a,b},
\end{align*}
which are channels for the same reasons as $\E$ and $\D$, and let $\mathsf{P}' = (\omega_{A_0B_0},\E',\D')$.
Since $\E'$ starts by measuring its classical input in the basis $\{\ket{k,a,b}\}$ and $\D'$ outputs diagonal operators, $\mathsf{P}'[\Pi_{\mathbf{d}}]$ satisfies $\mathsf{P}'[\Pi_{\mathbf{d}}] = \Delta_S \circ \mathsf{P}'[\Pi_{\mathbf{d}}] \circ \Delta_S$.
Hence, it suffices to show that $\mathsf{P}'[\Pi_{\mathbf{d}}](\proj{k,a,b}) = \proj{k,a,b}$ for every $(k,a,b)\in\Sigma$.
Indeed, in that case, for every reference system $R$ and every $x = \sum_{\sigma,\sigma'\in\Sigma}\ketbra{\sigma}{\sigma'}\otimes x_{\sigma\sigma'} \in \cL(\C^\Sigma\otimes\cH_R)$,
\begin{equation*}
    (\mathsf{P}'[\Pi_{\mathbf{d}}]\otimes\Id_R)(x) = \sum_{\sigma\in\Sigma}\mathsf{P}'[\Pi_{\mathbf{d}}](\proj{\sigma})\otimes x_{\sigma\sigma} = \sum_{\sigma\in\Sigma}\proj{\sigma}\otimes x_{\sigma\sigma} = (\Delta_S\otimes\Id_R)(x).
\end{equation*}
Let us fix $(k,a,b) \in \Sigma$. We have
\begin{equation*}
    (\E'\otimes\Id_{B_0})(\proj{k,a,b}\otimes\omega_{A_0B_0}) = (J_k\otimes\mathds{1}_{B_k})\proj{\phi^{(k)}_{a,b}}_{A_kB_k}(J_k\otimes\mathds{1}_{B_k})^*\otimes\bigotimes_{l \neq k}\frac{\mathds{1}_{B_l}}{d_l},
\end{equation*}
where we used $(U^{(k)}_{a,b}\otimes \mathds{1})\ket{\phi_{d_k}} = \ket{\phi^{(k)}_{a,b}}$ and $\tr_{A_l}(\proj{\phi_{d_l}}) = \mathds{1}_{B_l}/d_l$.
This operator is invariant under $\Pi_{\mathbf{d}}\otimes\Id_{B_0}$, as it is supported on the $k$-th block of $\cH_{\mathbf{d}}$.
Finally, for $(k',a',b') \in \Sigma$,
\begin{equation*}
    \tr\Bigl(N_{k',a',b'}\,(J_k\otimes\mathds{1}_{B_k})\proj{\phi^{(k)}_{a,b}}(J_k\otimes\mathds{1}_{B_k})^*\otimes\bigotimes_{l \neq k}\frac{\mathds{1}_{B_l}}{d_l}\Bigr) = \delta_{k,k'}|\langle\phi^{(k)}_{a',b'}|\phi^{(k)}_{a,b}\rangle|^2 = \delta_{k,k'}\delta_{a,a'}\delta_{b,b'},
\end{equation*}
where we used $J_{k'}^*J_k = \delta_{k,k'}\mathds{1}$ and the orthonormality of the Bell basis.
Hence $\mathsf{P}'[\Pi_{\mathbf{d}}](\proj{k,a,b}) = \proj{k,a,b}$, which shows that $\Delta_S = \mathsf{P}'[\Pi_{\mathbf{d}}]$.
\end{proof}

Combining Proposition~\ref{prop:idempotent_vs_pinching} and Proposition~\ref{prop:EA_equivalence_pinching} with the composability of entanglement-assisted protocols in eq.~\eqref{eq:composition_EA}, we obtain the following corollary.

\begin{corollary}\label{coro:EA_equivalence_classical}
Let $\F$ be an idempotent channel and let $s(\F) = \|\la(\F)\|_2^2$.
Then there are entanglement-assisted protocols $\mathsf{Q}$ and $\mathsf{Q}'$ such that
\begin{equation}
    \F = \mathsf{Q}[\Delta_{s(\F)}] \qquad \text{and} \qquad \Delta_{s(\F)} = \mathsf{Q}'[\F].
\end{equation}
\end{corollary}

\begin{remark}
Corollary~\ref{coro:EA_equivalence_classical} explains why the \emph{squares} $d_k^2$ of the entries of the shape vector, rather than the entries $d_k$ themselves or other powers of them, govern the entanglement-assisted regime.
A block of dimension $d_k$ can carry $d_k^2$ classical symbols with superdense coding, and it can be simulated by the transmission of $d_k^2$ classical symbols with teleportation.
In contrast, without entanglement, a $d_k$-dimensional block can only carry $d_k$ classical symbols, and it cannot be simulated at all by a classical channel when $d_k \geq 2$, as shown by the fourth row of Table~\ref{tab:lower_bounds}.
\end{remark}

\subsection{Entanglement-assisted emulation of noiseless classical channels}

By Corollary~\ref{coro:EA_equivalence_classical}, it remains to understand the entanglement-assisted emulation of a noiseless classical channel by another one.
The following lemma provides the exact minimal error for this task.
Its converse part is a no-signalling argument: it only uses the fact that the local operations of the encoder do not change the marginal state of the decoder, so that it also holds for the larger class of protocols assisted by arbitrary no-signalling correlations, for which linear and semidefinite programming characterisations are available at finite blocklength~\cite{Matthews12,DuanWinter16}.
Zero-error capacities and exact simulation costs of classical channels assisted by non-local correlations were determined along these lines in~\cite{CLMW11}.

\begin{lemma}\label{lem:EA_classical_channels}
Let $M, N \in \N_{>0}$. Then
\begin{equation}\label{eq:EA_classical_channels}
    \deltaEA{1}{1}{\Delta_N}{\Delta_M} = \max\Bigl\{0, 1 - \frac{N}{M}\Bigr\}.
\end{equation}
Furthermore, if $N \geq M$, then $\Delta_M = \mathsf{P}[\Delta_N]$ for an unassisted protocol $\mathsf{P}$.
\end{lemma}

\begin{proof}
\emph{Achievability.}
Suppose first that $N \geq M$ and let $\E : \cL(\C^M) \rightarrow \cL(\C^N)$ and $\D : \cL(\C^N) \rightarrow \cL(\C^M)$ be the channels
\begin{equation*}
    \E(x) = \sum_{m \in [M]}\bra{m}x\ket{m}\proj{m}, \qquad \D(y) = \sum_{z \in [M]}\bra{z}y\ket{z}\proj{z} + \sum_{z = M+1}^{N}\bra{z}y\ket{z}\proj{1}.
\end{equation*}
The encoder $\E$ embeds $[M]$ into $[N]$, so that the symbols $z \in \{M+1,\dots,N\}$ are never transmitted and the second sum defining $\D$, which sends them to the default symbol $1$, never contributes.
Hence $\D\circ\Delta_N\circ\E = \Delta_M$, so that $\Delta_M = \mathsf{P}[\Delta_N]$ for the unassisted protocol $\mathsf{P}$ with encoder $\E$ and decoder $\D$, and $\deltaEA{1}{1}{\Delta_N}{\Delta_M} = 0$.

Suppose now that $N < M$. We use shared randomness, which is a special case of shared entanglement (see Section~\ref{subsec:emulation}).
Let $\mathfrak{S}_M$ be the set of permutations of $[M]$ and, for $\pi \in \mathfrak{S}_M$, let $f_\pi : [M] \rightarrow [N]$ be defined by $f_\pi(m) = \pi(m)$ if $\pi(m) \leq N$ and $f_\pi(m) = 1$ otherwise.
We define the channels $\E_\pi : \cL(\C^M) \rightarrow \cL(\C^N)$ and $\D_\pi : \cL(\C^N) \rightarrow \cL(\C^M)$ by
\begin{equation*}
    \E_\pi(x) = \sum_{m \in [M]}\bra{m}x\ket{m}\proj{f_\pi(m)}, \qquad \D_\pi(y) = \sum_{z \in [N]}\bra{z}y\ket{z}\proj{\pi^{-1}(z)}.
\end{equation*}
In words, Alice and Bob agree on a uniformly random permutation $\pi$; Alice sends $\pi(m)$ when it belongs to $[N]$ and the symbol $1$ otherwise, and Bob outputs $\pi^{-1}$ of the received symbol.
Let $\widetilde{\Delta} = \frac{1}{M!}\sum_{\pi\in\mathfrak{S}_M}\D_\pi\Delta_N\E_\pi$ be the channel induced by this shared randomness strategy.
For $m \in [M]$, we have $\D_\pi\Delta_N\E_\pi(\proj{m}) = \proj{\pi^{-1}(f_\pi(m))}$, which is equal to $\proj{m}$ if $\pi(m) \leq N$ and to $\proj{\pi^{-1}(1)}$ with $\pi^{-1}(1) \neq m$ otherwise (since then $\pi(m) > N \geq 1$).
Splitting the average over $\mathfrak{S}_M$ according to whether $\pi(m) \leq N$ or $\pi(m) > N$, we obtain
\begin{equation*}
    \widetilde{\Delta}(\proj{m}) = \frac{1}{M!}\sum_{\substack{\pi \in \mathfrak{S}_M \\ \pi(m) \leq N}}\proj{m} \;+\; \frac{1}{M!}\sum_{\substack{\pi \in \mathfrak{S}_M \\ \pi(m) > N}}\proj{\pi^{-1}(1)}.
\end{equation*}
As $\pi(m)$ is uniformly distributed on $[M]$ when $\pi$ is uniformly distributed on $\mathfrak{S}_M$, the first sum has weight $N/M$ and the second has weight $1 - N/M$, so that
\begin{equation*}
    \widetilde{\Delta}(\proj{m}) = \frac{N}{M}\proj{m} + \Bigl(1-\frac{N}{M}\Bigr)\tau_m, \qquad \tau_m = \frac{1}{|\{\pi \mid \pi(m) > N\}|}\sum_{\substack{\pi \in \mathfrak{S}_M \\ \pi(m) > N}}\proj{\pi^{-1}(1)}.
\end{equation*}
Then, $\tau_m$ is a diagonal state satisfying $\bra{m}\tau_m\ket{m} = 0$, since $\pi^{-1}(1) \neq m$ for every $\pi$ occurring in the sum.
Now, since each $\E_\pi$ first dephases its input, $\widetilde{\Delta} = \widetilde{\Delta}\circ\Delta_M$.
Let $R$ be a reference system and $\rho \in \D(\C^M\otimes\cH_R)$. Then $(\Delta_M\otimes\Id_R)(\rho) = \sum_{m \in [M]}\proj{m}\otimes\rho_m$ with $\rho_m = (\bra{m}\otimes\mathds{1}_R)\rho(\ket{m}\otimes\mathds{1}_R) \geq 0$ and $\sum_{m \in [M]}\tr(\rho_m) = 1$, so that
\begin{align*}
    \frac{1}{2}\|((\Delta_M - \widetilde{\Delta})\otimes\Id_R)(\rho)\|_1 &= \frac{1}{2}\|((\Delta_M - \widetilde{\Delta})\circ\Delta_M\otimes\Id_R)(\rho)\|_1 \\
    &= \frac{1}{2}\Bigl\|\sum_{m \in [M]}(\Delta_M - \widetilde{\Delta})(\proj{m})\otimes\rho_m\Bigr\|_1 \\
    &\leq \sum_{m \in [M]}\tr(\rho_m)\,\frac{1}{2}\Bigl(1 - \frac{N}{M}\Bigr)\|\proj{m} - \tau_m\|_1 \\
    &= 1 - \frac{N}{M},
\end{align*}
where the first equality uses $\widetilde{\Delta} = \widetilde{\Delta}\circ\Delta_M$ together with the idempotence of $\Delta_M$, and where we then used the triangle inequality, $\|y\otimes\rho_m\|_1 = \|y\|_1\tr(\rho_m)$ and $\|\proj{m} - \tau_m\|_1 = 2$ since $\proj{m}$ and $\tau_m$ are states with orthogonal supports.
Since the diamond norm of the Hermitian-preserving map $\Delta_M - \widetilde{\Delta}$ is attained on states (see e.g.~\cite[Chapter~3]{watrous18}), we obtain $\frac{1}{2}\|\Delta_M - \widetilde{\Delta}\|_\diamond \leq 1 - N/M$, so that $\deltaEA{1}{1}{\Delta_N}{\Delta_M} \leq \deltaSR{1}{1}{\Delta_N}{\Delta_M} \leq 1 - N/M$.

\emph{Converse.}
Let $\mathsf{P} = (\omega_{A_0B_0},\E,\D)$ be an arbitrary entanglement-assisted protocol, with $\E : \cL(\C^M\otimes\cH_{A_0}) \rightarrow \cL(\C^N)$ and $\D : \cL(\C^N\otimes\cH_{B_0}) \rightarrow \cL(\C^M)$, and let $\widetilde{\Delta} = \mathsf{P}[\Delta_N]$.
We first show that the average probability that $\widetilde{\Delta}$ transmits correctly a uniformly distributed symbol is at most $N/M$, that is
\begin{equation}\label{eq:EA_success_probability}
    \frac{1}{M}\sum_{m \in [M]}\bra{m}\widetilde{\Delta}(\proj{m})\ket{m} \leq \frac{N}{M}.
\end{equation}
For $m \in [M]$, let $Y_m = (\E\otimes\Id_{B_0})(\proj{m}\otimes\omega_{A_0B_0}) \in \cL(\C^N\otimes\cH_{B_0})$ and, for $z \in [N]$, let $\sigma_{m,z} = (\bra{z}\otimes\mathds{1}_{B_0})Y_m(\ket{z}\otimes\mathds{1}_{B_0}) \in \cL(\cH_{B_0})$.
Since $Y_m$ is positive semidefinite, so are the operators $\sigma_{m,z}$, and
\begin{equation*}
    (\Delta_N\otimes\Id_{B_0})(Y_m) = \sum_{z \in [N]}\proj{z}\otimes\sigma_{m,z}.
\end{equation*}
Furthermore, for every $w \in \cL(\C^M \otimes \cH_{A_0}\otimes\cH_{B_0})$, we have $\tr_{\C^N}\bigl((\E\otimes\Id_{B_0})(w)\bigr) = \tr_{\C^MA_0}(w)$: indeed, for every $v \in \cL(\cH_{B_0})$, $\tr\bigl(\tr_{\C^N}((\E\otimes\Id_{B_0})(w))\,v\bigr) = \tr\bigl((\E\otimes\Id_{B_0})(w)\,(\mathds{1}_{\C^N}\otimes v)\bigr) = \tr\bigl(w\,(\E^*(\mathds{1}_{\C^N})\otimes v)\bigr) = \tr\bigl(\tr_{\C^MA_0}(w)\,v\bigr)$, where the last equality uses that $\E^*$ is unital, which is equivalent to $\E$ being trace preserving.
Applying this with $w = \proj{m}\otimes\omega_{A_0B_0}$, we get
\begin{equation}\label{eq:no_signalling}
    \sum_{z \in [N]}\sigma_{m,z} = \tr_{\C^N}(Y_m) = \tr_{\C^M A_0}(\proj{m}\otimes\omega_{A_0B_0}) = \omega_{B_0},
\end{equation}
where $\omega_{B_0} = \tr_{A_0}(\omega_{A_0B_0})$.
This equation expresses the no-signalling principle: the local operation of Alice does not change the marginal state of Bob.
In particular, $0 \leq \sigma_{m,z} \leq \omega_{B_0}$ for all $m \in [M]$ and $z \in [N]$.
Now, for $z \in [N]$ and $m \in [M]$, let
\begin{equation*}
    D^{(z)}_m = (\bra{z}\otimes\mathds{1}_{B_0})\D^*(\proj{m})(\ket{z}\otimes\mathds{1}_{B_0}) \in \cL(\cH_{B_0}).
\end{equation*}
Since $\D^*$ is unital and completely positive, $D^{(z)}_m \geq 0$ and $\sum_{m \in [M]}D^{(z)}_m = (\bra{z}\otimes\mathds{1}_{B_0})\D^*(\mathds{1})(\ket{z}\otimes\mathds{1}_{B_0}) = \mathds{1}_{B_0}$, i.e.\ $\{D^{(z)}_m\}_{m \in [M]}$ is the POVM that Bob applies on $B_0$ when he receives the symbol $z$.
Then, for $m \in [M]$,
\begin{equation*}
    \bra{m}\widetilde{\Delta}(\proj{m})\ket{m} = \tr\Bigl(\D^*(\proj{m})\sum_{z \in [N]}\proj{z}\otimes\sigma_{m,z}\Bigr) = \sum_{z \in [N]}\tr(D^{(z)}_m\sigma_{m,z}) \leq \sum_{z \in [N]}\tr(D^{(z)}_m\omega_{B_0}),
\end{equation*}
where the inequality follows from $\sigma_{m,z} \leq \omega_{B_0}$ and $D^{(z)}_m \geq 0$.
Summing over $m$, we get
\begin{equation*}
    \sum_{m \in [M]}\bra{m}\widetilde{\Delta}(\proj{m})\ket{m} \leq \sum_{z \in [N]}\tr\Bigl(\sum_{m \in [M]}D^{(z)}_m\,\omega_{B_0}\Bigr) = \sum_{z \in [N]}\tr(\omega_{B_0}) = N,
\end{equation*}
which proves eq.~\eqref{eq:EA_success_probability}.

We then apply Lemma~\ref{lemma:channel_Holevo} to the channels $\Delta_M$ and $\widetilde{\Delta}$, with the auxiliary system $\cH_R = \C^M$ and
\begin{equation*}
    \sigma = \frac{1}{M}\sum_{m \in [M]}\proj{m}\otimes\proj{m}_R, \qquad \mu = \sum_{m \in [M]}\proj{m}\otimes\proj{m}_R.
\end{equation*}
Since $\mu$ is an orthogonal projector, $0 \leq \mu \leq \mathds{1}$. Furthermore, $(\Delta_M\otimes\Id_R)(\sigma) = \sigma$, so that $\tr(\mu(\Delta_M\otimes\Id_R)(\sigma)) = \tr(\mu\sigma) = 1$, and
\begin{equation*}
    \tr(\mu(\widetilde{\Delta}\otimes\Id_R)(\sigma)) = \frac{1}{M}\sum_{m \in [M]}\bra{m}\widetilde{\Delta}(\proj{m})\ket{m} \leq \frac{N}{M},
\end{equation*}
by eq.~\eqref{eq:EA_success_probability}.
Hence, by Lemma~\ref{lemma:channel_Holevo}, $\frac{1}{2}\|\Delta_M - \mathsf{P}[\Delta_N]\|_\diamond \geq 1 - N/M$.
As this holds for every entanglement-assisted protocol $\mathsf{P}$, and as the diamond norm is non-negative, we get $\deltaEA{1}{1}{\Delta_N}{\Delta_M} \geq \max\{0,1-N/M\}$, which ends the proof.
\end{proof}

Note that the achievability part of Lemma~\ref{lem:EA_classical_channels} only uses shared randomness, and that the converse holds with arbitrary shared entanglement, and even with arbitrary no-signalling correlations~\cite{Matthews12}.
Hence, for noiseless classical channels, shared entanglement does not help compared to shared randomness, and the minimal error in eq.~\eqref{eq:EA_classical_channels} is also the minimal error with shared randomness.

\subsection{Proof of Theorem~\ref{thm:main_EA}}

We are now ready to prove the two main results of this section.

\begin{proof}[Proof of item~(ii) of Theorem~\ref{thm:main_EA}]
Let $\F$ and $\cG$ be idempotent channels and $k,n \in \N_{>0}$.
By Proposition~\ref{prop:multiplicativity}, $\la(\F^{\otimes k}) = \la(\F)^{\otimes k}$, so that, by the multiplicativity of the $\ell_2$-norm under tensor products,
\begin{equation*}
    s(\F^{\otimes k}) = \|\la(\F)^{\otimes k}\|_2^2 = \|\la(\F)\|_2^{2k} = s(\F)^k,
\end{equation*}
and similarly $s(\cG^{\otimes n}) = s(\cG)^n$.
By Corollary~\ref{coro:EA_equivalence_classical} applied to the idempotent channels $\F^{\otimes k}$ and $\cG^{\otimes n}$, there are entanglement-assisted protocols $\mathsf{Q}_\F$, $\mathsf{Q}_\F'$, $\mathsf{Q}_\cG$ and $\mathsf{Q}_\cG'$ such that
\begin{equation}\label{eq:tensor_powers_as_classical}
    \F^{\otimes k} = \mathsf{Q}_\F[\Delta_{s(\F)^k}], \quad \Delta_{s(\F)^k} = \mathsf{Q}_\F'[\F^{\otimes k}], \quad \cG^{\otimes n} = \mathsf{Q}_\cG[\Delta_{s(\cG)^n}], \quad \Delta_{s(\cG)^n} = \mathsf{Q}_\cG'[\cG^{\otimes n}].
\end{equation}
Therefore, by eq.~\eqref{eq:delta_kn_vs_delta_11}, Lemma~\ref{lem:monotonicity_EA} and Lemma~\ref{lem:EA_classical_channels},
\begin{align*}
    \deltaEA{k}{n}{\cG}{\F} &= \deltaEA{1}{1}{\cG^{\otimes n}}{\F^{\otimes k}} = \deltaEA{1}{1}{\Delta_{s(\cG)^n}}{\Delta_{s(\F)^k}} \\
    &= \max\Bigl\{0,1-\frac{s(\cG)^n}{s(\F)^k}\Bigr\} = \max\Bigl\{0,1 - \frac{\|\la(\cG)\|_2^{2n}}{\|\la(\F)\|_2^{2k}}\Bigr\}.
\end{align*}
\end{proof}

Note that the proof of item~(ii) of Theorem~\ref{thm:main_EA} also shows that, when $\|\la(\cG)\|_2^{2n} \geq \|\la(\F)\|_2^{2k}$, the minimal error $\deltaEA{k}{n}{\cG}{\F} = 0$ is \emph{attained}, i.e.\ there is an entanglement-assisted protocol $\mathsf{P}$ such that $\F^{\otimes k} = \mathsf{P}[\cG^{\otimes n}]$.
Indeed, in that case, $s(\cG)^n \geq s(\F)^k$, so that $\Delta_{s(\F)^k} = \mathsf{R}[\Delta_{s(\cG)^n}]$ for some entanglement-assisted protocol $\mathsf{R}$ by Lemma~\ref{lem:EA_classical_channels}.
Composing with eq.~\eqref{eq:tensor_powers_as_classical} gives $\F^{\otimes k} = \mathsf{Q}_\F[\mathsf{R}[\mathsf{Q}_\cG'[\cG^{\otimes n}]]]$, and the right-hand side is of the form $\mathsf{P}[\cG^{\otimes n}]$ by eq.~\eqref{eq:composition_EA}.

\begin{remark}
The protocol $\mathsf{P}$ above consumes only maximally entangled states: by eq.~\eqref{eq:composition_EA}, its shared state is the tensor product of the shared states of $\mathsf{Q}_\F$, $\mathsf{R}$ and $\mathsf{Q}_\cG'$, which are tensor products of finitely many maximally entangled states of finite dimension by Proposition~\ref{prop:EA_equivalence_pinching}, the protocol $\mathsf{R}$ being unassisted.
Hence the exact interconversion of idempotent channels never requires entanglement beyond a finite supply of maximally entangled states, a feature of interest from the point of view of resource theories, where such states are the standard unit of entanglement.
\end{remark}

\begin{proof}[Proof of item~(i) of Theorem~\ref{thm:main_EA}]
Let $\F$ and $\cG$ be idempotent channels with $\la(\F) \neq (1)$, so that $s(\F) \geq 2$, and let
\begin{equation*}
R = \frac{\log(s(\cG))}{\log(s(\F))} = \frac{\log(\|\la(\cG)\|_2)}{\log(\|\la(\F)\|_2)},
\end{equation*}
where we used $\log(s(\cG)) = 2\log(\|\la(\cG)\|_2)$ and $\log(s(\F)) = 2\log(\|\la(\F)\|_2) > 0$.

\emph{Converse.} Let $r$ be an achievable rate in the entanglement-assisted regime, and let $(k_\nu)_\nu$ and $(n_\nu)_\nu$ be two sequences of positive integers satisfying eq.~\eqref{eq:def_achievable_rate} with $\delta_\nu = \deltaEA{k_\nu}{n_\nu}{\cG}{\F} \to 0$.
By item~(ii) of Theorem~\ref{thm:main_EA}, $s(\cG)^{n_\nu}/s(\F)^{k_\nu} \geq 1 - \delta_\nu$ for every $\nu$.
For $\nu$ large enough, $\delta_\nu < 1$, and, taking the logarithm and dividing by $n_\nu\log(s(\F)) > 0$,
\begin{equation*}
    \frac{k_\nu}{n_\nu} \leq R - \frac{\log(1 - \delta_\nu)}{n_\nu\log(s(\F))}.
\end{equation*}
The last term vanishes as $\nu \to \infty$, so that $r \leq \liminf_\nu k_\nu/n_\nu \leq R$.
Hence $\CEA{\cG}{\F} \leq R$, and the zero-error entanglement-assisted emulation capacity is at most $R$ as well, by Section~\ref{subsec:emulation}.

\emph{Achievability with zero error.} By the discussion following the proof of item~(ii) of Theorem~\ref{thm:main_EA}, for $k,n \in \N_{>0}$, $\F^{\otimes k}$ can be emulated exactly by $\cG^{\otimes n}$ with shared entanglement as soon as $s(\cG)^n \geq s(\F)^k$, that is as soon as $k/n \leq R$.
If $\la(\cG) = (1)$, then $R = 0$ and $\CEA{\cG}{\F} = 0$, so that both capacities vanish, in accordance with eq.~\eqref{eq:capacity_zero_error_EA}.
Otherwise, $R > 0$ and, by density of the rational numbers in the reals, the supremum of the ratios $k/n \leq R$ is $R$, so that the zero-error entanglement-assisted emulation capacity is at least $R$.
Since it is at most $\CEA{\cG}{\F} \leq R$, both capacities are equal to $R$, which proves eq.~\eqref{eq:capacity_zero_error_EA} and the last statement of item~(i).
\end{proof}

As in the unassisted and shared randomness regimes (see the end of Section~\ref{sec:bound_finite_blocklength_UA_SR}), the finite blocklength result of item~(ii) directly yields the strong converse of item~(iii), here with an explicit exponent.
No optimisation over $p$ is needed, since the only relevant value is $p = 2$.

\begin{proof}[Proof of item~(iii) of Theorem~\ref{thm:main_EA}]
Let $\F$ and $\cG$ be two idempotent channels such that $\la(\F) \neq (1)$, so that $\|\la(\F)\|_2 > 1$, and let $\varepsilon > 0$.
Throughout this proof, we take the logarithm in base $2$.
Let $k,n \in \N_{>0}$ with $k/n \geq \CEA{\cG}{\F} + \varepsilon$.
Multiplying this inequality by $n\log(\|\la(\F)\|_2) > 0$ and using eq.~\eqref{eq:capacity_zero_error_EA}, we get
\begin{equation*}
    k\log(\|\la(\F)\|_2) \geq n\log(\|\la(\cG)\|_2) + \varepsilon n\log(\|\la(\F)\|_2),
\end{equation*}
and thus, by eq.~\eqref{eq:error_approx_case_EA},
\begin{equation*}
    \deltaEA{k}{n}{\cG}{\F} \geq 1 - \frac{\|\la(\cG)\|_2^{2n}}{\|\la(\F)\|_2^{2k}} = 1 - 2^{2(n\log(\|\la(\cG)\|_2) - k\log(\|\la(\F)\|_2))} \geq 1 - 2^{-cn},
\end{equation*}
with $c = 2\varepsilon\log(\|\la(\F)\|_2) > 0$, which proves item~(iii).
\end{proof}

Since eq.~\eqref{eq:error_approx_case_EA} is an equality, the exponent $c = 2\varepsilon\log(\|\la(\F)\|_2)$ is optimal: for $k/n = \CEA{\cG}{\F} + \varepsilon$, the error is exactly $1 - 2^{-cn}$.

\subsection{Consequences and examples}

We now draw some consequences of Theorem~\ref{thm:main_EA}, and compare them with the unassisted and shared randomness regimes.
First, the entanglement-assisted emulation capacity is reversible and additive.

\begin{corollary}\label{coro:EA_reversibility_additivity}
Let $\F$, $\cG_1$ and $\cG_2$ be idempotent channels with $\la(\F) \neq (1)$.
\begin{enumerate}[label=(\roman*)]
    \item If $\la(\cG_1) \neq (1)$, then $\CEA{\cG_1}{\F}\,\CEA{\F}{\cG_1} = 1$.
    \item $\CEA{\cG_1\otimes\cG_2}{\F} = \CEA{\cG_1}{\F} + \CEA{\cG_2}{\F}$.
\end{enumerate}
\end{corollary}

\begin{proof}
(i) follows directly from eq.~\eqref{eq:capacity_zero_error_EA}, since $\log(\|\la(\F)\|_2)$ and $\log(\|\la(\cG_1)\|_2)$ are both positive.
For (ii), $\cG_1 \otimes \cG_2$ is idempotent and, by Proposition~\ref{prop:multiplicativity}, $\la(\cG_1\otimes\cG_2) = \la(\cG_1)\otimes\la(\cG_2)$, so that $\log(\|\la(\cG_1\otimes\cG_2)\|_2) = \log(\|\la(\cG_1)\|_2) + \log(\|\la(\cG_2)\|_2)$, and (ii) follows from eq.~\eqref{eq:capacity_zero_error_EA}.
\end{proof}

Second, exact emulation with a single use of the source channel is governed by a single integer.

\begin{corollary}\label{coro:EA_exact_emulation}
Let $\F$ and $\cG$ be idempotent channels.
Then there is an entanglement-assisted protocol $\mathsf{P}$ such that $\F = \mathsf{P}[\cG]$ if and only if $\|\la(\F)\|_2 \leq \|\la(\cG)\|_2$.
In particular, there are entanglement-assisted protocols $\mathsf{P}$, $\mathsf{P}'$ with $\F = \mathsf{P}[\cG]$ and $\cG = \mathsf{P}'[\F]$ if and only if $\|\la(\F)\|_2 = \|\la(\cG)\|_2$.
\end{corollary}

\begin{proof}
If $\F = \mathsf{P}[\cG]$ for some entanglement-assisted protocol $\mathsf{P}$, then $\deltaEA{1}{1}{\cG}{\F} = 0$, so that $\|\la(\cG)\|_2^2 \geq \|\la(\F)\|_2^2$ by item~(ii) of Theorem~\ref{thm:main_EA} with $k = n = 1$.
Conversely, if $\|\la(\F)\|_2 \leq \|\la(\cG)\|_2$, then such a protocol exists by the discussion following the proof of item~(ii) of Theorem~\ref{thm:main_EA}.
The second statement follows by applying the first one in both directions.
\end{proof}

Corollary~\ref{coro:EA_exact_emulation} shows that, in the presence of shared entanglement, idempotent channels are totally preordered by exact emulation with a single use, and that their equivalence classes are labelled by the single integer $\|\la(\F)\|_2^2$.
This contrasts with the unassisted regime, where inequalities $\|\la(\F)\|_p \leq \|\la(\cG)\|_p$ for every $p \in [1,+\infty]$ are necessary for exact unassisted emulation after taking a common tensor power, by item~(ii) of Theorem~\ref{thm:main_UA_SR}. 
For instance, let $\cG$ be an idempotent channel with $\la(\cG) = (2,2,1)$.
Then $\|\la(\Id_3)\|_\infty = 3 > 2 = \|\la(\cG)\|_\infty$ and $\|\la(\cG)\|_1 = 5 > 3 = \|\la(\Id_3)\|_1$, so that neither of $\Id_3$ and $\cG$ can be emulated with zero error by one use of the other, even with shared randomness.
However, $\|\la(\Id_3)\|_2^2 = 9 = \|\la(\cG)\|_2^2$, so that $\Id_3$ and $\cG$ can be exactly emulated by each other with shared entanglement, by Corollary~\ref{coro:EA_exact_emulation}.

As in the unassisted case, we gather in Table~\ref{tab:EA_capacities} the expressions of $\CEA{\cG}{\F}$ when either $\F$ or $\cG$ is an identity or a completely dephasing channel, which follow directly from item~(i) of Theorem~\ref{thm:main_EA} and from $\|\la(\Id_d)\|_2 = d$, $\|\la(\Delta_d)\|_2 = \sqrt{d}$.
In particular, for every idempotent channel $\cG$, the entanglement-assisted quantum and classical capacities of $\cG$ (in qubits and bits per channel use) are respectively $\CEA{\cG}{\Id_2} = \log_2(\|\la(\cG)\|_2)$ and $\CEA{\cG}{\Delta_2} = 2\log_2(\|\la(\cG)\|_2)$, as expected from teleportation and superdense coding.

\begin{table}[h!]
    \centering
    \begin{tabular}{|c|c|c|}
        \hline
        $\F$ & $\cG$ & $\CEA{\cG}{\F}$  \rule{0pt}{13pt}\\
        \hline
        $\Id_d$ & $\cG$ & $\frac{\log(\|\lambda(\cG)\|_2)}{\log(d)}$ \rule{0pt}{11pt}\\[4pt]
        $\F$ & $\Id_d$ & $\frac{\log(d)}{\log(\|\lambda(\F)\|_2)}$ \\[4pt]
        $\Delta_d$ & $\cG$ & $\frac{2\log(\|\lambda(\cG)\|_2)}{\log(d)}$ \\[4pt]
        $\F$ & $\Delta_d$ & $\frac{\log(d)}{2\log(\|\lambda(\F)\|_2)}$ \\[4pt]
        \hline
    \end{tabular}
    \caption{Entanglement-assisted emulation capacities when either $\F$ or $\cG$ is the identity or the completely dephasing channel (with $d \geq 2$ and $\la(\F) \neq (1)$).}
    \label{tab:EA_capacities}
\end{table}

The following examples illustrate how shared entanglement modifies the emulation capacities of idempotent channels.

\begin{example}[Emulation of a quantum channel by a classical one]\label{ex:EA_teleportation}
Let $\F = \Id_2$ and $\cG = \Delta_2$, so that $\la(\F) = (2)$ and $\la(\cG) = (1,1)$.
As computed in Example~\ref{ex:UA_teleportation}, $\CUA{\Delta_2}{\Id_2} = \CSR{\Delta_2}{\Id_2} = 0$ and $\CUA{\Id_2}{\Delta_2} = 1$: a noiseless qubit channel cannot be emulated with zero error by a classical channel, even with shared randomness.
In contrast, by Table~\ref{tab:EA_capacities},
\begin{equation*}
    \CEA{\Delta_2}{\Id_2} = \frac{\log(\sqrt{2})}{\log(2)} = \frac{1}{2}, \qquad \CEA{\Id_2}{\Delta_2} = \frac{2\log(2)}{\log(2)} = 2.
\end{equation*}
Operationally, two uses of $\Delta_2$ transmit the two bits needed to teleport a qubit, and one use of $\Id_2$ transmits two bits with superdense coding.
Shared entanglement is known to strictly increase zero-error rates already for classical channels, both in the one-shot~\cite{CLMW10} and in the asymptotic~\cite{LMMOR12} regimes; this example is an extreme instance of this phenomenon, since the emulation capacity jumps from $0$ to $1/2$.
Note that the product of these two capacities is $1$, in accordance with Corollary~\ref{coro:EA_reversibility_additivity} (i), whereas the corresponding product is $0$ in the unassisted regime.
\end{example}

\begin{example}[Complementary source channels]\label{ex:EA_complementary}
Let $\cG_1 = \Id_2$, $\cG_2 = \Delta_2$ and $\F = \cG_1 \otimes \cG_2$, so that
\begin{equation*}
    \la(\cG_1) = (2), \qquad \la(\cG_2) = (1,1), \qquad \la(\F) = \la(\cG_1)\otimes\la(\cG_2) = (2,2).
\end{equation*}
As computed in Example~\ref{ex:UA_complementary}, the superadditivity of Corollary~\ref{coro:additivity} is strict for these channels.
With shared entanglement, since $\|\la(\cG_1)\|_2^2 = 4$, $\|\la(\cG_2)\|_2^2 = 2$ and $\|\la(\F)\|_2^2 = 8$, we get instead
\begin{equation*}
    \CEA{\cG_1}{\F} = \frac{\log(4)}{\log(8)} = \frac{2}{3}, \qquad \CEA{\cG_2}{\F} = \frac{\log(2)}{\log(8)} = \frac{1}{3}, \qquad \CEA{\cG_1\otimes\cG_2}{\F} = 1,
\end{equation*}
which is additive, in accordance with Corollary~\ref{coro:EA_reversibility_additivity} (ii).
Entanglement therefore restores additivity here: the classical source channel $\cG_2$, which was useless on its own without entanglement, now contributes the missing third of the rate.
\end{example}

\section*{Acknowledgements}
 We thank Robert Salzmann for helpful discussions about~\cite{kitaev25} and \cite{Kretschman04}, as well as about sets of approximate fixed points of channels. We thank Satvik Singh and Arturo Konderak for pointing out a mistake in Proposition~\ref{prop:structure_idempotent_channel} in the first version of the present article. We also thank Arturo Konderak for drawing our attention to~\cite{Amato.2025}.

ID and OF acknowledge funding by the European Research Council (ERC Grant AlgoQIP, Agreement No. 851716). LG acknowledges funding by National Natural Science Foundation of China (NSFC) via grant No. 12401163.
This work started when MR was a Marie Sk\l odowska-Curie Fellow at ENS de Lyon, and he acknowledges funding from the European Union’s Horizon Europe research and innovation programme, 
call HORIZON-MSCA-2022-PF-01 (Project and agreement number: 101108117). MR is partially supported by
the Wallenberg Centre for Quantum Technology (WACQT) funded by the Knut and Alice
Wallenberg Foundation (KAW).

\subsection*{Use of AI tools}

This article supersedes two earlier preprint versions on arXiv, as well as a short version presented at ISIT 2026~\cite{self.short}, which were prepared without the assistance of generative AI, except for minor polishing. In those versions, we established a weaker form of item~(ii) of Theorem~\ref{thm:main_UA_SR} where the infimum over $p \in [1,+\infty]$ was restricted to the set $\{1,+\infty\}$, leaving the full interval as an open problem. Considering recent developments in the use of AI to assist, if not automate, mathematical research, we tasked Anthropic's Claude Opus 5 with solving it. The model produced the statement and proofs of Lemmas~\ref{lem:cauchy_schwarz},~\ref{lem:bilinear}, and~\ref{lem:key} (without shared randomness), as well as the ansatz in Definition~\ref{def:ansatz}.

Furthermore, we asked OpenAI's Astra model to find the algorithms of Section~\ref{sec:algorithm}, as well as to prove their correctness and to analyse their complexity, which it achieved. 
Interestingly, Astra itself defined the relevant computational model.
Astra also solved the entanglement-assisted regime.
Finally, most of the writing of Sections~\ref{sec:algorithm} and~\ref{sec:EA} was left to Anthropic's Claude Opus 5, based on the material produced by OpenAI's Astra model.

The authors independently verified all AI-generated steps and assume full responsibility for the mathematical correctness of the results.

Marginally, Google's Gemini Flash was used to assist with the linguistic revision and polishing of the remaining, human-written parts of the manuscript.

\bibliographystyle{quantum}
\bibliography{Quantum_bibliography}

\appendix

\section{An algebraic proof of the converse of item~(i) of Theorem~\ref{thm:main_UA_SR} in the unassisted setting}\label{sec:algebraic_converse}

Since a reduced channel $\widehat{\F}$ has an adjoint $\widehat{\F}^*$ which projects its input onto a finite-dimensional $*$-algebra, it is natural to ask whether it is possible to approach the problem of the unassisted emulation of idempotent channels through purely algebraic arguments. 
We provide such a proof of the converse of item~(i) of Theorem~\ref{thm:main_UA_SR} for exact unassisted emulation in this section, that is, we show that for any two idempotent channels $\F$ and $\cG$ with $\lambda(\F) \neq (1)$, the zero-error unassisted emulation capacity of $\F$ by $\cG$, i.e.\ the supremum of the ratios $k/n$ for which $\F^{\otimes k} = \D\cG^{\otimes n}\E$ for some channels $\D$ and $\E$, is at most
\begin{equation}
    \inflognormratio{\cG}{\F}.
\end{equation}
This algebraic approach yields a slightly stronger version of the converse of item~(i) of Theorem~\ref{thm:main_UA_SR} for the interconversion of reduced idempotent channels when allowing for dimension-dependent errors.
We provide the latter theorem and its proof in Appendix~\ref{sec:approximate_algebras}, which is an application of approximate $C^*$-algebras.
These algebraic structures were recently introduced by Kitaev in~\cite{kitaev25}.

To prove this direction, we use a result of~\cite{kuperberg} which is almost a converse to Theorem~\ref{thm:kuperberg_arxiv}.

\begin{theorem}[Forward direction of {\cite[Theorem~1.1]{kuperberg}}]\label{thm:converse_kuperberg}
Let $\A$ and $\B$ be two finite-dimensional $*$-algebras.
If there is an injective $*$-homomorphism of $\A$ into $\B$, then for all $p \in [1,+\infty]$,
\begin{equation}
    \|\la(\A)\|_p \leq \|\la(\B)\|_p.
\end{equation}
\end{theorem}

Note that Theorems~\ref{thm:kuperberg_arxiv} and~\ref{thm:converse_kuperberg} are not converse to one another, because the inequality between the $\ell_p$-norms of the shape vectors is strict in one case but not in the other.

The converse of item~(i) of Theorem~\ref{thm:main_UA_SR} then follows from Theorem~\ref{thm:converse_kuperberg} combined with the following proposition.

\begin{prop}\label{prop:encoding_implies_morphism}
    Let $\F : \cL(\cH_\F) \rightarrow \cL(\cH_\F)$ and $\cG : \cL(\cH_\cG) \rightarrow \cL(\cH_\cG)$ be two idempotent channels equal to their reduced channels, i.e.\ $\F = \widehat{\F}$ and $\cG = \widehat{\cG}$. Suppose that there are two channels $\E$, $\D$ satisfying
    \begin{equation}
        \F = \D\cG\E.
    \end{equation}
    We write $\mathds{1}_\F$ and $\mathds{1}_\cG$ for the identity operators in $\cL(\cH_\F)$ and $\cL(\cH_\cG)$ respectively. To simplify notation, we also write $e_{\cG\E} = \cG\E(\mathds{1}_\F)^0$. Define
    \begin{equation}
        \widetilde{\cG}^* : x \mapsto e_{\cG\E}\cG^*(x)e_{\cG\E}.
    \end{equation}
    Then, $e_{\cG\E}\Rg(\cG^*)e_{\cG\E}$ is a unital $*$-subalgebra of $\cL(\cH_{\cG})$ whose identity operator is $e_{\cG\E}$ and $\widetilde{\cG}^*\D^*$ restricted to the $*$-algebra $\Rg(\F^*)$ is an injective unital $*$-homomorphism which embeds $\Rg(\F^*)$ into $e_{\cG\E}\Rg(\cG^*)e_{\cG\E}$ as a subalgebra.

    Furthermore, for all $p \in [1,+\infty]$ 
    \begin{equation*}\|\la(e_{\cG\E}\Rg(\cG^*)e_{\cG\E})\|_p \leq \|\la(\Rg(\cG^*))\|_p.\end{equation*}
\end{prop}

The proof of this proposition uses the following lemma.

\begin{lemma}\label{lemma:image_of_identity}
Let $\Phi: \cL(\cH) \rightarrow \cL(\cH')$ be a quantum channel and write $e = \Phi(\mathds{1})^0$. Then,
\begin{enumerate}[label=\roman*]
\item[i)] for all $x \in \cL(\cH')$, $\Phi^*(x) = \Phi^*(exe)$,
\item[ii)] for all positive semidefinite $x \in \cL(\cH')$, $\Phi^*(x) = 0$ implies $exe = 0$.
\end{enumerate}
\end{lemma}

\begin{proof}
Let $\cH_0' = \supp(\Phi(\mathds{1}))$ and $V : \cH_0' \rightarrow\cH'$ be the isometry such that $e = VV^*$. Let $y \in \cL(\cH)$ and $x \in \cL(\cH')$, then
\begin{equation*}\tr(y\Phi^*(x))= \tr(\Phi(y)x) = \tr(VV^*\Phi(y)VV^*x) = \tr(e\Phi(y)ex)= \tr(y\Phi^*(exe)),\end{equation*}
where the second equality follows by Lemma~\ref{lemma:stability_support_one}. Therefore $\Phi^*(x) = \Phi^*(exe)$.

For the second statement, let $x \in \cL(\cH')$ be positive semidefinite such that $\Phi^*(x) = 0$. Then
\begin{equation*}
0 = \tr(\Phi^*(x)) = \tr(\Phi(\mathds{1})x) = \tr(e\Phi(\mathds{1})e x) \geq \lambda_{\min}(\Phi(\mathds{1}))\tr(exe),
\end{equation*}
with $\lambda_{\min}(\Phi(\mathds{1}))$ the minimal non-zero eigenvalue of $\Phi(\mathds{1})$. Therefore, $\tr(exe) = 0$ and, by positivity of $x$, $exe = 0$.
\end{proof}

Then, we move on to the proof of Proposition~\ref{prop:encoding_implies_morphism}. This proof uses the Kadison--Schwarz inequality \cite[eq.~(5.2)]{Wolf12}, which is a generalisation of the well-known Cauchy--Schwarz inequality to maps acting on operators. For every channel $\F: \cL(\cH) \rightarrow \cL(\cH')$, for all $x \in \cL(\cH')$,
\begin{equation}\label{eq:kadison-schwarz}
    \F^*(x^*)\F^*(x) \leq \F^*(x^* x).
\end{equation}

Furthermore, in this proof, we will use the notion of \emph{multiplicative domain} of linear maps. Given $\F^* : \cL(\cH') \rightarrow \cL(\cH)$ a linear map, the multiplicative domain of $\F^*$ is the set 
\begin{equation*}
\{x \in \cL(\cH')\mid \forall y \in \cL(\cH'),~\F^*(xy) = \F^*(x)\F^*(y) \text{ and } \F^*(yx) = \F^*(y)\F^*(x)\}.
\end{equation*}
If $\F^*$ is furthermore Hermitian preserving, then its multiplicative domain is closed under taking the adjoint and is therefore a $*$-algebra. Moreover, in this case, $\F^*$ acts as a (not necessarily unital) $*$-homomorphism on its multiplicative domain. 
Furthermore, when $\F^*$ satisfies the Kadison--Schwarz inequality, by~\cite[Theorem~5.7]{Wolf12}, we can write its multiplicative domain as 
\begin{equation*}
\{x \in \cL(\cH')\mid \F^*(x^* x) = \F^*(x^*)\F^*(x) \text{ and } \F^*(xx^*) = \F^*(x)\F^*(x^*)\}.
\end{equation*}
In particular, if $\F^*$ is unital completely positive\footnote{In fact, it suffices for $\F^*$ to be subunital $2$-positive for it to satisfy the Kadison--Schwarz inequality, see~\cite[Section~5.2]{Wolf12} for more details.}, it satisfies the Kadison--Schwarz inequality (eq.~\eqref{eq:kadison-schwarz}) and we can use this second characterisation of its multiplicative domain.

\begin{proof}[Proof of Proposition~\ref{prop:encoding_implies_morphism}]
    Let $\F : \cL(\cH_\F) \rightarrow \cL(\cH_\F)$ and $\cG : \cL(\cH_\cG) \rightarrow \cL(\cH_\cG)$ be two idempotent channels equal to their reduced channels, i.e.\ $\F = \widehat{\F}$ and $\cG = \widehat{\cG}$, and let $\D$, $\E$ be two channels such that
    \begin{equation*}
    \F = \D\cG\E.
    \end{equation*}
    Taking the adjoint of this equation, we get
    \begin{equation*}\F^* = \E^*\cG^*\D^* = \E^*\cG^*\cG^*\D^*,\end{equation*}
    where the right-hand side follows from the idempotence of $\cG^*$. 
    Let $x \in \cL(\cH_\F)$, applying Lemma~\ref{lemma:image_of_identity} to the channel $\Phi = \cG\E$, we get that 
    \begin{equation*}
        \F^*(x) = \E^*\cG^*(e_{\cG\E}\cG^*\D^*(x)e_{\cG\E}) = \E^*\cG^*\widetilde{\cG}^*\D^*(x),
    \end{equation*}
    with $\widetilde{\cG}^*$ defined in the statement of Proposition~\ref{prop:encoding_implies_morphism}. 
    Let $x \in \Rg(\F^*)$. By the assumption that $\F = \widehat{\F}$, $\Rg(\F^*)$ is a $*$-subalgebra of $\cL(\cH_\F)$, thus $x^* x \in \Rg(\F^*)$, and as $\F^*$ is idempotent, $x^*x$ is a fixed point of $\F^*$, so that we have:
    \begin{align*}
        x^* x &= \F^*(x^* x) \\
        &= \E^*\cG^*\widetilde{\cG}^*\D^*(x^* x) \\
        &\overset{(a)}{\geq} \E^*\cG^*(\widetilde{\cG}^*\D^*(x^*)\widetilde{\cG}^*\D^*(x)) \\
        &\overset{(b)}{\geq} \E^*\cG^*(\widetilde{\cG}^*\D^*(x^*))\E^*\cG^*(\widetilde{\cG}^*\D^*(x)) \\
        &= x^* x,
    \end{align*}
    where we used the Kadison--Schwarz inequality to get both $(a)$ and $(b)$. Since the left-hand side and the right-hand side of this chain of inequalities coincide, we conclude that $(a)$ and $(b)$ are in fact equalities. Therefore, 
    \begin{equation*}
    \E^*\cG^*(\widetilde{\cG}^*\D^*(x^* x) - \widetilde{\cG}^*\D^*(x^*)\widetilde{\cG}^*\D^*(x)) = 0,
    \end{equation*}
    where, by the Kadison--Schwarz inequality,
    \begin{equation*}\widetilde{\cG}^*\D^*(x^* x) - \widetilde{\cG}^*\D^*(x^*)\widetilde{\cG}^*\D^*(x) \geq 0.\end{equation*}
    Therefore, by the second statement in Lemma~\ref{lemma:image_of_identity} applied to the channel $\Phi  = \cG\E$, we obtain 
    \begin{equation*}\widetilde{\cG}^*\D^*(x^* x) - \widetilde{\cG}^*\D^*(x^*)\widetilde{\cG}^*\D^*(x) = 0.\end{equation*}
    We can swap the roles of $x$ and $x^*$ in $(a)$ and $(b)$ to find that 
    \begin{equation*}\widetilde{\cG}^*\D^*(x x^*) - \widetilde{\cG}^*\D^*(x)\widetilde{\cG}^*\D^*(x^*) = 0.\end{equation*}
    Therefore, $x$ is in the multiplicative domain of $\widetilde{\cG}^*\D^*$. Thus, $\widetilde{\cG}^*\D^*$ restricted to $\Rg(\F^*)$ is an injective unital $*$-homomorphism into the $*$-algebra $e_{\cG\E}\Rg(\cG^*)e_{\cG\E}$, which is a $*$-subalgebra of $\cL(\cH_\cG)$ and whose identity operator is $e_{\cG\E}$. It is injective as, for all $x \in \Rg(\F^*)$, 
    \begin{equation*}x = \F^*(x) = \E^*\cG^*\widetilde{\cG}^*\D^*(x),\end{equation*}
    hence $\widetilde{\cG}^*\D^*(x) = 0$ implies $x=0$. It is unital as it maps the identity of $\cL(\cH_\F)$ to $e_{\cG\E}$ which is the identity operator of the $*$-algebra $e_{\cG\E}\Rg(\cG^*)e_{\cG\E}$. Indeed, as both $\cG$ and $\D$ are channels, $\cG^*$ and $\D^*$ are unital so that $\cG^*\D^*(\mathds{1}_\F ) = \mathds{1}_\cG$ and $\widetilde{\cG}^*\D^*(\mathds{1}_\F) = e_{\cG\E}$.

    Now, to show that $e_{\cG\E}\Rg(\cG^*)e_{\cG\E}$ is a $*$-subalgebra of $\cL(\cH_{\cG})$ which satisfies for all $p \in [1,+\infty]$, 
    ${\|\la(e_{\cG\E}\Rg(\cG^*)e_{\cG\E})\|_p \leq \|\la(\Rg(\cG^*))\|_p}$, it suffices to remark that, by Proposition~\ref{prop:structure_idempotent_channel}, we can decompose the Hilbert space $\cH_\cG$ as 
    \begin{equation*}
    \cH_\cG = \bigoplus_{k=1}^K \cH_{k,1} \otimes \cH_{k,2},
    \end{equation*}
    and the $*$-algebra $\Rg(\cG^*)$ as
    \begin{equation*}
    \Rg(\cG^*) = \bigoplus_{k=1}^K \cL(\cH_{k,1}) \otimes \mathds{1}_{m_k},
    \end{equation*}
    with $m_k = \dim(\cH_{k,2})$ for all $k \in [K]$. Then $\cG\E(\mathds{1}_\F) \in \Rg(\cG)$ can itself be decomposed as 
    \begin{equation*}
    \cG\E(\mathds{1}_\F) = \sum_{k=1}^K x_k \otimes \rho_k,
    \end{equation*}
    with $\{x_k \in \cL(\cH_{k,1}) \mid  k \in [K]\}$ and $\{\rho_k \in \D(\cH_{k,2})\mid k \in [K]\}$. 
    Thus,
    \begin{equation*}e_{\cG\E} = \cG\E(\mathds{1}_\F)^0 = \sum_{k=1}^K x_k^0 \otimes \rho_k^0.\end{equation*}
    Since $\cG$ is reduced, $\cG(\mathds{1}_\cG)$ is full-rank, so that for all $k \in [K]$, $\rho_k^0 = \mathds{1}_{m_k}$.
    Then,
    \begin{equation*}
        e_{\cG\E} = \sum_{k=1}^K x_k^0 \otimes \mathds{1}_{m_k} \in \Rg(\cG^*),
    \end{equation*}
    and 
    \begin{equation*}
    e_{\cG\E}\Rg(\cG^*)e_{\cG\E} = \bigoplus_{k=1}^K x_k^0 \cL(\cH_{k,1})x_k^0 \otimes \rho_k^0.
    \end{equation*}
    Hence, $e_{\cG\E}\Rg(\cG^*)e_{\cG\E}$ is a $*$-subalgebra whose shape vector consists of the non-zero ranks of the projections $x_k^0$, arranged in non-increasing order.
Since $\operatorname{rank}(x_k^0) \leq \dim(\cH_{k,1})$ for every $k \in [K]$, we obtain, for every $p \in [1,+\infty)$,
\begin{equation*}\|\la(e_{\cG\E}\Rg(\cG^*)e_{\cG\E})\|_p^p = \sum_{k=1}^K \operatorname{rank}(x_k^0)^p \leq \sum_{k=1}^K \dim(\cH_{k,1})^p = \|\la(\Rg(\cG^*))\|_p^p.
\end{equation*}
For $p=+\infty$, the same conclusion follows by taking the maximum instead of the sum.
This proves the claimed inequality.
\end{proof}

We can now prove the converse of item~(i) of Theorem~\ref{thm:main_UA_SR} for exact unassisted emulation.

\begin{proof}[Proof of the converse for exact unassisted emulation]
Let $\F : \cL(\cH_\F) \rightarrow \cL(\cH_\F)$ and $\cG : \cL(\cH_\cG) \rightarrow \cL(\cH_\cG)$ be two idempotent channels, let $k,n \in \N_{>0}$ such that there are two channels $\E$ and $\D$ satisfying
\begin{equation*}
\F^{\otimes k} = \D\cG^{\otimes n}\E.
\end{equation*}
By Corollary~\ref{coro:equivalence_reduced_version}, then, if we denote $\widehat{\F^{\otimes k}}$ and $\widehat{\cG^{\otimes n}}$ the reduced channels associated to $\F^{\otimes k}$ and $\cG^{\otimes n}$, this condition is equivalent to 
\begin{equation*}
\widehat{\F^{\otimes k}} = \D'\widehat{\cG^{\otimes n}}\E',
\end{equation*}
for some channels $\D'$, $\E'$.
Thus, by Proposition~\ref{prop:encoding_implies_morphism}, there is an injective unital $*$-homomorphism from $\Rg(\widehat{\F^{\otimes k}}^*)$ into the $*$-algebra $e_{\widehat{\cG^{\otimes n}}\E'}\Rg(\widehat{\cG^{\otimes n}}^*)e_{\widehat{\cG^{\otimes n}}\E'}$, which is itself a $*$-subalgebra of $\cL(\cH^{\otimes n}_\cG)$. Then, by Theorem~\ref{thm:converse_kuperberg} and the second statement of Proposition~\ref{prop:encoding_implies_morphism}, we have for all $p \in [1,+\infty]$:
\begin{equation*}
 \|\la(\widehat{\F^{\otimes k}})\|_p = \|\la(\Rg(\widehat{\F^{\otimes k}}^*))\|_p \leq \|\la(e_{\widehat{\cG^{\otimes n}}\E'}\Rg(\widehat{\cG^{\otimes n}}^*)e_{\widehat{\cG^{\otimes n}}\E'})\|_p \leq \|\la(\Rg(\widehat{\cG^{\otimes n}}^*))\|_p = \|\la(\widehat{\cG^{\otimes n}})\|_p.
\end{equation*}
By the multiplicativity property of the shape vectors and the $\ell_p$-norms with the tensor product (Proposition~\ref{prop:multiplicativity}), we get for all $p \in [1,+\infty]$:
\begin{equation*}\|\la(\F)\|_p^k = \|\la(\widehat{\F})\|_p^k \leq \|\la(\widehat{\cG})\|_p^n = \|\la(\cG)\|^n_p.\end{equation*}
Since $\lambda(\F)\neq(1)$, we have $\log(\|\lambda(\F)\|_p)>0$ for every finite $p\geq1$.
The same holds at $p=+\infty$ unless $\F$ is classical.
In that case, we omit $p=+\infty$, following the convention of Section~\ref{sec:emul_capa}.
Taking logarithms and dividing therefore gives,
\begin{equation*}\frac{k}{n} \leq \inf_{p \in [1,+\infty]}\frac{\log(\|\la(\cG)\|_p)}{\log(\|\la(\F)\|_p)},\end{equation*}
which ends the proof.
\end{proof}

\section{An application of approximate \texorpdfstring{$C^*$-algebras}{C*-algebras}}\label{sec:approximate_algebras}

In a recent article, Kitaev introduced and studied approximate $C^*$-algebras and almost idempotent quantum channels, in order to find an algebraic framework which allows one to naturally express the effects of noise on quantum information~\cite{kitaev25}. As most of our proofs in Section~\ref{sec:emul_capa} exploit the decomposition properties of finite-dimensional $*$-algebras, we use approximate $C^*$-algebras to strengthen the converse of item~(i) of Theorem~\ref{thm:main_UA_SR}, for the interconversion of reduced idempotent channels, and show that it holds even if a small error is allowed.

\begin{theorem}\label{thm:dimension_dependent_error}
    Let $\F : \cL(\cH_\F) \rightarrow \cL(\cH_\F)$ and $\cG : \cL(\cH_\cG) \rightarrow \cL(\cH_\cG)$ be two idempotent channels, which we assume to be reduced, that is $\F = \widehat{\F}$ and $\cG = \widehat{\cG}$. 
    Suppose $\la(\F) \neq (1)$.
    Let $k,n \in \N_{>0}$ and let $\E : \cL(\cH_\F^{\otimes k}) \rightarrow \cL(\cH_\cG^{\otimes n})$, $\D : \cL(\cH_\cG^{\otimes n}) \rightarrow \cL(\cH_\F^{\otimes k})$ be two channels. 
    Let $d = \dim(\cH_\F)$, $\delta_{\max}$ be the absolute constant\footnote{Here, `absolute' means that $\delta_{\max}$ depends neither on the dimension of $\cH_\F$ nor on that of $\cH_\cG$.} of Theorem~\ref{thm:kitaev} and let $\la_{\min, n}$ be the minimal non-zero eigenvalue of $\cG^{\otimes n}\E(\mathds{1}_{\F^{\otimes k}})$. If 
    \begin{equation*}\|\F^{\otimes k} - \D\cG^{\otimes n}\E\|_\diamond \leq \frac{\delta_{\max}\lambda_{\min, n}}{6d^k},\end{equation*}
    then
    \begin{equation*}\frac{k}{n} \leq \inf_{p \in [1,+\infty]} \frac{\log(\|\lambda(\cG)\|_p)}{\log(\|\lambda(\F)\|_p)}.\end{equation*}
\end{theorem}

Note that in Theorem~\ref{thm:dimension_dependent_error}, the error allowed depends on the dimension of the input space of $\F$, on the number of copies of $\F$ obtained, on the number of copies of $\cG$ used, as well as on the encoder used for the emulation. 
This theorem is therefore less operationally motivated than items~(i) and~(iii) of Theorem~\ref{thm:main_UA_SR}, but it still strengthens the converse of item~(i), and its proof leads to interesting mathematical developments, especially on multiplicative domains (see  Theorem~\ref{thm:approximate_multiplicative_domain}).

An interesting direction would be to extend item~(i) of Theorem~\ref{thm:main_UA_SR} to the problem of emulating almost idempotent channels with one another, but we leave this for future work.

This section is organised as follows. In Subsection~\ref{subsec:approx_*_algebras}, we recall, for completeness, the \emph{Error Reduction Theorem} for approximate $C^*$-algebras. We then move on to proving Theorem~\ref{thm:approximate_multiplicative_domain} on the approximate multiplicative domain of maps satisfying the Kadison--Schwarz inequality in Subsection~\ref{subsec:approx_multiplicative_domain}, which will be used in the proof of Theorem~\ref{thm:dimension_dependent_error} in Subsection~\ref{subsec:dimension_dependent_error}.

Note that this section assumes familiarity with basic concepts from operator algebras.
Note also that, in this section, we measure the quality of an emulation with the diamond norm itself, and not with the normalised diamond norm $\frac{1}{2}\|\cdot\|_\diamond$ used in the rest of the article. Also, throughout this section, the idempotent channels $\F$ and $\cG$ are assumed equal to their reduced channels.

\subsection{The Error Reduction Theorem}\label{subsec:approx_*_algebras}

We begin by recalling the definitions of $\eps$-$C^*$-algebras and of $\delta$-homomorphisms.

\begin{definition}[$\eps$-$C^*$-algebra, {\cite[Definition~2.1]{kitaev25}}]
    An \emph{$\eps$-Banach} algebra is a Banach space $\A$, endowed with a bilinear multiplication map $\A \times \A \rightarrow \A$ such that
    \begin{equation*}
    \begin{aligned}
        &\forall x,y \in \A,~~\|xy\| \leq (1+\eps)\|x\|\|y\|, \\
        &\forall x,y,z \in \A,~~\|(xy)z - x(yz)\| \leq \eps\|x\|\|y\|\|z\|.
    \end{aligned}
    \end{equation*}
    A $*\eps$-\emph{Banach algebra} is a complex $\eps$-Banach algebra with a conjugate linear involution $x \mapsto x^*$ satisfying the equations:
    \begin{equation}
        \forall x,y \in \A,~\|x^*\| = \|x\|,~~(xy)^* = y^* x^*.
    \end{equation}
    An $\eps$-$C^*$-algebra is one satisfying the following property:
    \begin{equation}
        \forall x \in \A,~\|x^* x\| \geq (1-\eps)\|x\|^2.
    \end{equation}
    The unit element $\mathds{1} \in \A$ should satisfy the following conditions:
    \begin{equation}
    \|x\mathds{1} - x\| \leq \eps\|x\|,~\|\mathds{1}x - x\| \leq \eps \|x\|,~| \|\mathds{1}\| - 1| \leq \eps. 
    \end{equation}
    If $\A$ is involutive, we also have
    \begin{equation}
    \mathds{1}^* = \mathds{1}.
    \end{equation}
\end{definition}

A $\delta$-homomorphism is then defined in the following way.

\begin{definition}[$\delta$-homomorphism, {\cite[Definition~2.2]{kitaev25}}]
     A $\delta$-\emph{homomorphism} from an $\eps'$-Banach algebra $\A'$ to an $\eps''$-Banach algebra $\A''$ is a bounded linear map $v: \A' \rightarrow \A''$ that almost preserves the unit and the multiplication:
    \begin{equation}\label{eq:def_delta_homo}
    \begin{aligned}
    &\|v(\mathds{1}) - \mathds{1}\| \leq \delta, \\
    &\forall x,y \in \A',~\|v(xy) - v(x)v(y)\| \leq \delta \|x\|\|y\|.
    \end{aligned}
    \end{equation}
    A \emph{non-unital} $\delta$-homomorphism is defined by imposing only the second condition. In the $*$-algebra setting, it is also required that $v(x^*) = v(x)^*$.

    A $\delta$-\emph{inclusion} is a $\delta$-homomorphism such that 
    \begin{equation*}
    \forall x \in \A',~(1-\delta)\|x\| \leq \|v(x)\| \leq (1+\delta)\|x\|.
    \end{equation*}
\end{definition}

The Error Reduction Theorem then states that if there is a $\delta$-inclusion of an exact $C^*$-algebra $\A$ into an $\eps$-$C^*$-algebra $\B$ for $\delta \leq \delta_{\max}$ and $\eps < \eps_{\max}$, which are constants that do not depend on the dimensions of $\A$ and $\B$, one can `correct' this $\delta$-inclusion into a less noisy approximate inclusion.

\begin{theorem}[Error Reduction Theorem, {\cite[Corollary~8.3]{kitaev25}}]\label{thm:kitaev}
There exist positive constants $\eps_{\max}$, $\delta_{\max}$ and $c_0$ such that for all $\eps < \eps_{max}$, if a finite-dimensional $C^*$-algebra $\A$ is $\delta_{max}$-included into an $\eps$-$C^*$-algebra $\B$, then there is also a $(c_0\eps)$-inclusion. If the original inclusion is bijective, then so is the new inclusion.    
\end{theorem}

Note that as every $C^*$-algebra is a $0$-$C^*$-algebra, Theorem~\ref{thm:kitaev} states that when we have a $\delta_{\max}$-inclusion between two $C^*$-algebras, it is possible to lift it to an \emph{exact} inclusion. This can be interpreted as a case of \emph{complete correction}, in the sense that the reduction of the error parameter $\delta$ is total.

\subsection{Approximate Multiplicative Domains}\label{subsec:approx_multiplicative_domain}

We used in the proof of the converse of item~(i) of Theorem~\ref{thm:main_UA_SR} the fact that, for a map $\F^*: \cL(\cH') \rightarrow \cL(\cH)$, which satisfies the Kadison--Schwarz inequality, its multiplicative domain can be expressed as

\begin{equation*}\{x \in \cL(\cH')\mid \F^*(x^*x) = \F^*(x^*)\F^*(x) \text{ and } \F^*(xx^*) = \F^*(x)\F^*(x^*)\}.\end{equation*}

In order to prove Theorem~\ref{thm:dimension_dependent_error}, we have to extend this result to the approximate case. For $\F^*: \cL(\cH') \rightarrow \cL(\cH)$ a linear map satisfying the Kadison--Schwarz inequality, we define, for $\delta \geq 0$ the following sets:

 \begin{equation*}
 \begin{aligned}
 \mathcal{C}_\delta^R &= \{x \in \cL(\cH')\mid \|\F^*(x^* x) - \F^*(x^*)\F^*(x)\| \leq \delta\|x\|^2\}, \\
 \mathcal{C}_{\delta}^L &= \{x \in \cL(\cH') \mid
 \|\F^*(xx^*) - \F^*(x)\F^*(x^*)\| \leq \delta\|x\|^2\}, \\ 
 \mathcal{C}_\delta &= \mathcal{C}_\delta^R \cap \mathcal{C}_\delta^L.
 \end{aligned}
 \end{equation*} 

We can show the following theorem.

\begin{theorem}\label{thm:approximate_multiplicative_domain}
Let $\F^*: \cL(\cH') \rightarrow \cL(\cH)$ be a Hermitian-preserving linear map satisfying the Kadison--Schwarz inequality. Then, for all $\delta > 0$, $\F^*$ is $2\delta$-multiplicative on $\mathcal{C}_\delta$, i.e.

\begin{equation}
\forall x,y \in \mathcal{C}_\delta, \|\F^*(xy) - \F^*(x)\F^*(y)\| \leq 2\delta\|x\|\|y\|.
\end{equation}
More precisely, we have that

\begin{equation}
\forall x \in \mathcal{C}^L_\delta, \forall y \in \mathcal{C}^R_\delta, \|\F^*(xy) - \F^*(x)\F^*(y)\| \leq 2\delta\|x\|\|y\|.
\end{equation}
\end{theorem}

\begin{proof} [Proof of Theorem~\ref{thm:approximate_multiplicative_domain}]
    This proof generalises the ones of~\cite[Theorems~5.4 and~5.7]{Wolf12}.
    In this proof, we let $\mathds{1}$ denote the identity operator on $\cH$.

    We have that $x \in \mathcal{C}^L_\delta$ if and only if $x^* \in \mathcal{C}^R_\delta$ and vice versa. Therefore, it suffices to show that for all $x,y \in \mathcal{C}^R_\delta$, 
    \begin{equation}
    \|\F^*(x^* y) - \F^*(x^*)\F^*(y)\| \leq 2\delta \|x\|\|y\|.
    \end{equation}
    Let $x,y\in \mathcal{C}^R_\delta$, $t \in \R$ and $z = tx + y$. As $\F^*$ satisfies the Kadison--Schwarz inequality,
    \begin{equation*}0 \leq \F^*(z^* z) - \F^*(z^*)\F^*(z).\end{equation*}
    If we develop the right-hand side, we get:
    \begin{equation*}
    0 \leq t^2 Q + tV + W,
    \end{equation*}
    with
    \begin{equation*}
    \begin{aligned}
    Q &= \F^*(x^* x) - \F^*(x^*)\F^*(x), \\
    W &= \F^*(y^* y) - \F^*(y^*)\F^*(y), \\
    v &= \F^*(x^* y) - \F^*(x^*)\F^*(y), \\
    V &= v + v^*.
    \end{aligned}
    \end{equation*}

    Then, as $\F^*$ satisfies the Kadison--Schwarz inequality, $0 \leq Q$ and $0 \leq W$. Furthermore, $Q \leq \|Q\|\mathds{1} \leq \delta\|x\|^2\mathds{1}$ and $W \leq \|W\|\mathds{1} \leq \delta\|y\|^2\mathds{1}$ by assumption. Since for all $t \in \R$, $t^2 \geq 0$, we have

    \begin{equation*}0 \leq t^2Q + tV+W \leq t^2\delta\|x\|^2\mathds{1} + t V + \delta \|y\|^2\mathds{1}.\end{equation*}
    As $V = v + v^*$ is Hermitian, we have $V = UDU^*$ with $U$ a unitary matrix and $D$ a diagonal matrix with real coefficients. Conjugating  the last inequality by $U$, we get

    \begin{equation*}0 \leq t^2\delta\|x\|^2\mathds{1} + tD + \delta \|y\|^2\mathds{1}.\end{equation*}

    The right-hand side of this equation is simply a diagonal operator, whose diagonal coefficients are quadratic polynomials in $t$. As this equation gives us that this operator is positive semidefinite for all $t \in \R$, each of these quadratic polynomials is always non-negative, and can thus have at most a single root. Therefore, if we write $\{D_i\}_i$ for the diagonal entries of $D$, the discriminant of the polynomial
    \begin{equation*}\delta \|x\|^2 t^2 + D_i t + \delta \|y\|^2,\end{equation*}
    has to be non-positive for all indices $i$, that is
    \begin{equation*}D_i^2 \leq 4\delta^2 \|x\|^2 \|y\|^2.\end{equation*}
    Thus,
    \begin{equation*}|D_i| \leq 2 \delta \|x\|\|y\|.\end{equation*}
    As the Schatten $\infty$-norm of $V = UDU^*$ is $\|V\| = \max_i |D_i|$, we have:
    \begin{equation*} \|v + v^*\| = \|V\| \leq 2\delta\|x\|\|y\|.\end{equation*}
    We can apply the same line of argument replacing $y$ with $iy$ to find 
    \begin{equation}\label{eq:inequ_for_iX}
    \|iv -iv^*\| \leq 2\delta\|x\|\|y\|.
    \end{equation}
    We can decompose $v = h_1 + ih_2$ with $h_1$ and $h_2$ Hermitian operators\footnote{Taking $h_1 = \frac{1}{2}(v+v^*)$, $h_2 = \frac{1}{2i}(v - v^*)$.}, to find:
    \begin{equation*}\|v + v^*\| = 2\|h_1\| \leq 2\delta\|x\|\|y\|,\end{equation*}
    \begin{equation*}\|iv - iv^*\| = 2\|h_2\| \leq 2\delta\|x\|\|y\|.\end{equation*}
    Therefore,
    \begin{equation}\label{eq:inequality_v_2delta}
    \|v\| \leq \|h_1\| + \|h_2\| \leq 2 \delta\|x\|\|y\|,
    \end{equation}
    which is what we claimed.  
\end{proof}

\subsection{Proof of Theorem~\ref{thm:dimension_dependent_error}}\label{subsec:dimension_dependent_error}

In this subsection, we finally prove Theorem~\ref{thm:dimension_dependent_error}, in the same way as we proved the converse of item~(i) of Theorem~\ref{thm:main_UA_SR}, that is, by first proving the one-shot case and then using the compatibility of the shape vector and of the $\ell_p$-norms with the tensor product to conclude. 

\begin{prop}\label{prop:one-shot_dimension_dependent_error}
    Let $\F : \cL(\cH_\F) \rightarrow \cL(\cH_\F)$ and $\cG: \cL(\cH_\cG) \rightarrow \cL(\cH_\cG)$ be two idempotent channels equal to their reduced channels, i.e.\ $\F = \widehat{\F}$ and $\cG = \widehat{\cG}$, such that there exist two channels $\E$ and $\D$ satisfying \begin{equation}\label{eq:one_shot_dimeension_dependent_error}
    \|\F - \D\cG\E\|_\diamond \leq \frac{\delta_{\max}\la_{\min,1}}{6d},
    \end{equation}
    with $d = \dim(\cH_\F)$, $\delta_{\max}$ is the constant of Theorem~\ref{thm:kitaev} and $\la_{\min,1}$ is the minimal non-zero eigenvalue of $\cG\E(\mathds{1}_\F)$. 
    We write $e_{\cG\E} = \cG\E(\mathds{1}_\F)^0$ for the orthogonal projector onto the support $\supp(\cG\E(\mathds{1}_\F))$, and we define
    \begin{equation*}
    \widetilde{\cG}^* : x \mapsto e_{\cG\E}\cG^*(x)e_{\cG\E}.
    \end{equation*}

    Then, $\widetilde{\cG}^*\D^*$ is a $\delta_{\max}$-inclusion from the $*$-algebra $\Rg(\F^*)$ into the $*$-algebra $e_{\cG\E}\Rg(\cG^*)e_{\cG\E}$. 
    By the Error Reduction Theorem~\ref{thm:kitaev}, there is therefore an injective unital $*$-homomorphism embedding $\Rg(\F^*)$ into $e_{\cG\E}\Rg(\cG^*)e_{\cG\E}$,
    thus, for all $p \in [1,+\infty]$,
    \begin{equation*}
    \|\la(\F)\|_p \leq \|\la(\cG)\|_p.
    \end{equation*}
\end{prop}

We split the proof that the map $\widetilde{\cG}^*\D^*$ restricted to $\Rg(\F^*)$ satisfies all the properties of a $\delta_{\max}$-inclusion into $e_{\cG\E}\Rg(\cG^*)e_{\cG\E}$ into several lemmas, each of them devoted to one of these properties.

We begin with the one related to the almost preservation of the Schatten $\infty$-norm. Its proof uses the \emph{completely bounded operator norm}, which is defined for all linear maps $\F : \cL(\cH) \rightarrow \cL(\cH')$ as

\begin{equation*}\|\F\|_{\mathrm{cb}} = \sup_{n \in \N_{>0}} \|\F\otimes \Id_n\|_{\infty \rightarrow \infty},\end{equation*}
with $\|\cdot\|_{\infty \rightarrow \infty}$ the `infinity-to-infinity' norm, i.e.\ the map norm of linear maps from $\cL(\cH)$ to $\cL(\cH')$ when the input and output spaces are endowed with the Schatten $\infty$-norms. That is:
\begin{equation*}
\|\F\|_{\infty \to \infty} = \sup_{\|x\|\leq 1}\|\F(x)\|.
\end{equation*}
Note that for any linear map $\F : \cL(\cH) \rightarrow \cL(\cH')$, we have
\begin{equation*}
\|\F\|_\diamond = \|\F^*\|_{\mathrm{cb}}.
\end{equation*}

\begin{lemma}[Almost norm preservation]\label{lemma:norm_preservation}
Let $\F$, $\cG$, $\widetilde{\cG}$, $\E$ and $\D$ be defined as in the statement of Proposition~\ref{prop:one-shot_dimension_dependent_error} and let 
\begin{equation*}
\delta \geq \|\F - \D\cG\E\|_\diamond = \|\F^* - \E^*\cG^*\D^*\|_{\mathrm{cb}}.
\end{equation*}
Then, for all $x \in \Rg(\F^*)$, 
\begin{equation*}
(1 - \delta)\|x\| \leq \|\widetilde{\cG}^*\D^*(x)\| \leq \|x\|.
\end{equation*}
\end{lemma}

Note that this lemma implies that, whenever $\|\F-\D\cG\E\|_\diamond\leq\delta<1$, the restriction of $\widetilde{\cG}^{*}\D^{*}$ to $\Rg(\F^{*})$ is injective.
Indeed, if $x\in\Rg(\F^{*})$ and $\widetilde{\cG}^{*}\D^{*}(x)=0$, then $(1-\delta)\|x\|\leq0$, hence $x=0$.

\begin{proof}
The right-hand side inequality follows from the fact that $\widetilde{\cG}^*\D^* : \cL(\cH_\F) \rightarrow \cL(\cH_\cG)$ is a subunital completely positive map (when the identity considered in the output space of $\widetilde{\cG}^*\D^*$ is the identity of $\cL(\cH_\cG)$). For the left-hand side inequality, recall that, by Lemma~\ref{lemma:image_of_identity}, for all $x \in \cL(\cH_\F)$, 
\begin{equation*}\E^*\cG^*\D^*(x) = \E^*\cG^*\widetilde{\cG}^*\D^*(x).\end{equation*}
Then, we have for all $x \in \Rg(\F^*)$:
\begin{equation*}
\begin{aligned}
\|x\| - \|\widetilde{\cG}^*\D^*(x)\| &\overset{(a)}{\leq} \|x\| - \|\E^*\cG^*\widetilde{\cG}^*\D^*(x)\| \\
&\overset{(b)}{\leq} \|x - \E^*\cG^*\D^*(x)\| \\
&\overset{(c)}{=} \|\F^*(x) - \E^*\cG^*\D^*(x)\| \\
&\leq \|\F^* - \E^*\cG^*\D^*\|_{\mathrm{cb}}\|x\| \\
&\leq \delta\|x\|,
\end{aligned}
\end{equation*}
where $(a)$ follows from the fact that $\|\E^*\cG^*\widetilde{\cG}^*\D^*(x)\| \leq \|\widetilde{\cG}^*\D^*(x)\|$,
since $\E^*\cG^*$ is a unital completely positive map. Then, $(b)$ follows from the triangle inequality. Finally, $(c)$ follows from the fact that $\F^*(x) = x$ for all $x \in \Rg(\F^*)$.
The last inequality can be rewritten as 
\begin{equation*}(1-\delta)\|x\| \leq \|\widetilde{\cG}^*\D^*(x)\|.\end{equation*}
\end{proof}

We then move on to the proof that $\widetilde{\cG}^*\D^* : \Rg(\F^*) \rightarrow e_{\cG\E}\Rg(\cG^*)e_{\cG\E}$ is  unital, where we recall that the unit of the $*$-algebra $e_{\cG\E}\Rg(\cG^*)e_{\cG\E}$ is $e_{\cG\E}$, as, for every element $x = e_{\cG\E}x'e_{\cG\E} \in e_{\cG\E}\Rg(\cG^*)e_{\cG\E}$, $e_{\cG\E}x = xe_{\cG\E} = x$. Note also that for a general $\delta$-inclusion, only almost unitality is required.

\begin{lemma}[Exact unitality]\label{lemma:exact_unitality}
Let $\F$, $\cG$, $\widetilde{\cG}$, $\E$ and $\D$ be defined as in Proposition~\ref{prop:one-shot_dimension_dependent_error}. Then:
\begin{equation*}
\widetilde{\cG}^*\D^*(\mathds{1}_\F) = e_{\cG\E}.
\end{equation*}
\end{lemma}

\begin{proof}
As $\cG$ and $\D$ are assumed to be channels, their adjoints are unital maps, so that
\begin{equation*}
\cG^*\D^*(\mathds{1}_\F) = \mathds{1}_\cG.
\end{equation*}
Therefore 
\begin{equation*}\widetilde{\cG}^*\D^*(\mathds{1}_\F) = e_{\cG\E}\cG^*(\D^*(\mathds{1}_\F))e_{\cG\E} = e_{\cG\E}\mathds{1}_\cG e_{\cG\E} = e_{\cG\E}.\end{equation*}
\end{proof}

Finally, to prove the almost preservation of the multiplication on $\Rg(\F^*)$ by $\widetilde{\cG}^*\D^*$, we will need the following additional lemma.

\begin{lemma}\label{lemma:expansion}
Let $\Phi : \cL(\cH) \rightarrow \cL(\cH')$ be a channel, write $d = \dim(\cH)$, $e = \Phi(\mathds{1})^0$ and $\la_{\min,1}$ the minimal non-zero eigenvalue of $\Phi(\mathds{1})$. Then, for all positive semidefinite $x \in \cL(\cH')$,
\begin{equation}
\frac{\la_{\min,1}}{d}\|exe\| \leq \|\Phi^*(x)\|.
\end{equation}
\end{lemma}

\begin{proof}
Let $x\in \cL(\cH')$ be positive semidefinite, then
\begin{align*}
\|\Phi^*(x)\| &\geq \frac{1}{d}\tr(\Phi^*(x)) \\
&= \frac{1}{d}\tr(x\Phi(\mathds{1})) \\
&\geq \frac{\la_{\min,1}}{d}\tr(exe) \\
&\geq \frac{\la_{\min,1}}{d}\|exe\|.
\end{align*}
\end{proof}

Note that this lemma implies the second statement of Lemma~\ref{lemma:image_of_identity}. We can now move on to the proof of the almost preservation of the multiplication, which uses our Theorem~\ref{thm:approximate_multiplicative_domain} on the approximate multiplicative domain of maps satisfying the Kadison--Schwarz inequality.

\begin{lemma}[Almost preservation of the multiplication]\label{lemma:approx_multiplicativity}
    Let $\F$, $\cG$, $\widetilde{\cG}$, $\E$ and $\D$ be defined as in the statement of Proposition~\ref{prop:one-shot_dimension_dependent_error}, let 
    \begin{equation*}\delta \geq \|\F - \D\cG\E\|_\diamond,\end{equation*}
    let $\la_{\min,1}$ be the minimal non-zero eigenvalue of $\cG\E(\mathds{1}_\F)$, let $d = \dim(\cH_\F)$ and let $x, y \in \Rg(\F^*)$. We then have:
    \begin{equation*}
    \|\widetilde{\cG}^*\D^*(x y) - \widetilde{\cG}^*\D^*(x)\widetilde{\cG}^*\D^*(y)\| \leq \frac{6\delta d}{\la_{\min,1}}\|x\|\|y\|.
    \end{equation*}
\end{lemma}

\begin{proof}
    We prove that, for all $x \in \Rg(\F^*)$,
    \begin{equation*}
    \begin{aligned}
    \|\widetilde{\cG}^*\D^*(x^* x) - \widetilde{\cG}^*\D^*(x^*)\widetilde{\cG}^*\D^*(x)\| &\leq \frac{3 \delta d}{\la_{\min,1}}\|x\|^2,\\
    \|\widetilde{\cG}^*\D^*(x x^*) - \widetilde{\cG}^*\D^*(x)\widetilde{\cG}^*\D^*(x^*)\| &\leq \frac{3 \delta d}{\la_{\min,1}}\|x\|^2,
    \end{aligned}
    \end{equation*}
    and then use Theorem~\ref{thm:approximate_multiplicative_domain} to conclude. 
    First, to simplify notation, we let $\Phi^* = \E^*\cG^*\D^* = \E^*\cG^*\widetilde{\cG}^*\D^*$. We have:
    \begin{equation*}
    \begin{aligned}
    \|\Phi^*(x^* x) - \Phi^*(x^*)\Phi^*(x)\| &= \|\Phi^*(x^* x) - \Phi^*(x^*)\Phi^*(x) + x^* x - x^*x\| \\
    &\leq \|\Phi^*(x^* x) - x^* x\| + \|x^* x - \Phi^*(x^*)\Phi^*(x) + x^*\Phi^*(x) - x^*\Phi^*(x)\| \\
    &\overset{(a)}{\leq} \delta\|x\|^2 + \|x\| \|x - \Phi^*(x)\| + \|x^* - \Phi^*(x^*)\|\|\Phi^*(x)\| \\
    &\overset{(b)}{\leq} 3\delta\|x\|^2,
    \end{aligned}
    \end{equation*}
    where $(a)$ follows from the fact that as $x \in \Rg(\F^*)$, we have
    \begin{equation*}
    \|\Phi^*(x^*x) - x^*x\| = \|\Phi^*(x^*x) - \F^*(x^*x)\| \leq \|\Phi^* - \F^*\|_{\mathrm{cb}} \|x^*x\| \leq \delta \|x\|^2.
    \end{equation*}
    By the same argument, we also have that $\|x - \Phi^*(x)\| \leq \delta \|x\|$ and $\|x^* - \Phi^*(x^*)\| \leq \delta \|x\|$, hence $(b)$ follows.
    
    Then, applying the Kadison--Schwarz inequality, we get:
    \begin{equation*}
    \E^*\cG^*\widetilde{\cG}^*\D^*(x^* x) \geq \E^*\cG^*(\widetilde{\cG}^*\D^*(x^*)\widetilde{\cG}^*\D^*(x)) 
    \geq \E^*\cG^*\widetilde{\cG}^*\D^*(x^*)\E^*\cG^*\widetilde{\cG}^*\D^*(x).
    \end{equation*}
    Therefore:
    \begin{equation*}
    \|\E^*\cG^*(\widetilde{\cG}^*\D^*(x^* x) - \widetilde{\cG}^*\D^*(x^*)\widetilde{\cG}^*\D^*(x))\| 
    \leq  \|\Phi^*(x^* x) - \Phi^*(x^*)\Phi^*(x)\| 
    \leq 3\delta\|x\|^2.
    \end{equation*}
    Thus, by Lemma~\ref{lemma:expansion}, we obtain
    \begin{equation*}
    \|\widetilde{\cG}^*\D^*(x^* x) - \widetilde{\cG}^*\D^*(x^*)\widetilde{\cG}^*\D^*(x)\|
    \leq \frac{3d\delta}{\la_{\min,1}}\|x\|^2.
    \end{equation*}
    We can apply the exact same line of reasoning to prove the same inequality but with the roles of $x$ and $x^*$ swapped.
\end{proof}

Gathering all these lemmas, we obtain the proof of Proposition~\ref{prop:one-shot_dimension_dependent_error}.

\begin{proof}[Proof of Proposition~\ref{prop:one-shot_dimension_dependent_error}]
    Let $\delta \geq \|\F - \D\cG\E\|_\diamond$. Since we assumed that $\F$ and $\cG$ are equal to their reduced channels, i.e.\ $\F = \widehat{\F}$ and $\cG = \widehat{\cG}$, $\Rg(\F^*)$ is a $*$-subalgebra of $\cL(\cH_\F)$ and $e_{\cG\E}\Rg(\cG^*)e_{\cG\E}$ is a $*$-subalgebra of $\cL(\cH_\cG)$.
    By Lemmas~\ref{lemma:norm_preservation},~\ref{lemma:exact_unitality} and~\ref{lemma:approx_multiplicativity}, $\widetilde{\cG}^*\D^*$ is a $\max\{\delta, \frac{6d\delta}{\la_{\min,1}}\}$-inclusion (we trivially have that $\widetilde{\cG}^*\D^*$ is Hermitian preserving). However, as 
    \begin{equation*}d = \tr(\mathds{1}_\F) = \tr(\cG\E(\mathds{1}_\F)) \geq \la_{\min,1},\end{equation*}
$\frac{d}{\la_{\min,1}} \geq 1$ so that $\delta \leq \frac{6d}{\la_{\min,1}}\delta$. Finally, $\widetilde{\cG}^*\D^*$ is a $\frac{6d\delta}{\la_{\min,1}}$-inclusion of $\Rg(\F^*)$ into $e_{\cG\E}\Rg(\cG^*)e_{\cG\E}$.

Then, if we suppose that 
\begin{equation*}\frac{6d\delta}{\la_{\min,1}} \leq \delta_{\max},
\end{equation*}
i.e.
\begin{equation*}
\|\F - \D\cG\E\|_\diamond \leq \delta \leq \frac{\delta_{\max}\la_{\min,1}}{6d},
\end{equation*}
by the Error Reduction Theorem (Theorem~\ref{thm:kitaev}), there is an injective $*$-homomorphism from $\Rg(\F^*)$ into $e_{\cG\E}\Rg(\cG^*)e_{\cG\E}$, thus by Kuperberg's theorem (Theorem~\ref{thm:converse_kuperberg}) and the second part of Proposition~\ref{prop:encoding_implies_morphism}, for all $p \in [1,+\infty]$,
\begin{equation*}
\|\la(\F)\|_p \leq \|\la(e_{\cG\E}\Rg(\cG^*)e_{\cG\E})\|_p \leq \|\la(\cG)\|_p.
\end{equation*}
\end{proof}

We can finally prove Theorem~\ref{thm:dimension_dependent_error}.

\begin{proof}[Proof of Theorem~\ref{thm:dimension_dependent_error}]

We simply apply Proposition~\ref{prop:one-shot_dimension_dependent_error} to $\F^{\otimes k}$ and $\cG^{\otimes n}$.
Note that the dimension of $\cH_\F^{\otimes k}$ is $\dim(\cH_\F)^k = d^k$, with $d = \dim(\cH_\F)$, so that the hypothesis of Theorem~\ref{thm:dimension_dependent_error} is sufficient.
Thus, by Proposition~\ref{prop:one-shot_dimension_dependent_error}, if the hypotheses of Theorem~\ref{thm:dimension_dependent_error} are met, for all $p \in [1,+\infty]$, we have:
\begin{equation*}
\|\la(\F^{\otimes k})\|_p \leq\|\la(\cG^{\otimes n})\|_p.
\end{equation*}
If $\F$ is classical, i.e.\ $\lambda(\F) = (1,\dots,1)$, with $\lambda(\F) \neq (1)$ by assumption, then the infimum of the ratios $\log(\|\lambda(\cG)\|_p)/\log(\|\lambda(\F)\|_p)$ for $p \in [1,+\infty]$ is reached at $p=1$, see Table~\ref{tab:lower_bounds}.
Since $\|\lambda(\F^{\otimes k})\|_1 \leq \|\lambda(\cG^{\otimes n})\|_1$, we have 
\begin{equation*}
    \frac{k}{n} \leq \frac{\log(\|\lambda(\cG)\|_1)}{\log(\|\lambda(\F)\|_1)} = \inf_{p \in [1,+\infty]}\frac{\log(\|\la(\cG)\|_p)}{\log(\|\la(\F)\|_p)} = \CUA{\cG}{\F}.
\end{equation*}
If $\F$ is not classical, then for all $p \in [1,+\infty]$, $\log(\|\la(\F)\|_p) > 0$.
Thus, by the compatibility of the $\ell_p$-norms of the shape vector with the tensor product, we have
\begin{equation*}
\frac{k}{n} \leq \inf_{p \in [1,+\infty]}\frac{\log(\|\la(\cG)\|_p)}{\log(\|\la(\F)\|_p)} = \CUA{\cG}{\F},
\end{equation*}
where the right-hand side equality is the first statement in item~(i) of Theorem~\ref{thm:main_UA_SR}.
    
\end{proof}

\end{document}